\documentclass[11pt]{article}

\usepackage[pagebackref,colorlinks]{hyperref}
\usepackage{amsmath} 
\usepackage{amsthm} 
\usepackage{amssymb}	
\usepackage{graphicx} 
\usepackage{multicol} 
\usepackage{multirow}
\usepackage{color}
\usepackage[dvips,letterpaper,margin=1in,bottom=1in]{geometry}
\usepackage[capitalize,noabbrev]{cleveref}

\usepackage[utf8]{inputenc}
\usepackage[english]{babel}

\usepackage{mathtools}
\usepackage[shortlabels]{enumitem}

\newtheorem{theorem}{Theorem}[section]

\newtheorem{lemma}[theorem]{Lemma}
\newtheorem{corollary}[theorem]{Corollary}
\newtheorem{proposition}[theorem]{Proposition}
\newtheorem{fact}[theorem]{Fact}

\newtheorem{definition}[theorem]{Definition}

\newcommand{\braket}[2]{\left< #1 \vphantom{#2} \middle| #2 \vphantom{#1} \right>} 
\newcommand{\ketbra}[2]{\ensuremath{\ket{#1}\!\bra{#2}}}

\DeclarePairedDelimiter\rbra{\lparen}{\rparen}
\DeclarePairedDelimiter\sbra{\lbrack}{\rbrack}
\DeclarePairedDelimiter\cbra{\{}{\}}
\DeclarePairedDelimiter\abs{\lvert}{\rvert}
\DeclarePairedDelimiter\Abs{\lVert}{\rVert}
\DeclarePairedDelimiter\ceil{\lceil}{\rceil}
\DeclarePairedDelimiter\floor{\lfloor}{\rfloor}
\DeclarePairedDelimiter\ket{\lvert}{\rangle}
\DeclarePairedDelimiter\bra{\langle}{\rvert}
\DeclarePairedDelimiter\ave{\langle}{\rangle}

\DeclareMathOperator*{\E}{\mathbb{E}}
\DeclareMathOperator*{\Var}{\mathbf{Var}}

\newcommand{\tr} {\operatorname{tr}}
\newcommand{\poly} {\operatorname{poly}}
\newcommand{\diag} {\operatorname{diag}}

\newcommand{\spanspace} {\operatorname{span}}

\newcommand{\Real} {\operatorname{Re}}

\usepackage{algorithm}
\usepackage{algpseudocode}

\usepackage{tabularx}
\usepackage{booktabs}
\usepackage{threeparttable}

\usepackage{adjustbox}

\usepackage{tikz}
\usetikzlibrary{quantikz2}

\allowdisplaybreaks

\begin{document}

\title{Towards Minimax Estimation of High-Order Functionals \\ by Quantum Arguments}

\author{Qisheng Wang\thanks{Qisheng Wang is with the School of Computer Science, Shanghai Jiao Tong University, Shanghai 200240, China (e-mail: \href{mailto:QishengWang1994@gmail.com}{\nolinkurl{QishengWang1994@gmail.com}}).}}
\date{}

\maketitle
\pagenumbering{roman}
\thispagestyle{empty}

\begin{abstract}
    We propose a novel approach to the minimax estimation of high-order functionals from the perspective of quantum computing. 
    Specifically, for any real number $\alpha \gg 1$, we present two estimators, one for the classical functional $\mathrm{F}_\alpha(P) = \sum_{i=1}^S p_i^\alpha$ of a discrete distribution $P$ and the other for the quantum functional $\mathrm{F}_\alpha(\rho) = \operatorname{tr}(\rho^\alpha)$ of a mixed state $\rho$. 
    These functionals have close connections with the R\'enyi entropy and the Tsallis entropy. 
    We show that both estimators achieve the minimax optimal $L_2$ rate $\alpha \mathsf{n}^{-1}$ in the range $\alpha \lesssim \mathsf{n} \lesssim \alpha^{3-o(1)}$, where the support size $S$ of $P$ or the dimension of $\rho$ can be much larger than the number of samples $\mathsf{n}$. 
    As a result, both estimators achieve the \textit{optimal} sample complexity $\mathsf{n} \asymp \alpha$, improving upon the prior best upper bounds $O(\alpha^2)$ established by \hyperlink{cite.JVHW17}{Jiao, Venkat, Han, and Weissman (\textit{IEEE Trans.\ Inf.\ Theory} 2017)} for classical functionals and \hyperlink{cite.CW25}{Chen and Wang (COLT 2025)} for quantum functionals. 
    Our estimators are constructed under a unified framework using quantum primitives and run in linear time on a quantum computer. 
    This work reveals an unexpected path from quantum computing to statistics, suggesting a conceptually new methodology for functional estimation.
    It adds to the growing list of quantum proofs for classical theorems.
\end{abstract}

\newpage
\tableofcontents
\thispagestyle{empty}
\newpage
\pagenumbering{arabic}

\section{Introduction}

\paragraph{Classical functional estimation.}
Given $\mathsf{n}$ samples drawn from an unknown discrete probability distribution $P = \rbra{p_1, p_2, \dots, p_S}$ of alphabet size $S$, the estimation of the functionals of the distribution $P$ of the form
\begin{equation} \label{eq:def-FP}
    F\rbra{P} = \sum_{i=1}^S f\rbra{p_i}
\end{equation}
has been extensively investigated in the literature. 
This fundamental problem has strong applications in entropy estimation. 
For example, when $f\rbra{x} = -x\ln\rbra{x}$, the functional becomes the Shannon entropy $\mathrm{H}\rbra{P} = - \sum_{i=1}^S p_i \ln\rbra{p_i}$ \cite{Sha48a,Sha48b}, and its estimation has been thoroughly studied in a series of works \cite{Pan03,BDKR05,Pan04,VV11a,VV11b,VV17,JVHW15,JVHW17,WY16}.

In particular, $f\rbra{x} = x^\alpha$ with parameter $\alpha$ gives a family of information measures of the form
\begin{equation}
\mathrm{F}_\alpha\rbra{P} \coloneqq \sum_{i=1}^S p_i^\alpha.
\end{equation}
This type of information measure is involved in a wide range of research areas, e.g., the Gini impurity $\mathrm{I}_{\textup{Gini}}\rbra{P} = 1 - \mathrm{F}_2\rbra{P}$ \cite{Gin12} in machine learning \cite{BFOS84}, the R\'enyi entropy $\mathrm{H}_\alpha^{\textup{R\'en}}\rbra{P} = \frac{\ln\rbra{\mathrm{F}_\alpha\rbra{P}}}{1-\alpha}$ \cite{Ren61} and the Tsallis entropy $\mathrm{H}_\alpha^{\textup{Tsa}}\rbra{P} = \frac{\mathrm{F}_\alpha\rbra{P}-1}{1-\alpha}$ \cite{Tsa88} in information theory, the Hill number $\prescript{\alpha}{}{\mathrm{D}}\rbra{P} = \rbra{\mathrm{F}_\alpha\rbra{P}}^{1/\rbra{1-\alpha}}$ \cite{Hil73} in ecology, and the frequency moment $\mathrm{M}_{\alpha}\rbra{P} = S^\alpha \mathrm{F}_\alpha\rbra{P}$ \cite{AMS99} in computational complexity theory \cite{BYKS01,IW05,BGKS06,GC07}.
The estimation of the high-order functional $\mathrm{F}_\alpha\rbra{P}$ has been investigated in the literature \cite{AK01,CJ15,JVHW15,JVHW17,AOST17,OS17}.
The prior best estimator for $\mathrm{F}_\alpha\rbra{P}$ is due to \cite{JVHW15,JVHW17}, which can return an estimate to within any constant additive error with sample complexity $\mathsf{n} = O\rbra{\alpha^2}$ for any real number $\alpha \gg 1$.
Moreover, the estimator in \cite{JVHW17} achieves a minimax mean squared error (MSE) rate of $\alpha^2 \mathsf{n}^{-1}$ when $\mathsf{n} \gtrsim \alpha^2$. 
A direct question naturally arises:
\begin{equation} \label{eq:question}
    \textit{Can we estimate $\mathrm{F}_\alpha\rbra{P}$ using $\mathsf{n} \ll \alpha^2$ samples?} \tag{$*$}
\end{equation}
This fundamental question is at the core of the estimation of functionals of discrete distributions. 
A positive answer to this question will immediately lead to better estimators for a series of information-theoretic quantities such as the Tsallis entropy \cite{Tsa88} and the mutual Tsallis entropy \cite{Fur06}. 
In particular, the notably high-order functional $\mathrm{F}_\alpha\rbra{P}$, its quantum analog $\mathrm{F}_\alpha\rbra{\rho}$, and their variants have been employed in practice with $\alpha \sim 10^2$, e.g., investigating the properties of Rydberg hydrogenic systems \cite{TD16,TPCD16} and antiferromagnetic Heisenberg models \cite{WD20}. 

\paragraph{Quantum functional estimation.}
The estimation of functionals of quantum states, a fundamental task in quantum property testing \cite{MdW16}, has recently attracted a lot of attention in the literature.
In the quantum setting, $\mathsf{n}$ samples of an unknown quantum state $\rho$ (i.e., $\rho^{\otimes \mathsf{n}}$) are given and the goal is to estimate the functionals of $\rho$ of the form 
\begin{equation} \label{eq:def-Frho}
    F\rbra{\rho} = \tr\rbra{f\rbra{\rho}},
\end{equation}
which generalizes the classical functional. 
For example, when $f\rbra{x} = -x\ln\rbra{x}$, the functional becomes the von Neumann entropy $\mathrm{S}\rbra{\rho} = -\tr\rbra{\rho\ln\rbra{\rho}}$ \cite{vN27}, and its estimation has been studied in a series of works \cite{BMW16,AISW20,GL20,CLW20,GHS21,WZW23,WGL+24,WZYW23,WZ25}.

In particular, $f\rbra{x} = x^\alpha$ gives the high-order functional 
\begin{equation}
    \mathrm{F}_\alpha\rbra{\rho} = \tr\rbra{\rho^{\alpha}},
\end{equation}
which generalizes the classical functional $\mathrm{F}_\alpha\rbra{P}$. 
This information quantity has close connections to the quantum R\'enyi entropy $\mathrm{S}_\alpha^{\textup{R\'en}}\rbra{\rho} = \frac{\ln\rbra{\mathrm{F}_\alpha\rbra{\rho}}}{1-\alpha}$ and the quantum Tsallis entropy $\mathrm{S}_\alpha^{\textup{Tsa}}\rbra{\rho} = \frac{\mathrm{F}_\alpha\rbra{\rho} - 1}{1-\alpha}$. 
The estimation of $\mathrm{F}_\alpha\rbra{\rho}$ has been investigated in the literature \cite{AISW20,SH21,WZW23,WGL+24,WZYW23,WZ25,LW25,CW25,Wan25}.
In particular, in \cite{LW25}, an estimator for $\mathrm{F}_\alpha\rbra{\rho}$ was proposed, which can estimate $\mathrm{F}_\alpha\rbra{\rho}$ to within any constant additive error with sample complexity $\mathsf{n} = \exp\rbra{O\rbra{\alpha}}$ for any real number $\alpha \gg 1$; later, in \cite{CW25}, the sample complexity was improved to $\mathsf{n} = O\rbra{\alpha^2}$, which matches the prior best sample complexity for estimating the classical functional $\mathrm{F}_\alpha\rbra{P}$ given in \cite{JVHW17}. 
As a quantum analog of Question~\eqref{eq:question}, again we ask:
\begin{equation} \label{eq:question-q}
    \textit{Can we estimate $\mathrm{F}_\alpha\rbra{\rho}$ using $\mathsf{n} \ll \alpha^2$ samples?} \tag{$**$}
\end{equation}
This question is not only of independent interest in quantum computing, but also strictly generalizes the classical Question~\eqref{eq:question}. Any improvement in Question~\eqref{eq:question-q} also implies an improvement in Question~\eqref{eq:question}.

\paragraph{Our contributions.}

In this paper, we give a positive answer to both Questions~\eqref{eq:question} and \eqref{eq:question-q}.
\begin{theorem}[Optimal estimator for high-order functionals] \label{thm:main-intro}
    For any real number $\alpha \gg 1$, the sample complexity of estimating $\mathrm{F}_\alpha\rbra{P}$ and $\mathrm{F}_\alpha\rbra{\rho}$ to within any constant additive error is $\mathsf{n} \asymp \alpha$. 
\end{theorem}

\cref{thm:main-intro} improves both the prior best classical upper bound $O\rbra{\alpha^2}$ due to \cite{JVHW17} and the prior best quantum upper bound $O\rbra{\alpha^2}$ due to \cite{CW25}. 
Moreover, our estimator achieves the minimax \textit{optimal} MSE rate $\alpha \mathsf{n}^{-1}$ when $\alpha \lesssim \mathsf{n} \lesssim \alpha^{3-o(1)}$, in contrast to the minimax MSE rate $\alpha^2 \mathsf{n}^{-1}$ for $\mathsf{n} \gtrsim \alpha^2$ due to \cite{JVHW15,JVHW17}. 
Our main result with more details is formally stated in \cref{sec:main}. 

This is achieved through a unified framework that directly considers the estimation of the quantum functional $\mathrm{F}_\alpha\rbra{\rho}$, thus applicable to the classical functional $\mathrm{F}_\alpha\rbra{P}$. It is worth noting that our results are obtained through an unexpected path from \textit{quantum computing} to statistics. 
Specifically, our approach is built on quantum primitives by extending the functional estimation task to the quantum case. 
It is \textit{surprising} that from this new perspective of quantum computing, we can offer a brand new solution to this classical problem that is even better than previously known approaches \cite{JVHW15,JVHW17}.
In sharp contrast to the previous literature, our result therefore suggests a \textit{conceptually new methodology} for functional estimation, which adds to the growing list of quantum proofs for classical theorems \cite{DdW11}. 
Although our approach borrows ideas from quantum computing, our estimator for the classical functional $\mathrm{F}_\alpha\rbra{P}$ does not need to actually run on a quantum computer.\footnote{It is well-known that with an information-theoretic argument, extending a property testing problem for discrete distributions to the quantum case cannot make it any easier in sample complexity (cf.\ \cite[Fact 1.7]{OW21} and \cite[Lemma 6 in the full version]{AISW20}).}
On the other hand, our approach significantly improves the prior quantum estimators for $\mathrm{F}_\alpha\rbra{\rho}$ \cite{LW25,CW25}, which is of individual interest in the quantum computing literature (see \cref{sec:tech} for the techniques). 
Moreover, our approach can be made time-efficient on a quantum computer (see \cref{sec:time} for more details), while retaining the same sample complexity.

\subsection{Main results} \label{sec:main}

Throughout this paper, we use the asymptotic notation: $\lesssim$, $\gtrsim$, $\asymp$. 
For two non-negative quantities $A$ and $B$, we write $A \lesssim B$ if $A \leq CB$ for a universal constant $C > 0$. 
We write $A \gtrsim B$ if $B \lesssim A$, and we write $A \asymp B$ if $A \lesssim B$ and $A \gtrsim B$. 
The notation $A \gg 1$ means that $A$ is greater than a sufficiently large universal constant. 

Let $\mathcal{M}_S$ be the set of all discrete probability distributions of alphabet size $S \geq 2$ and let $P \in \mathcal{M}_S$ denote a distribution. 
We use $\hat{E}_\mathsf{n}$ to denote an estimator with sample complexity $\mathsf{n}$ and use $\hat{E}_\mathsf{n}\rbra{P}$ to denote its output (which is a random variable) for a specific unknown distribution $P$. 
For any functional of distributions, $F \colon \mathcal{M}_S \to \mathbb{R}$, let 
\begin{equation}
\textup{MSE}\rbra{\hat{E}_\mathsf{n}, F, P} \coloneqq \E\sbra*{\rbra*{\hat{E}_\mathsf{n}\rbra{P} - F\rbra{P}}^2}
\end{equation}
be the MSE risk of the estimator $\hat{E}_\mathsf{n}$ for the functional $F$ with respect to the distribution $P$. 
Similarly, we use following notations for the quantum case.
Let $\mathcal{D}\rbra{\mathcal{H}}$ be the set of (the density operators of) all quantum states in the Hilbert space $\mathcal{H}$ with $\dim\rbra{\mathcal{H}} \geq 2$. 
For a quantum state $\rho \in \mathcal{D}\rbra{\mathcal{H}}$, we use $\hat{E}_{\mathsf{n}}$ to denote an estimator with sample complexity $\mathsf{n}$ and use $\hat{E}_{\mathsf{n}}\rbra{\rho}$ to denote its output. 
For any functional of quantum states, $F \colon \mathcal{D}\rbra{\mathcal{H}} \to \mathbb{R}$, let 
\begin{equation}
    \textup{MSE}\rbra{\hat{E}_\mathsf{n}, F, \rho} \coloneqq \E\sbra*{\rbra*{\hat{E}_\mathsf{n}\rbra{\rho} - F\rbra{\rho}}^2}
\end{equation}
be the MSE risk of the estimator $\hat{E}_{\mathsf{n}}$ for the functional $F$ with respect to the quantum state $\rho$. 
Note that for any distribution $P = \rbra{p_1, p_2, \dots, p_n}$, the quantum state 
\begin{equation}
    \rho_P = \sum_{i=1}^n p_i \ketbra{i}{i}
\end{equation}
is statistically equivalent to $P$. 
In particular, if $F$ is defined in terms of a function $f$ by \cref{eq:def-FP,eq:def-Frho}, then $\textup{MSE}\rbra{\hat{E}_\mathsf{n}, F, P} = \textup{MSE}\rbra{\hat{E}_\mathsf{n}, F, \rho_P}$. 

Using the above notations, our main result can be stated as follows.

\begin{theorem}[Minimax estimation of high-order functionals, \cref{thm:minimax} simplified] \label{thm:main}
    For any real number $\alpha \gg 1$ and any  integer $\mathsf{n} \gtrsim \alpha$, we have
    \begin{equation}
    \alpha \mathsf{n}^{-1} 
    \lesssim \inf_{\hat{E}_\mathsf{n}} \sup_{P \in \mathcal{M}_S} \textup{MSE}\rbra{\hat{E}_\mathsf{n}, \mathrm{F}_\alpha,P} 
    \leq \inf_{\hat{E}_\mathsf{n}} \sup_{\rho \in \mathcal{D}\rbra{\mathcal{H}}} \textup{MSE}\rbra{\hat{E}_\mathsf{n}, \mathrm{F}_\alpha,\rho} 
    \lesssim \alpha\mathsf{n}^{-1} + \mathsf{n}^{-\frac{2}{3}+o\rbra{1}}.
    \end{equation}
    In particular, when $\alpha \lesssim \mathsf{n} \lesssim \alpha^{3-o(1)}$, we achieve the minimax optimal MSE rate
    \begin{equation}
    \inf_{\hat{E}_\mathsf{n}} \sup_{P \in \mathcal{M}_S} \textup{MSE}\rbra{\hat{E}_\mathsf{n}, \mathrm{F}_\alpha,P} \asymp \inf_{\hat{E}_\mathsf{n}} \sup_{\rho \in \mathcal{D}\rbra{\mathcal{H}}} \textup{MSE}\rbra{\hat{E}_\mathsf{n}, \mathrm{F}_\alpha,\rho} \asymp \alpha\mathsf{n}^{-1}.
    \end{equation}
\end{theorem}

The previously known minimax MSE rate given in \cite{JVHW15,JVHW17} is 
\begin{equation}
\alpha^{-4}\mathsf{n}^{-1} \lesssim \inf_{\hat{E}_\mathsf{n}} \sup_{P \in \mathcal{M}_S} \textup{MSE}\rbra{\hat{E}_\mathsf{n}, \mathrm{F}_\alpha,P} \lesssim \alpha^{2}\mathsf{n}^{-1}  \textup{ for } \mathsf{n} \gtrsim \alpha^2. 
\end{equation}
In comparison, \cref{thm:main} significantly improves both the upper and lower bounds above. 
In particular, when $\alpha \lesssim \mathsf{n} \lesssim \alpha^{3-o(1)}$, \cref{thm:main} achieves the minimax optimal MSE rate $\alpha\mathsf{n}^{-1}$.
As an application, \cref{thm:main} directly produces estimators for both $\mathrm{F}_\alpha\rbra{P}$ and $\mathrm{F}_\alpha\rbra{\rho}$ with optimal sample complexity $\mathsf{n} \asymp \alpha$, which has already been stated in \cref{thm:main-intro}.

See \cref{sec:overview} for more details of \cref{thm:main}. 
It is worth noting that our estimators achieving \cref{thm:main} can run in \textit{linear} time on a quantum computer (but in exponential time on a classical computer, see \cref{sec:time}). 
In the remainder of the Introduction, we will mainly focus on the idea of how to achieve our results. 

\subsection{Techniques} \label{sec:tech}

Our approach uses tools from quantum computing as intermediate steps in our construction, but our final estimator for the classical functional $\mathrm{F}_\alpha\rbra{P}$ is \textit{completely quantum-independent}. 
For readability by readers not familiar with quantum computing, our main idea below is described in the language of linear algebra with sufficient detail and background.
For readers familiar with quantum computing and interested in the estimator for the quantum functional $\mathrm{F}_\alpha\rbra{\rho}$, you can skip directly to \cref{sec:estimator-intro}. 

\paragraph{Notations.}
Let $z_1, z_2, \dots, z_{\mathsf{n}}$ be $\mathsf{n}$ samples from the distribution $P \in \mathcal{M}_S$. 
Let $\mathcal{H}_{d}$ denote the $d$-dimensional Hilbert space with $\ket{i}$ the standard basis vector where the $i$-th entry is $1$ and all other entries are $0$.
A vector in $\mathcal{H}_d$ has the form $\ket{\psi} = \sum_{i=1}^d a_i \ket{i}$, with $\bra{\psi} = \sum_{i=1}^d a_i^* \bra{i}$ its Hermitian conjugate and $\Abs{\ket{\psi}}$ its Euclidean norm.
In particular, the state of a qubit is a vector in $\mathcal{H}_2 \coloneqq \spanspace\cbra{\ket{0}, \ket{1}}$. 
A matrix on $\mathcal{H}_d$ is a $d \times d$ complex-valued matrix. 
For a matrix $A$, we use $A^\dag$ to denote its Hermitian conjugate and use $\Abs{A}$ to denote its operator norm. 
For a Hilbert space $\mathcal{H} = \mathcal{H}_{\mathsf{A}} \otimes \mathcal{H}_{\mathsf{B}} \otimes \mathcal{H}_{\mathsf{C}} \otimes \cdots$ with subsystems $\mathcal{H}_{\mathsf{A}}, \mathcal{H}_{\mathsf{B}}, \mathcal{H}_{\mathsf{C}}, \dots$, if a matrix $A$ acts on $\mathcal{H}_{\mathsf{A}}$, then we use $A_{\mathsf{A}} \coloneqq A_{\mathsf{A}} \otimes I_{\mathsf{B}} \otimes I_{\mathsf{C}} \otimes \cdots$ to denote the matrix on $\mathcal{H}$ that acts as $A$ on $\mathcal{H}_{\mathsf{A}}$ and acts trivially on other subsystems ($\mathcal{H}_{\mathsf{B}}, \mathcal{H}_{\mathsf{C}}, \dots$). 
For more details, see \cref{sec:preliminaries}. 

\subsubsection{Relate classical functionals to quantum functionals} \label{sec:preparation-intro}

\paragraph{Problem reformulation.}
Instead of dealing with the classical data $z_1, z_2, \dots, z_{\mathsf{n}} \sim P$, we consider the vector $\ket{\mathbf{z}} \coloneqq \ket{z_1} \otimes \ket{z_2} \otimes \dots \otimes \ket{z_{\mathsf{n}}} \in \mathcal{H}_{S^{\mathsf{n}}}$ as a quantum state. 
Intuitively, this expands the scope of the toolkit, and we will show later how the quantum toolkit helps solve our problem. 
Our starting point is that the expectation of the output of any quantum algorithm on the input quantum state $\ket{\mathbf{z}}$ is of the form
\begin{equation}
    f_{\mathcal{O}}\rbra{\mathbf{z}} \coloneqq \tr\rbra{\mathcal{O}\ketbra{\mathbf{z}}{\mathbf{z}}},
\end{equation}
where $\mathcal{O}$ is an Hermitian matrix on $\mathcal{H}_{S^{\mathsf{n}}}$ (understood as an observable in quantum computing). 
In other words, $f_{\mathcal{O}}\rbra{\mathbf{z}}$ is determined by the Hermitian matrix $\mathcal{O}$. 
Our idea is to choose $\hat{E}_{\mathsf{n}}\rbra{P} \coloneqq f_{\mathcal{O}}\rbra{\mathbf{z}}$ as our estimator. 
Our goal is then to find an Hermitian matrix $\mathcal{O}$ that minimizes the maximum MSE risk:
\begin{equation}
    \sup_{P \in \mathcal{M}_S} \textup{MSE}\rbra{\hat{E}_\mathsf{n}, \mathrm{F}_\alpha,P} = \sup_{P \in \mathcal{M}_S} \E_{\mathbf{z} \sim P^{\mathsf{n}}} \sbra*{ \rbra[\big]{ f_{\mathcal{O}}\rbra{\mathbf{z}} - \mathrm{F}_\alpha\rbra{P} }^2 }.
\end{equation}
It can be seen that once the Hermitian matrix $\mathcal{O}$ is fixed, the value of $f_{\mathcal{O}}\rbra{\mathbf{z}}$ can be classically computed by direct matrix multiplication \textit{without} the need of quantum computing. 

\paragraph{Represent $\mathcal{O}$ with quantum primitives.}
In this part, we explain how to understand the Hermitian matrix $\mathcal{O}$ from the perspective of quantum computing. 
Generally, a quantum algorithm is described by a unitary matrix $U$ that acts on the space $\mathcal{H}_{S^{\mathsf{n}}} \otimes \mathcal{H}_{2^{\ell}}$, where $\mathcal{H}_{S^{\mathsf{n}}}$ is the subsystem of $\ket{\mathbf{z}}$ and $\mathcal{H}_{2^{\ell}}$ is an ancilla subsystem. 
If we apply the unitary matrix $U$ to the quantum state $\ket{\mathbf{z}} \otimes \ket{\bar 0}_{\mathsf{anc}}$ where $\ket{\bar 0}_{\mathsf{anc}} \coloneqq \ket{0}^{\otimes \ell} \in \mathcal{H}_{2^{\ell}}$, the quantum state will become $U \rbra{\ket{\mathbf{z}} \otimes \ket{\bar 0}_{\mathsf{anc}}}$. 
If we further perform a quantum measurement $\cbra{M_m}$ (with the completeness condition $\sum_m M_m^\dag M_m = I$) on the quantum state $U \rbra{\ket{\mathbf{z}} \otimes \ket{\bar 0}_{\mathsf{anc}}}$, we will obtain a random variable $X$ as the measurement outcome such that
\begin{equation} \label{eq:prob-X=m}
    \Pr\sbra*{ X = m } = \Abs*{M_m U \rbra[\big]{ \ket{\mathbf{z}} \otimes \ket{\bar 0}_{\mathsf{anc}} }}^2.
\end{equation}
Then, it can be seen that the expected value of $X$ has the form $\E\sbra{X} = f_{\mathcal{O}}\rbra{\mathbf{z}}$, where
\begin{equation} \label{def:O-intro}
    \mathcal{O} = \sum_{m} m \bra{\bar 0}_{\mathsf{anc}} U^\dag M_m^\dag M_m U \ket{\bar 0}_{\mathsf{anc}}.
\end{equation}
This means that the Hermitian matrix $\mathcal{O}$ is determined by specifying the unitary matrix $U$ and the set of matrices $\cbra{M_m}$ (with the completeness condition). 

\paragraph{Relate $f_{\mathcal{O}}\rbra{\mathbf{z}}$ to $\mathrm{F}_\alpha\rbra{\rho}$.}
Our goal is to explicitly specify 
the Hermitian matrix $\mathcal{O}$. 
Before the detail of our construction, note that the MSE of the estimator $f_{\mathcal{O}}\rbra{\mathbf{z}}$ is determined by the variance and the expectation of $f_{\mathcal{O}}\rbra{\mathbf{z}}$ over all $\mathbf{z}$:
\begin{equation} \label{eq:mse-intro}
    \textup{MSE}\rbra{\hat{E}_\mathsf{n}, \mathrm{F}_\alpha,P} 
    = \Var_{\mathbf{z} \sim P^{\mathsf{n}}} \sbra*{ f_{\mathcal{O}}\rbra{\mathbf{z}} } + \rbra*{ \E_{\mathbf{z} \sim P^{\mathsf{n}}} \sbra*{f_{\mathcal{O}}\rbra{\mathbf{z}}} - \mathrm{F}_\alpha\rbra{P} }^2.
\end{equation}
Note that the expectation of $f_{\mathcal{O}}\rbra{\mathbf{z}}$ can be written as:
\begin{align}
    \E_{\mathbf{z} \sim P^{\mathsf{n}}} \sbra*{f_{\mathcal{O}}\rbra{\mathbf{z}}}
    & = \sum_{z_1 = 1}^S \sum_{z_2 = 1}^S \cdots \sum_{z_{\mathsf{n}}=1}^S p_{z_1} p_{z_2} \dots p_{z_\mathsf{n}} \tr\rbra*{\mathcal{O} \rbra[\big]{\ketbra{z_1}{z_1} \otimes \ketbra{z_2}{z_2} \otimes \dots \otimes \ketbra{z_{\mathsf{n}}}{z_{\mathsf{n}}}}} \\
    & = \tr\rbra*{ \mathcal{O} \rbra*{ \sum_{z_1 = 1}^S p_{z_1} \ketbra{z_1}{z_1} \otimes \sum_{z_2 = 1}^S p_{z_2} \ketbra{z_2}{z_2} \otimes \dots \otimes \sum_{z_{\mathsf{n}} = 1}^S p_{z_{\mathsf{n}}} \ketbra{z_{\mathsf{n}}}{z_{\mathsf{n}}} } } \\
    & = \tr\rbra*{ \mathcal{O} \rbra*{\sum_{i=1}^S p_i \ketbra{i}{i}}^{\otimes \mathsf{n}} } = \tr\rbra*{\mathcal{O} \rho_P^{\otimes \mathsf{n}}}, \label{eq:fOz=Orhon}
\end{align}
where $\rho_P \coloneqq \sum_{i=1}^S p_i \ketbra{i}{i} = \diag\rbra{p_1, p_2, \dots, p_S}$ is a density matrix on $\mathcal{H}_{S}$ (which can be understood as a mixed quantum state). 
With regard to this, our goal is to find an Hermitian matrix $\mathcal{O}$ such that $\tr\rbra{\mathcal{O} \rho^{\otimes \mathsf{n}}}$ is close to $\mathrm{F}_\alpha\rbra{P} = \mathrm{F}_\alpha\rbra{\rho_P} = \tr\rbra{\rho_P^\alpha}$.

\subsubsection{The construction of our estimator} \label{sec:estimator-intro}

Now we consider the estimator $\hat{E}_\mathsf{n}$ for a general quantum state $\rho$ of dimension $S$, which can be also described with an observable $\mathcal{O}$ (an Hermitian matrix) defined by \cref{def:O-intro}, or, equivalently, by the unitary matrix $U$ and the quantum measurement $\cbra{M_m}$.
The difference is that the quantum estimator can be understood by replacing $\ket{\mathbf{z}}$ in \cref{eq:prob-X=m} with $\rho^{\otimes \mathsf{n}}$. 
Specifically, the estimator $\hat{E}_\mathsf{n}$ is defined by
\begin{equation} \label{eq:def-hatE-q}
    \Pr\sbra*{ \hat{E}_\mathsf{n}\rbra{\rho} = m } = \tr\rbra*{M_m U \rbra[\big]{ \rho^{\otimes \mathsf{n}} \otimes \ketbra{\bar 0}{\bar 0}_{\mathsf{anc}} } U^\dag M_m^\dag }.
\end{equation}
Similar to \cref{eq:mse-intro}, the MSE of the estimator $\hat{E}_\mathsf{n}\rbra{\rho}$ is given by: 
\begin{equation} \label{eq:mse-q-intro}
    \textup{MSE}\rbra{\hat{E}_\mathsf{n}, \mathrm{F}_\alpha,\rho} 
    = \Var \sbra*{\hat{E}_\mathsf{n}\rbra{\rho}} + \rbra*{ \E \sbra*{\hat{E}_\mathsf{n}\rbra{\rho}} - \mathrm{F}_\alpha\rbra{\rho} }^2.
\end{equation}
Let $\cbra{\alpha} \coloneqq \alpha - \floor{\alpha}$ denote the fractional part of $\alpha$. 
The construction of $\hat{E}_\mathsf{n}\rbra{\rho}$ is given below. 

\paragraph{Step 1: Generalized SWAP test.}
We notice a quantum circuit shown in \cref{fig:swap-intro} based on the SWAP test \cite{BCWdW01,EAO+02}, which allows us to estimate values of the form $\tr\rbra{B \rho^{k}}$ for any matrix $B$ given a unitary matrix $W$ whose upper left block is $B$. 

Specifically, suppose that the matrix $B$ acts on $\mathcal{H}_{S}$ and the unitary matrix $W$ acts on $\mathcal{H}_{2^aS}$ with $B$ its upper left block. 
Then, the quantum circuit in \cref{fig:swap-intro} can be described by a unitary matrix $U_{k,W}^{\rbra{1}}$ on the Hilbert space $\mathcal{H}_{\mathsf{C}} \otimes \mathcal{H}_{\mathsf{A}_1} \otimes \dots \otimes \mathcal{H}_{\mathsf{A}_{k}}$ with $\mathcal{H}_{\mathsf{C}} \simeq \mathcal{H}_2$ and $\mathcal{H}_{\mathsf{A}_i} \simeq \mathcal{H}_S \otimes \mathcal{H}_{2^a}$ such that 
\begin{equation} \label{eq:def-U1}
    U_{k,W}^{\rbra{1}} \coloneqq H_{\mathsf{C}} \cdot \textup{Ctrl}_{\mathsf{C}}\textup{-}W_{\mathsf{A}_1} \cdot \textup{Ctrl}_{\mathsf{C}}\textup{-}\rbra{\textup{Shift}_{k}}_{\mathsf{A}_1\dots\mathsf{A}_{k}} \cdot H_{\mathsf{C}},
\end{equation}
where $\textup{Ctrl}_{\mathsf{C}}\textup{-}V_{\mathsf{B}} \coloneqq \diag\rbra{I_{\mathsf{B}}, V_{\mathsf{B}}}$ acts on $\mathcal{H}_{\mathsf{C}} \otimes \mathcal{H}_{\mathsf{B}}$ for any unitary matrix $V_{\mathsf{B}}$ on any subsystem $\mathcal{H}_{\mathsf{B}}$. 
Here, $H$ is the Hadamard matrix of order $2$ and $\textup{Shift}_{k}$ is a unitary matrix such that
\begin{equation}
    \textup{Shift}_{k} \colon \ket{\psi_1}_{\mathsf{A}_1} \otimes \ket{\psi_2}_{\mathsf{A}_2} \otimes \dots \otimes \ket{\psi_{k-1}}_{\mathsf{A}_{k-1}} \otimes \ket{\psi_{k}}_{\mathsf{A}_{k}} \mapsto \ket{\psi_2}_{\mathsf{A}_1} \otimes \ket{\psi_3}_{\mathsf{A}_2} \otimes \dots \otimes \ket{\psi_{k}}_{\mathsf{A}_{k-1}} \otimes \ket{\psi_1}_{\mathsf{A}_{k}}
\end{equation}
for any $\ket{\psi_1}, \ket{\psi_2}, \dots, \ket{\psi_{k}}$.

\begin{figure} [!htp]
    \centering
    \begin{quantikz} [row sep = {25pt, between origins}]
        & \lstick{$\ket{0}_{\mathsf{C}}$} \wireoverride{n} & \gate{H} & \ctrl{2} & \ctrl{2} & \gate{H} & \meter{} & \setwiretype{c} \rstick{$x$} \\
        \lstick[2]{$\mathsf{A}_1$}~ & \lstick{$\rho$} \wireoverride{n} & & \gate[6,disable auto height]{\textup{Shift}_{k}} & \gate[2,disable auto height]{W = \begin{bmatrix}
            B & * \\
            * & * 
        \end{bmatrix}} & & \\
        & \lstick{$\ket{\bar 0}$} \wireoverride{n} & & & & &  \\
        & \lstick{$\vdots$} \wireoverride{n} & & & & & \\
        & \lstick{$\vdots$} \wireoverride{n} & & & & & \\
        \lstick[2]{$\mathsf{A}_{k}$}~ & \lstick{$\rho$} \wireoverride{n} & & & & & \\
        & \lstick{$\ket{\bar 0}$} \wireoverride{n} & & & & & 
    \end{quantikz}
    \caption{Quantum circuit for estimating $\tr\rbra{B\rho^{k}}$.}
    \label{fig:swap-intro}
\end{figure}
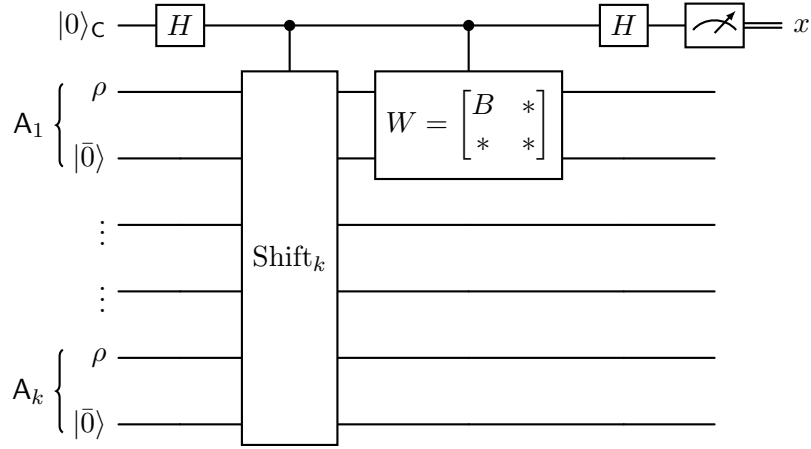

It can be shown (see \cref{lemma:swap-block}) that if we apply the unitary matrix $U_{k,W}^{\rbra{1}}$ to the mixed quantum state
\begin{equation}
    \sigma \coloneqq \ketbra{0}{0}_{\mathsf{C}} \otimes \rbra{\rho \otimes \ketbra{\bar 0}{\bar 0}}_{\mathsf{A}_1} \otimes \dots \otimes \rbra{\rho \otimes \ketbra{\bar 0}{\bar 0}}_{\mathsf{A}_{k}} \simeq \ketbra{0}{0} \otimes \rbra{\rho \otimes \ketbra{\bar 0}{\bar 0}}^{\otimes k},
\end{equation}
followed by a measurement in the computational basis of $\mathcal{H}_\mathsf{C}$, where $\ket{\bar 0} \coloneqq \ket{0}^{\otimes a}$, then the measurement outcome $x$ will be $0$ with probability
\begin{equation} \label{eq:prob-x=b}
    \Pr\sbra{x = 0} = \tr\rbra*{ \Pi_0 U_{k,W}^{\rbra{1}} \sigma \rbra{U_{k,W}^{\rbra{1}}}^\dag \Pi_0^\dag } = \frac{1 + \Real\rbra*{\tr\rbra{B\rho^{k}}}}{2},
\end{equation}
where $\Pi_0 \coloneqq \ketbra{0}{0}_{\mathsf{C}}$.
As a warm-up, taking $k = \floor{\alpha}$ and if we have a unitary matrix $W$ whose upper left block is $B = \rho^{\cbra{\alpha}}$ (which is Hermitian), then $\Pr\sbra{x = 0} = \rbra{1+\tr\rbra{\rho^\alpha}}/{2}$. 

\paragraph{Step 2: Quantum singular value transformation.}
Let $A$ be an Hermitian matrix on $\mathcal{H}_S$ and $p \in \mathbb{R}\sbra{x}$ be an even or odd polynomial with the condition that $\abs{p\rbra{x}} \leq 1$ for $x \in \sbra{-1, 1}$. 
If we have a unitary matrix $U_A$ on $\mathcal{H}_{2^aS}$ with upper left block $A$, then quantum singular value transformation \cite{GSLW19} (see \cref{lemma:qsvt}) allows us to construct another unitary matrix $U_{p\rbra{A}}$ with upper left block $p\rbra{A}$ such that $U_{p\rbra{A}}$ is of the form
\begin{equation}
    U_{p\rbra{A}} = V_{Q} \rbra{U'_Q \otimes I} V_{Q-1} \rbra{U'_{Q-1} \otimes I} \dots V_1 \rbra{U'_1 \otimes I} V_0, 
\end{equation}
where $Q = O\rbra{\deg\rbra{p}}$, $V_0, V_1, \dots, V_Q$ are unitary matrices independent of $U_A$ but dependent on $p$, and each $U'_1, U'_2, \dots, U'_Q$ is either $\textup{Ctrl}\textup{-}U_A \coloneqq \diag\rbra{I, U_A}$ or $\textup{Ctrl}\textup{-}U_A^\dag \coloneqq \diag\rbra{I, U_A^\dag}$. 
Here, we say that $U_{p\rbra{A}}$ uses $U_A$ for $Q$ times.\footnote{In quantum computing, we say that the unitary operator $U_{p\rbra{A}}$ can be implemented by a quantum circuit with $Q$ queries to the unitary operator $U_A$. Here, a query to $U_A$ means a query to \mbox{(controlled-)$U_A$} or \mbox{(controlled-)$U_A^\dag$}.} 

Let $U_{k, p, A}^{\rbra{2}} \coloneqq U_{k,U_{p\rbra{A}}}^{\rbra{1}}$ be the unitary matrix $U_{k,W}^{\rbra{1}}$ defined by \cref{eq:def-U1} with $W \coloneqq U_{p\rbra{A}}$, where, in other words, we take $B \coloneqq p\rbra{A}$ for our purpose.
Then, the same argument for \cref{eq:prob-x=b} gives that
\begin{equation} \label{eq:prob2}
    \tr\rbra*{ \Pi_0 U_{k, p, A}^{\rbra{2}} \sigma \rbra{U_{k, p, A}^{\rbra{2}}}^\dag \Pi_0^\dag } = \frac{1 + \Real\rbra*{\tr\rbra*{p\rbra{A} \rho^{k}}}}{2}.
\end{equation}
As a warm-up, if we take $k = \floor{\alpha}$ and $p\rbra{x} = x^{\cbra{\alpha}}$, then $\eqref{eq:prob2} = \rbra{1+\tr\rbra{\rho^\alpha}}/{2}$.

\paragraph{Step 3: Samplization.}
In Step 2, the unitary matrix $U_{k, p, A}^{\rbra{2}}$ is constructed using a given unitary matrix $U_A$ with upper left block $A$. 
At this stage, our goal is to construct a unitary matrix $U_{A}$ with upper left block $A \propto \rho$.
This can be done with the quantum samplizer \cite{WZ25b,WZ25} (see \cref{lemma:samplizer}), which allows us to approximate a quantum circuit $\mathcal{C}\sbra{U_{\rho/2}}$ with a quantum channel $\mathcal{E}\sbra{\rho}$ to precision $\delta$. 
Specifically, if $\mathcal{C}\sbra{U_A}$ uses $Q$ queries to the unitary oracle $U_A$ (whose upper left block is $A$), then $\mathcal{E}\sbra{\rho}$ can be implemented using $\widetilde{\Theta}\rbra{Q^2/\delta}$ samples of $\rho$ such that $\Abs{\mathcal{C}\sbra{U_{\rho/2}} - \mathcal{E}\sbra{\rho}}_{\diamond} \leq \delta$ for some unitary oracle $U_{\rho/2}$ with upper left block $\rho/2$, where $\Abs{\cdot}_{\diamond}$ is the diamond norm for superoperators. 

In our case, we recall that $U_{k, p, A}^{\rbra{2}} = U_{k, p, A}^{\rbra{2}}\sbra{U_A}$ is a unitary matrix that uses $U_A$ for $O\rbra{\deg\rbra{p}}$ times. 
That is, $U_{k, p, A}^{\rbra{2}}$ is a quantum circuit with $O\rbra{\deg\rbra{p}}$ queries to the unitary oracle $U_A$. 
By the quantum samplizer, we can implement a quantum channel $\mathcal{E}\sbra{\rho}$ using $\mathsf{S} = \widetilde{O}\rbra{\rbra{\deg\rbra{p}}^2/\delta}$ samples of $\rho$ such that there is a unitary matrix $U_{\rho/2}$ with upper left block $\rho/2$ satisfying
\begin{equation}
    \Abs*{\rbra{U_{k, p, \rho/2}^{\rbra{2}}\sbra{U_{\rho/2}}}\rbra{\cdot}\rbra{U_{k, p, \rho/2}^{\rbra{2}}\sbra{U_{\rho/2}}}^\dag - \mathcal{E}\sbra{\rho}}_{\diamond} \leq \delta.
\end{equation}
This means that there is a unitary matrix $U^{\rbra{3}}_{k,p}$ (which implements $\mathcal{E}\sbra{\rho}$) and a unitary matrix $U_{\rho/2}$ (which has upper left block $\rho/2$) such that 
\begin{equation} \label{eq:prob3}
    \Abs*{ \tr_{\mathsf{env}}\rbra*{ U^{\rbra{3}}_{k,p} \rbra*{ \underbrace{\rho^{\otimes \mathsf{S}} \otimes \ketbra{0}{0}^{\otimes l}}_{\mathsf{env}} \otimes \sigma} \rbra{U^{\rbra{3}}_{k,p}}^\dag } - \rbra{U_{k, p, \rho/2}^{\rbra{2}}\sbra{U_{\rho/2}}} \sigma \rbra{U_{k, p, \rho/2}^{\rbra{2}}\sbra{U_{\rho/2}}}^\dag }_{1} \leq \delta.
\end{equation}
Combining \cref{eq:prob2,eq:prob3}, we have (note that $\tr\rbra{p\rbra{\rho/2} \rho}$ is a real number)
\begin{equation} \label{eq:prob4}
    \abs*{ \tr\rbra*{\Pi_0  U^{\rbra{3}}_{k,p} \rbra*{ \rho^{\otimes \mathsf{S}} \otimes \ketbra{0}{0}^{\otimes l} \otimes \sigma} \rbra{U^{\rbra{3}}_{k,p}}^\dag  \Pi_0^\dag} - \frac{1 + \tr\rbra*{p\rbra{\rho/2} \rho^{k}}}{2} } \leq \delta.
\end{equation}

\paragraph{Step 4: A baby estimator.}
Now for our purpose, we need the polynomial $p_{c,\delta,\varepsilon} \in \mathbb{R}\sbra{x}$ in \cite{Gil19} (see \cref{lemma:poly-approx-pos}) such that (i) $\abs{p_{c,\delta,\varepsilon}\rbra{x} - \frac{1}{2} x^c} \leq \varepsilon$ for $x \in \sbra{\delta, 1}$, (ii) $\abs{p_{c,\delta,\varepsilon}\rbra{x}} \leq 1$ for $x \in \sbra{-1, 1}$, and (iii) $\deg\rbra{p_{c,\delta,\varepsilon}} = O\rbra{\frac{1}{\delta}\log\rbra{\frac{1}{\varepsilon}}}$, where $O\rbra{\cdot}$ hides a constant that is independent of $c$.
Taking $p \coloneqq p_{\cbra{\alpha},\delta',\varepsilon'}$ and $k \coloneqq \floor{\alpha}$ in $U^{\rbra{3}}_{k,p}$, \cref{eq:prob4} gives
\begin{equation} \label{eq:before-perm}
    \abs*{ \tr\rbra*{\rbra*{ 2 \rbra{U^{\rbra{3}}_{\floor{\alpha},p_{\cbra{\alpha},\delta',\varepsilon'}}}^\dag \Pi_0 U^{\rbra{3}}_{\floor{\alpha},p_{\cbra{\alpha},\delta',\varepsilon'}} - I } \rbra*{ \rho^{\otimes \mathsf{S}} \otimes \ketbra{0}{0}^{\otimes l} \otimes \sigma }} - \tr\rbra*{p_{\cbra{\alpha},\delta',\varepsilon'}\rbra*{\frac{\rho}{2}}\rho^{\floor{\alpha}}}} \leq 2\delta,
\end{equation}
where $\mathsf{S} = \widetilde{O}\rbra{\rbra{\deg\rbra{p_{\cbra{\alpha},\delta',\varepsilon'}}}^2/\delta} = \widetilde{O}\rbra{1/\rbra{\delta'^2\delta}}$. 

Note that 
\begin{equation}
    \rho^{\otimes \mathsf{S}} \otimes \ketbra{0}{0}^{\otimes l} \otimes \sigma \simeq \underbrace{\rho^{\otimes \rbra{\floor{\alpha} + \mathsf{S}}}}_{\mathsf{A}} \otimes \underbrace{\ketbra{0}{0}^{\otimes \ell}}_{\mathsf{B}}
\end{equation}
for some $\ell \geq 0$.
Then, \cref{eq:before-perm} gives 
\begin{equation}
    \abs*{ \tr\rbra*{\widetilde{\mathcal{O}}_{\textup{baby}} \rho^{\otimes \rbra{\floor{\alpha} + \mathsf{S}}}} - \tr\rbra*{p_{\cbra{\alpha},\delta',\varepsilon'}\rbra*{\frac{\rho}{2}}\rho^{\floor{\alpha}}} } \leq 2\delta,
\end{equation}
where
\begin{equation}
    \widetilde{\mathcal{O}}_{\textup{baby}} \coloneqq 2 \underbrace{\bra{0}^{\otimes \ell}}_{\mathsf{B}} \rbra{U^{\rbra{3}}_{\floor{\alpha},p_{\cbra{\alpha},\delta',\varepsilon'}}}^\dag \Pi_0 U^{\rbra{3}}_{\floor{\alpha},p_{\cbra{\alpha},\delta',\varepsilon'}} \underbrace{\ket{0}^{\otimes \ell}}_{\mathsf{B}} - I.
\end{equation}
Let
\begin{equation}
    \mathcal{O}_{\textup{baby}} \coloneqq 2^{1+\cbra{\alpha}}\widetilde{\mathcal{O}}_{\textup{baby}}.
\end{equation}
Then, it can be shown (see \cref{prop:baby-diff}) that
\begin{equation}
    \abs*{ \tr\rbra*{\mathcal{O}_{\textup{baby}} \rho^{\otimes \rbra{\floor{\alpha} + \mathsf{S}}}} - \tr\rbra{\rho^\alpha} } \leq 8\delta + 4\varepsilon' + 6 \rbra{2\delta'}^{\floor{\alpha}-1}.
\end{equation} 

The Hermitian matrix $\mathcal{O}_{\textup{baby}}$ defines an estimator $\hat{E}_{\textup{baby},\mathsf{n}}$ with sample complexity $\mathsf{n} = \floor{\alpha} + \mathsf{S}$, which is able to give an upper bound on the second term of \cref{eq:mse-q-intro}. 
However, the first term of \cref{eq:mse-q-intro} only has the trivial upper bound $\Var \sbra{\hat{E}_{\textup{baby},\mathsf{n}}\rbra{\rho}} \leq O\rbra{1}$.
In the next step, we construct the Hermitian matrix $\mathcal{O}$ based on $\mathcal{O}_{\textup{baby}}$. 

\paragraph{Step 5: Final estimator.}
For $\mathsf{n} \geq \floor{\alpha} + \mathsf{S}$, let $\mathsf{m} = \frac{\mathsf{n}}{\floor{\alpha} + \mathsf{S}}$ (without loss of generality, we assume that $\mathsf{m}$ is an integer) and we split $\rho^{\otimes \mathsf{n}}$ into $\mathsf{m}$ subsystems $\mathsf{Z}_1, \mathsf{Z}_{2}, \dots, \mathsf{Z}_{\mathsf{m}}$, each with the same dimension as $\mathcal{O}_{\textup{baby}}$. 
Then, we can construct the following Hermitian matrix
\begin{equation}
    \mathcal{O} \coloneqq \frac{1}{\mathsf{m}} \sum_{i=1}^{\mathsf{m}} \rbra{\mathcal{O}_{\textup{baby}}}_{\mathsf{Z}_i},
\end{equation}
which defines the estimator $\hat{E}_{\mathsf{n}}$. 
It can be verified (see \cref{prop:E-Var-O}) that 
\begin{equation}
    \Var \sbra*{\hat{E}_\mathsf{n}\rbra{\rho}} \leq O\rbra*{\frac{1}{\mathsf{m}}}, \text{ and } \E \sbra*{\hat{E}_\mathsf{n}\rbra{\rho}} = \tr\rbra*{\mathcal{O}_{\textup{baby}} \rho^{\otimes \rbra{\floor{\alpha}+\mathsf{S}}}}.
\end{equation}
By \cref{eq:mse-q-intro}, the MSE of the estimator $\hat{E}_{\mathsf{n}}$ is given by
\begin{equation}
    \textup{MSE}\rbra{\hat{E}_\mathsf{n}, \mathrm{F}_\alpha,\rho} \leq O\rbra*{\frac{1}{\mathsf{m}}} + \rbra*{8\delta + 4\varepsilon' + 6 \rbra{2\delta'}^{\floor{\alpha}-1}}^2.
\end{equation}
Taking $\delta = \varepsilon' = \mathsf{n}^{-\frac{1}{3}+o\rbra{1}}$ and $\delta' = \mathsf{n}^{-o\rbra{1}}$ (which gives $\mathsf{m} = \min\cbra{\alpha^{-1}\mathsf{n}, \mathsf{n}^{\frac{2}{3}-o\rbra{1}}}$), we finally have the maximum MSE risk that
\begin{equation}
    \sup_{\rho \in \mathcal{D}\rbra{\mathcal{H}}} \textup{MSE}\rbra{\hat{E}_\mathsf{n}, \mathrm{F}_\alpha,\rho} \lesssim \alpha\mathsf{n}^{-1} + \mathsf{n}^{-\frac{2}{3}+o\rbra{1}}.
\end{equation}
Therefore, we have obtained an estimator $\hat{E}_{\mathsf{n}}\rbra{\rho}$ and thus an estimator $\hat{E}_{\mathsf{n}}\rbra{P} \coloneqq f_{\mathcal{O}}\rbra{\mathbf{z}}$ with an explicitly specified Hermitian matrix $\mathcal{O}$, achieving the upper bound stated in \cref{thm:main}.

\paragraph{Comparison with prior approaches.}
The prior approach in \cite{JVHW15,JVHW17} only works for the classical functional $\mathrm{F}_\alpha\rbra{P}$ while our approach also works for the quantum functional $\mathrm{F}_\alpha\rbra{\rho}$ using clearly different techniques. 

The prior approach in \cite{LW25} that estimates the quantum functional $\mathrm{F}_\alpha\rbra{\rho}$ with sample complexity $\mathsf{n} = \exp\rbra{O\rbra{\alpha}}$ is achieved by the Hadamard test framework \cite{AJL09}, involving the polynomial approximation of $x^{\alpha - 1}$ for the use of quantum singular value transformation \cite{GSLW19} and samplizer \cite{WZ25}. 
Due to the use of samplizer, their approach introduces an unexpected factor of $\rbra{\frac{1}{2}}^{\alpha - 1} = \exp\rbra{O\rbra{\alpha}}$ into the estimated value, which causes their complexity exponential in $\alpha$ (see \cite[Theorem 3.2]{LW25}). 
By sharp contrast, our approach employs the generalized SWAP test framework \cite{EAO+02}, which appears to be even more powerful for our purposes, and involves the polynomial approximation of a different function $x^{\cbra{\alpha}-1}$. 
Although we still use the samplizer, our approach only introduces a very small factor of $\rbra{\frac{1}{2}}^{\cbra{\alpha} - 1} = O\rbra{1}$. 
This is why our approach can improve the results in \cite{LW25} exponentially.

The prior approach in \cite{CW25} that estimates the quantum functional $\mathrm{F}_\alpha\rbra{\rho}$ with sample complexity $\mathsf{n} = O\rbra{\alpha^2}$ is achieved by the weak Schur sampling framework \cite{CHW07} and with the advanced analysis in \cite{OW17}, which is very different from our approach from the very beginning. 
Their sample complexity $\mathsf{n} = O\rbra{\alpha^2}$ cannot be directly improved due to the choice of the precision parameters in their algorithm (see \cite[Theorem 14]{CW25}). 

In addition, the approaches in \cite{LW25,CW25} did not consider the minimax MSE rate of their estimators. 

\subsubsection{Lower bounds}

Our lower bound $\Omega\rbra{\alpha\mathsf{n}^{-1}}$ on the minimax MSE for estimating $\mathrm{F}_{\alpha}\rbra{P}$ (see \cref{thm:lower}) is obtained by a sample lower bound of $\Omega\rbra{\alpha}$ for estimating $\mathrm{F}_{\alpha}\rbra{P}$ for sufficiently large real number $\alpha$ (see \cref{thm:sample-lb}).
This is achieved by combining the sample lower bound for hypothesis testing \cite{BY02} with the hard instance recently found in \cite{CWYZ25}. 
Specifically, we consider the problem of distinguishing the two distributions $P^\pm$: 
\begin{equation}
    p^{\pm}_1 = 1 - \frac{1}{\alpha} \pm \frac{\varepsilon}{\alpha}, \qquad p^{\pm}_2 = \frac{1}{\alpha} \mp \frac{\varepsilon}{\alpha}.
\end{equation}
It can be verified that the Hellinger distance between $P^+$ and $P^-$ is $d_{\textup{H}}\rbra{P^+, P^-} \lesssim \varepsilon/\sqrt{\alpha}$, whereas $\mathrm{F}_\alpha\rbra{P^+} - \mathrm{F}_\alpha\rbra{P^-} \gtrsim \varepsilon$. 
Therefore, any estimator for $\mathrm{F}_\alpha\rbra{P}$ to within additive error $\Theta\rbra{\varepsilon}$ can be used to distinguish the two distributions $P^\pm$ with high probability, thus requiring sample complexity $\Omega\rbra{1/d_{\textup{H}}^2\rbra{P^+, P^-}} \geq \Omega\rbra{\alpha/\varepsilon^2}$. 
Finally, we complete it by relating the sample lower bound $\Omega\rbra{\alpha/\varepsilon^2}$ to the minimax MSE risk of estimating $\mathrm{F}_\alpha\rbra{P}$. 

\paragraph{Comparison with prior approaches.}
Our lower bound on the minimax MSE for estimating $\mathrm{F}_\alpha\rbra{P}$ is shown by an information-theoretic argument for distinguishing two discrete distributions. 
Our approach is different from the prior approach in \cite{JVHW17}, where they used the van Trees inequality (cf.\ \cite{GL95}) to relate the minimax risk to the Bayes risk (see \cite[Section VI-A]{JVHW17}). 

To establish the non-asymptotic lower bound on the minimax MSE for estimating $\mathrm{F}_\alpha\rbra{P}$, we need to provide a non-asymptotic lower bound on $\mathrm{F}_\alpha\rbra{P^+} - \mathrm{F}_\alpha\rbra{P^-}$ for any sufficiently large real number $\alpha \gg 1$ (see \cref{lemma:F-alpha-diff-lb}). 
In comparison, the prior work \cite{CWYZ25} only considered the asymptotic lower bound on $\mathrm{F}_\alpha\rbra{P^+} - \mathrm{F}_\alpha\rbra{P^-}$, which does not directly apply to our non-asymptotic regime. 
To prove \cref{lemma:F-alpha-diff-lb}, we need some specific inequalities for real $\alpha$, compared to the analysis of function limits in \cite{CWYZ25}.

\subsection{Related work}

The estimation of functionals of multiple unknown distributions such as the total variation distance and the Kullback--Leibler divergence has also been studied in \cite{VV17,JHW18,HJW20,BZLV18}.
The promised decision versions (i.e., the closeness testing of discrete distributions) were studied in \cite{BFR+13,CDVV14}.
When one of the distributions is known, the closeness testing was studied in \cite{BFF+01,Pan08,DK16}.
There are quantum approaches to the property testing of a classical distribution in the quantum query model, including closeness testing \cite{BHH11,CFMdW10,Mon15,GL20,LWL24,CKO25} and entropy estimation \cite{LW19,GL20,WZL24,LW25,Wan25,SJ25}.

The estimation of functionals of multiple unknown quantum states has also been studied, including multivariate traces \cite{QKW24}, the Uhlmann fidelity and trace distance \cite{WZC+23,GP22,WGL+24,WZ24,Wan24,WZ24b,UNWT25}, the Umegaki divergence (i.e., the von Neumann relative entropy) \cite{Hay25}, the $\ell_\alpha$ distance \cite{LW25b}, and the Tsallis relative entropy \cite{BGW25}.
The promised decision versions (i.e., the closeness testing of quantum states) were studied in \cite{OW21,BOW19,OW26}.

\subsection{Discussion}

In this paper, we present an estimator for the high-order functionals $\mathrm{F}_{\alpha}\rbra{P}$ of a discrete distribution $P$ and $\mathrm{F}_\alpha\rbra{\rho}$ of a quantum state $\rho$, which improves the prior results in \cite{JVHW15,JVHW17} and \cite{LW25,CW25}, and achieves the optimal sample complexity and the minimax MSE in a certain regime. 
This is done through a novel approach using techniques from quantum computing, which brings new insights into statistics. 
A conceptual message is that the theory of quantum computing can contribute to other fields, which adds to the growing list of quantum proofs for classical theorems \cite{DdW11}. 

We conclude with some open questions for future research. 
\begin{enumerate}
    \item Can we improve the minimax MSE bounds for estimating $\mathrm{F}_{\alpha}\rbra{P}$ and $\mathrm{F}_{\alpha}\rbra{\rho}$? 
    One may conjecture that the minimax MSE rate is $\alpha \mathsf{n}^{-1}$ for any large real number $\alpha \gg 1$. 
    \item As mentioned in \cref{sec:time}, the classical time complexity of our estimator is exponential in $\mathsf{n}$. 
    Nevertheless, our estimator can have a time complexity linear in $\mathsf{n}$ on a quantum computer. 
    Can we find an estimator for $\mathrm{F}_{\alpha}\rbra{P}$ that is time-efficient on a classical computer?
    \item As a similar estimation task, can we improve the prior bounds for estimating the $\alpha$-R\'enyi entropy $\mathrm{H}_\alpha^{\textup{R\'en}}\rbra{P} = \frac{\ln\rbra{\mathrm{F}_\alpha\rbra{P}}}{1-\alpha}$ in terms of $\alpha$? 
    The prior best approach due to \cite{AOST17} has sample complexity (doubly) exponential in $\alpha$. 
    Moreover, what about the estimation of the Hill number $\prescript{\alpha}{}{\mathrm{D}}\rbra{P} = \rbra{\mathrm{F}_\alpha\rbra{P}}^{1/\rbra{1-\alpha}}$ and frequency moment $\mathrm{M}_{\alpha}\rbra{P} = S^\alpha \mathrm{F}_\alpha\rbra{P}$?
    \item Can we find more classical applications to which quantum computing brings new insights? 
    Conversely, can we find any quantum applications to which novel classical techniques can be applied?
\end{enumerate}

\section{Overview} \label{sec:overview}

In this paper, we show tight upper and lower bounds on the MSE for estimating the high-order functionals $\mathrm{F}_{\alpha}\rbra{P}$ and $\mathrm{F}_\alpha\rbra{\rho}$, stated as follows. 

\begin{theorem}[Minimax estimation of $\mathrm{F}_\alpha\rbra{P}$ and $\mathrm{F}_\alpha\rbra{\rho}$, \cref{thm:upper,thm:lower} combined] \label{thm:minimax}
    For any real number $\alpha \gg 1$ and integer $\mathsf{n} \gtrsim \alpha$,\footnote{Note that the lower bound in \cref{eq:lb} holds for any $\alpha \geq 4$ according to \cref{thm:lower} while the upper bound in \cref{eq:ub} holds for any $\alpha \geq 2$ according to \cref{thm:upper}. 
    Actually, for any constant $2 \leq \alpha < 4$, the lower bound can be obtained by the lower bound $\Omega\rbra{\mathsf{n}^{-1}}$ for $\alpha \geq 1.5$ due to \cite[Theorem 2]{JVHW17}.}
    we have
    \begin{align}
        \frac{1}{4e^3} \cdot \min\cbra*{\frac{\alpha-1}{4\mathsf{n}}, 1} 
        & \leq \inf_{\hat{E}_\mathsf{n}} \sup_{P \in \mathcal{M}_S} \textup{MSE}\rbra{\hat{E}_\mathsf{n}, \mathrm{F}_\alpha,P} \label{eq:lb} \\
        & \leq \inf_{\hat{E}_\mathsf{n}} \sup_{\rho \in \mathcal{D}\rbra{\mathcal{H}}} \textup{MSE}\rbra{\hat{E}_\mathsf{n}, \mathrm{F}_\alpha,\rho} \\
        & \leq \frac{32\alpha}{\mathsf{n}} + \Theta\rbra*{ \rbra*{\frac{\log^4\rbra{\mathsf{n}}}{\mathsf{n}}}^{\frac{2\floor{\alpha}-2}{3\floor{\alpha}-1}} }. \label{eq:ub}
    \end{align}
    In particular, when 
    \begin{equation}
        \alpha \lesssim \mathsf{n} \lesssim \alpha^{3-\frac{4}{\floor{\alpha}+1}} \rbra*{\log\rbra{\alpha}}^{-8+\frac{16}{\floor{\alpha}+1}} = \alpha^{3 - o\rbra{1}},
    \end{equation}
    we achieve the minimax optimal MSE risk
    \begin{equation}
        \inf_{\hat{E}_\mathsf{n}} \sup_{P \in \mathcal{M}_S} \textup{MSE}\rbra{\hat{E}_\mathsf{n}, \mathrm{F}_\alpha,P} \asymp \inf_{\hat{E}_\mathsf{n}} \sup_{\rho \in \mathcal{D}\rbra{\mathcal{H}}} \textup{MSE}\rbra{\hat{E}_\mathsf{n}, \mathrm{F}_\alpha,\rho} \asymp \alpha\mathsf{n}^{-1}.
    \end{equation}
\end{theorem}

In \cite[Equations (16) and (116)]{JVHW17}, it was shown that for $\alpha \geq 2$, 
\begin{equation}
    \frac{\rbra*{60\alpha\rbra*{\mathrm{B}\rbra{\alpha+3,3}-\mathrm{B}\rbra{\alpha+2, 4}}}^2}{5\mathsf{n}+45} \leq \inf_{\hat{E}_\mathsf{n}} \sup_{P \in \mathcal{M}_S} \textup{MSE}\rbra{\hat{E}_\mathsf{n}, \mathrm{F}_\alpha,P} \leq \rbra*{\frac{\alpha\rbra{\alpha-1}}{2\mathsf{n}}}^2 + \frac{\alpha^2}{4\mathsf{n}},
\end{equation}
where 
\begin{equation}
    \mathrm{B}\rbra{a, b} = \int_{0}^{1} x^{a-1} \rbra{1-x}^{b-1} \mathrm{d} x
\end{equation}
is the Beta function. 
When $\mathsf{n} \gtrsim \alpha^2$, their minimax MSE risk is asymptotically
\begin{equation}
    \alpha^{-4}\mathsf{n}^{-1} \lesssim \inf_{\hat{E}_\mathsf{n}} \sup_{P \in \mathcal{M}_S} \textup{MSE}\rbra{\hat{E}_\mathsf{n}, \mathrm{F}_\alpha,P} \lesssim \alpha^{2}\mathsf{n}^{-1}. 
\end{equation}
By comparison, \cref{thm:minimax} gives the minimax MSE risk
\begin{equation}
    \alpha \mathsf{n}^{-1} \lesssim \inf_{\hat{E}_\mathsf{n}} \sup_{P \in \mathcal{M}_S} \textup{MSE}\rbra{\hat{E}_\mathsf{n}, \mathrm{F}_\alpha,P} \lesssim \alpha\mathsf{n}^{-1} + \mathsf{n}^{-\frac{2}{3}+o\rbra{1}}
\end{equation}
for $\mathsf{n} \gtrsim \alpha$,
which improves both the upper and lower bounds given in \cite{JVHW17}:
\begin{itemize}
    \item Our lower bound $\alpha \mathsf{n}^{-1}$ is better than the lower bound $\alpha^{-4}\mathsf{n}^{-1}$ in \cite{JVHW17} for any $\mathsf{n} \gtrsim \alpha$.
    \item Our upper bound $\alpha\mathsf{n}^{-1} + \mathsf{n}^{-\frac{2}{3}+o\rbra{1}}$ is better than the upper bound $\alpha^{2}\mathsf{n}^{-1}$ in \cite{JVHW17} for $\alpha \lesssim \mathsf{n} \lesssim \alpha^{6-o\rbra{1}}$. 
\end{itemize}
Combining \cref{thm:minimax} with the results in \cite{JVHW17}, the MSE of estimating $\mathrm{F}_\alpha\rbra{P}$ is bounded by
\begin{align}
    \inf_{\hat{E}_\mathsf{n}} \sup_{P \in \mathcal{M}_S} \textup{MSE}\rbra{\hat{E}_\mathsf{n}, \mathrm{F}_\alpha,P} 
    & \leq \min\cbra*{ \frac{32\alpha}{\mathsf{n}} + \Theta\rbra*{ \rbra*{\frac{\log^4\rbra{\mathsf{n}}}{\mathsf{n}}}^{\frac{2\floor{\alpha}-2}{3\floor{\alpha}-1}} }, \rbra*{\frac{\alpha\rbra{\alpha-1}}{2\mathsf{n}}}^2 + \frac{\alpha^2}{4\mathsf{n}}} \\
    & \asymp \min\cbra*{ \alpha\mathsf{n}^{-1} + \mathsf{n}^{-\frac{2}{3}+o\rbra{1}}, \alpha^2 \mathsf{n}^{-1} } \\
    & \asymp \begin{cases}
        \alpha \mathsf{n}^{-1}, & \alpha \lesssim \mathsf{n} \lesssim \alpha^{3-o\rbra{1}}, \\
        \mathsf{n}^{-\frac{2}{3}+o\rbra{1}}, &  \alpha^{3-o\rbra{1}} \lesssim \mathsf{n} \lesssim \alpha^{6-o\rbra{1}}, \\
        \alpha^2 \mathsf{n}^{-1}, & \mathsf{n} \gtrsim \alpha^{6 - o\rbra{1}}. 
    \end{cases}
\end{align}

\subsection{Sample complexity}

The estimator that achieves the minimax MSE risk in \cref{thm:minimax} can also give a sample complexity upper bound for estimating $\mathrm{F}_{\alpha}\rbra{P}$ and $\mathrm{F}_\alpha\rbra{\rho}$ to within additive error $\varepsilon$. 

\begin{corollary}[Sample complexity of estimating $\mathrm{F}_\alpha\rbra{P}$ and $\mathrm{F}_\alpha\rbra{\rho}$] \label{corollary:sample-complexity}
    For any real number $\alpha \gg 1$, the sample complexity of estimating $\mathrm{F}_\alpha\rbra{P}$ and $\mathrm{F}_\alpha\rbra{\rho}$ to within additive error $\varepsilon$ is 
    \begin{equation}
        \frac{\alpha}{\varepsilon^2} \lesssim \mathsf{n} \lesssim \frac{\alpha}{\varepsilon^2} + \rbra*{\frac{1}{\varepsilon}}^{\frac{3\floor{\alpha}-1}{\floor{\alpha}-1}} \log^4\rbra*{\frac{1}{\varepsilon}}.
    \end{equation}
    In particular, when $\varepsilon$ is a constant, the sample complexity is $\mathsf{n} \asymp \alpha$. 
\end{corollary}

\begin{proof}
The sample complexity lower bound is given in \cref{thm:sample-lb}. 
To show the upper bound, we apply the upper bound in \cref{thm:minimax} as follows. 
To make the maximum MSE risk 
\begin{equation}
    \sup_{\rho \in \mathcal{D}\rbra{\mathcal{H}}} \textup{MSE}\rbra{\hat{E}_\mathsf{n}, \mathrm{F}_\alpha,\rho} \leq 0.01 \varepsilon^2, 
\end{equation}
it is sufficient to take
\begin{equation} \label{eq:mse-up}
    \mathsf{n} \asymp \frac{\alpha}{\varepsilon^2} + \rbra*{\frac{1}{\varepsilon}}^{\frac{3\floor{\alpha}-1}{\floor{\alpha}-1}} \log^4\rbra*{\frac{1}{\varepsilon}} \lesssim \frac{\alpha}{\varepsilon^2} + \frac{1}{\varepsilon^{3+o\rbra{1}}}.
\end{equation}
Under this choice of $\mathsf{n}$, the error probability of the estimator is bounded by
\begin{align}
    \Pr\sbra*{\abs*{\hat{E}_\mathsf{n}\rbra{\rho} - \mathrm{F}_\alpha\rbra{\rho}} \geq \varepsilon} 
    & \leq \frac{1}{\varepsilon^2} \rbra*{\sup_{\rho \in \mathcal{D}\rbra{\mathcal{H}}} \textup{MSE}\rbra{\hat{E}_\mathsf{n}, \mathrm{F}_\alpha,\rho}} \\
    & \leq \frac{1}{\varepsilon^2} \cdot 0.01 \varepsilon^2 \label{eq:by-eq-mse-up} \\
    & = 0.01,
\end{align}
where \cref{eq:by-eq-mse-up} is due to \cref{eq:mse-up}.
Therefore, we can use $\mathsf{n}$ samples to estimate $\mathrm{F}_{\alpha}\rbra{\rho}$ to within additive error $\varepsilon$ with success probability at least $0.99$. 
\end{proof}

It is worth noting that when $\varepsilon = \Theta\rbra{1}$ is an arbitrary constant, \cref{corollary:sample-complexity} implies the optimal sample complexity $\mathsf{n} \asymp \alpha$ for estimating $\mathrm{F}_\alpha\rbra{P}$ and $\mathrm{F}_\alpha\rbra{\rho}$ to within additive error $\varepsilon$ for any real number $\alpha \gg 1$, which gives \cref{thm:main-intro} stated at the very beginning of this paper. 

\subsection{Time efficiency} \label{sec:time}

Our estimator $\hat{E}_{\mathsf{n}}\rbra{P}$ is defined by an $S^{\mathsf{n}} \times S^{\mathsf{n}}$ Hermitian matrix $\mathcal{O}$, or, equivalently, a unitary matrix $U$ and a quantum measurement $\cbra{M_m}$, where $U$ as well as each $M_m$ is an $S^{O\rbra{\mathsf{n}}} \times S^{O\rbra{\mathsf{n}}}$ Hermitian matrix. 

\paragraph{Classical time efficiency.}
For $\mathsf{n}$ samples $z_1, z_2, \dots, z_{\mathsf{n}}$, the value of $\tr\rbra{\mathcal{O}\ketbra{\mathbf{z}}{\mathbf{z}}}$ can be computed on a classical computer by direct matrix multiplication with time complexity $\poly\rbra{S^{\mathsf{n}}}$. 
When the required additive error $\varepsilon = \Theta\rbra{1}$ is a constant, by \cref{corollary:sample-complexity}, it is sufficient to choose $\mathsf{n} \asymp \alpha$ for $\alpha \gg 1$ to estimate $\mathrm{F}_\alpha\rbra{P}$ to within additive error $\varepsilon$, which gives a time complexity of $S^{O\rbra{\alpha}}$ exponential in $\alpha$. 

\paragraph{Quantum time efficiency.}
The construction of the Hermitian matrix $\mathcal{O}$ provided in \cref{sec:upper} can be implemented by a quantum circuit of size $O\rbra{\mathsf{n}\log\rbra{S}}$. 
In other words, the quantum time complexity is $O\rbra{\mathsf{n}\log\rbra{S}}$.
When the required additive error $\varepsilon = \Theta\rbra{1}$ is a constant, it is sufficient to choose $\mathsf{n} \asymp \alpha$ for $\alpha \gg 1$ to estimate $\mathrm{F}_\alpha\rbra{P}$ and $\mathrm{F}_\alpha\rbra{\rho}$ to within additive error $\varepsilon$, which gives a quantum time complexity of $O\rbra{\alpha\log\rbra{S}}$. 
This quantum time complexity is \textit{optimal} only up to a logarithmic factor of $\log\rbra{S}$ (which is optimal when $S$ is a constant) due to the matching sample lower bound $\Omega\rbra{\alpha}$ given in \cref{corollary:sample-complexity}. 

\section{Preliminaries} \label{sec:preliminaries}

In this section, we introduce sufficient preliminaries for our results. 

\subsection{Basic notations}

We use $\mathcal{M}_S$ to denote the set of all discrete probability distributions of alphabet size $S$. 
For a distribution $P \in \mathcal{M}_S$, we use the lower-case letter $p_i$ to denote the probability of event $i$ for $1 \leq i \leq S$. 
The Hellinger distance \cite{Hel09} between two discrete probability distributions $P, Q \in \mathcal{M}_S$ is defined by
\begin{equation}
    d_\mathrm{H}\rbra{P, Q} \coloneqq \sqrt{\frac{1}{2}\sum_{i=1}^S\rbra*{\sqrt{p_i}-\sqrt{q_i}}^2}. 
\end{equation}

Let $\mathcal{H}_d$ be a $d$-dimensional Hilbert space with the standard orthonormal basis $\cbra{\ket{i}}_{i=1}^d$ (also known as the computational basis). 
A vector in $\mathcal{H}_d$ is described by $\ket{\psi} = \sum_{i=1}^d a_i \ket{i}$ with $a_i \in \mathbb{C}$. 
We use $\bra{\psi} \coloneqq \sum_{i=1}^{d} a_i^* \bra{i}$ to denote the Hermitian conjugate of $\ket{\psi}$. 
The inner product between two vectors $\ket{\varphi}, \ket{\psi} \in \mathcal{H}_d$ is defined by $\braket{\varphi}{\psi} \coloneqq \bra{\varphi} \cdot \ket{\psi}$. 
The Euclidean norm of $\ket{\psi}$ is defined by
\begin{equation}
    \Abs{\ket{\psi}} \coloneqq \sqrt{\braket{\psi}{\psi}} = \sqrt{\sum_{i=1}^d \abs*{a_i}^2}.
\end{equation}

A linear operator on $\mathcal{H}_{d}$ is described by a $d \times d$ complex-valued matrix $A$, with $A^\dag$ its Hermitian conjugate. 
A matrix $A$ is said to be normal if $AA^\dag = A^\dag A$. 
Note that a normal matrix $A$ has the form $A = U \Sigma U^\dag$, where $U$ is a unitary matrix and $\Sigma$ is a diagonal matrix. 
For a function $f \colon \mathbb{C} \to \mathbb{C}$ and a normal matrix $A = U\Sigma U^\dag$ with $\Sigma = \diag\rbra{\lambda_1, \lambda_2, \dots, \lambda_d}$, define $f\rbra{A} \coloneqq U f\rbra{\Sigma} U^\dag$, where $f\rbra{\Sigma} \coloneqq \diag\rbra{f\rbra{\lambda_1}, f\rbra{\lambda_2}, \dots, f\rbra{\lambda_d}}$.
A matrix $A$ is said to be Hermitian, if $A = A^\dag$. 
Note that an Hermitian matrix is a normal matrix. 

For a linear operator $A$ on $\mathcal{H}_{\mathsf{A}} \otimes \mathcal{H}_{\mathsf{B}}$ and a vector $\ket{\psi} \in \mathcal{H}_{\mathsf{A}}$ where $\mathcal{H}_{\mathsf{A}}$ and $\mathcal{H}_{\mathsf{B}}$ are two Hilbert spaces, $A\ket{\psi}$ is a linear operator from $\mathcal{H}_{\mathsf{B}}$ to $\mathcal{H}_{\mathsf{A}} \otimes \mathcal{H}_{\mathsf{B}}$, defined by
\begin{equation}
    A \ket{\psi} \colon \ket{\varphi} \mapsto A \rbra{\ket{\psi} \otimes \ket{\varphi}}.
\end{equation}
Similarly, $\bra{\psi}A$ is a linear operator from $\mathcal{H}_{\mathsf{A}} \otimes \mathcal{H}_{\mathsf{B}}$ to $\mathcal{H}_{\mathsf{B}}$, defined by
\begin{equation}
    \bra{\psi}A \colon \ket{\varphi} \mapsto \sum_{i} \sum_{j} \alpha_{i,j} \braket{\psi}{i} \ket{j}, \text{ if } A \ket{\varphi} = \sum_{i} \sum_j \alpha_{i,j} \rbra*{ \ket{i} \otimes \ket{j} }.
\end{equation}

The operator norm of a matrix $A$ is defined by
\begin{equation}
    \Abs{A} \coloneqq \sup_{\Abs{\ket{\psi}} \leq 1} \Abs{A \ket{\psi}}. 
\end{equation}
The trace norm of a matrix $A$ is defined by 
\begin{equation}
    \Abs{A}_1 \coloneqq \tr\rbra*{\sqrt{A^\dag A}},
\end{equation}
where the trace of a matrix is defined by
\begin{equation}
    \tr\rbra{A} \coloneqq \sum_{i=1}^d \bra{i} A \ket{i}. 
\end{equation}
The matrix H\"older inequality connects between the trace and different norms of matrices. 
\begin{fact}[Matrix H\"older inequality, {\cite[Theorem 2]{Bau11}}] \label{fact:holder}
    For any two $d \times d$ matrices $A$ and $B$, 
    \begin{equation}
        \abs{\tr\rbra{AB}} \leq \Abs{A} \cdot \Abs{B}_1.
    \end{equation}
\end{fact}

A superoperator $\mathcal{E}$ on $\mathcal{H}_d$ is a linear operator on $\mathcal{H}_d \to \mathcal{H}_d$. 
The diamond norm \cite{AKN98} of the superoperator $\mathcal{E}$ is defined by
\begin{equation}
    \Abs{\mathcal{E}}_{\diamond} \coloneqq \sup_{\Abs{X}_1 \leq 1} \Abs*{ \rbra{\mathcal{E} \otimes \mathcal{I}} X }_1,
\end{equation}
where $\mathcal{I}$ is the identity superoperator on $\mathcal{H}_d$. 

\subsection{Quantum computing}

In quantum computing, a $d$-dimensional pure quantum state is described by a vector $\ket{\psi} \in \mathcal{H}_d$. 
In particular, a qubit is described by a vector in $\mathcal{H}_2 \coloneqq \spanspace\cbra{\ket{0}, \ket{1}}$, where
\begin{equation}
    \ket{0} \coloneqq \begin{bmatrix}
        1 \\
        0
    \end{bmatrix}, \qquad \ket{1} \coloneqq \begin{bmatrix}
        0 \\
        1
    \end{bmatrix}.
\end{equation}
A mixed quantum state is described by a density matrix $\rho$ on $\mathcal{H}_d$ satisfying $\tr\rbra{\rho} = 1$ and $\rho \sqsupseteq 0$, where $\sqsupseteq$ is the L\"owner order. 
In particular, the density matrix for the pure quantum state $\ket{\psi}$ is given by $\ketbra{\psi}{\psi}$. 

A quantum gate is described by a unitary matrix $U$ on $\mathcal{H}_d$ with $U^\dag U = U U^\dag = I$, where $I$ is the identity matrix. 
For example, the Hadamard gate is described by the Hadamard matrix of order $2$:
\begin{equation}
    H = \frac{1}{\sqrt{2}} \begin{bmatrix}
        1 & 1 \\
        1 & -1
    \end{bmatrix}. 
\end{equation}
After applying the quantum gate $U$ to the pure quantum state $\ket{\psi}$, the quantum state will become $U \ket{\psi}$. 
After applying the quantum gate $U$ to the mixed quantum state $\rho$, the quantum state will become $U \rho U^\dag$. 
A quantum measurement is described by a set of matrices $\cbra{M_m}$ with the completeness condition 
\begin{equation}
    \sum_{m} M_m^\dag M_m = I.
\end{equation}
In particular, the \mbox{controlled-$U$} gate is denoted by a unitary matrix $\textup{Ctrl}\textup{-}U \coloneqq \diag\rbra{I, U}$ on $\mathcal{H}_{2d}$ such that
\begin{equation}
    \textup{Ctrl}\textup{-}U \colon \ket{b} \otimes \ket{\psi} \mapsto \ket{b} \otimes \rbra*{U^b \ket{\psi}}
\end{equation}
for any $b \in \cbra{0, 1}$ and $\ket{\psi} \in \mathcal{H}_d$.
If we measure a pure quantum state $\ket{\psi}$ with the quantum measurement $\cbra{M_m}$, the measurement outcome will be $X$ such that
\begin{equation}
    \Pr\sbra{X = m} = \Abs*{ M_m \ket{\psi} }^2.
\end{equation}
If we measure a mixed quantum state $\rho$ with the quantum measurement $\cbra{M_m}$, the measurement outcome will be $X$ such that
\begin{equation}
    \Pr\sbra{X = m} = \tr\rbra*{ M_m \rho M_m^\dag }.
\end{equation}
A quantum channel is described by a superoperator $\mathcal{E}$ on $\mathcal{H}_d$ that is completely positive and trace preserving. 
Here, a superoperator $\mathcal{E}$ is said to be completely positive, if $\rbra{\mathcal{E} \otimes \mathcal{I}}\rbra{A} \sqsupseteq 0$ for any identity superoperator $\mathcal{I}$ of any dimension and any matrix $A \sqsupseteq 0$; and, a superoperator $\mathcal{E}$ is said to be trace preserving, if $\tr\rbra{\mathcal{E}\rbra{A}} = \tr\rbra{A}$ for any matrix $A$ on $\mathcal{H}_d$. 
The diamond norm distance between two quantum channels $\mathcal{E}$ and $\mathcal{F}$ is defined by $\Abs{\mathcal{E} - \mathcal{F}}_\diamond$. 

We refer the readers to \cite{NC10} for more details on quantum computing.

In the remaining of this paper, we use ``matrix'' and ``operator'' interchangeably. 
For example, unitary/density operator means unitary/density matrix. 
In addition, we use ``$k$-qubit'' to mean ``$2^k$-dimensional''.
For example, a $k$-qubit quantum state is a $2^k$-dimensional vector in $\mathcal{H}_{2^k}$ and a $k$-qubit (linear) operator means a $2^k$-dimensional matrix on $\mathcal{H}_{2^k}$. 

\subsection{Hadamard test}

The Hadamard test \cite{AJL09} is a useful quantum algorithmic tool that estimates the property of the form $\tr\rbra{U\rho}$, where $U$ is a unitary matrix and $\rho$ is a density matrix. 

\begin{lemma} [Hadamard test, see {\cite[Figure 1]{EAO+02}} and {\cite[Section 2.3]{AJL09}}] \label{lemma:hadamard}
    Let $W_{\mathsf{AB}} = \rbra{H_{\mathsf{A}} \otimes I_{\mathsf{B}}} \cdot \textup{Ctrl}_{\mathsf{A}}\textup{-}U_{\mathsf{B}} \cdot \rbra{H_{\mathsf{A}} \otimes I_{\mathsf{B}}}$ be a unitary operator acting on subsystems $\mathsf{A}$ and $\mathsf{B}$, where $U$ is a unitary operator and $H$ is the Hadamard gate. 
    For every density operator $\rho$ on subsystem $\mathsf{B}$, we have
    \begin{equation}
        \tr\rbra*{ \rbra*{\ketbra{0}{0}_{\mathsf{A}} \otimes I_{\mathsf{B}}} \cdot W_{\mathsf{AB}} \cdot \rbra*{ \ketbra{0}{0}_{\mathsf{A}} \otimes \rho_{\mathsf{B}} } \cdot W_{\mathsf{AB}}^\dag } = \frac{1 + \Real\rbra[\big]{\tr\rbra*{U\rho}}}{2}.
    \end{equation}
\end{lemma}

To apply the Hadamard test in our scenario, we also need the following lemma. 

\begin{lemma} [Adapted from {\cite[Equation (3)]{HMO+21}}] \label{lemma:tr-L-rho-k}
    Let $\textup{Shift}_{k}$ be the unitary operator that shifts $k$ subsystems $\mathsf{A}_1, \mathsf{A}_2, \dots, \mathsf{A}_k$ (of the same dimension) such that 
    \begin{equation}
        \textup{Shift}_k \ket{\psi_1}_{\mathsf{A}_1} \ket{\psi_2}_{\mathsf{A}_2} \dots \ket{\psi_{k-1}}_{\mathsf{A}_{k-1}} \ket{\psi_k}_{\mathsf{A}_k} = \ket{\psi_2}_{\mathsf{A}_1} \ket{\psi_3}_{\mathsf{A}_2} \dots \ket{\psi_k}_{\mathsf{A}_{k-1}} \ket{\psi_1}_{\mathsf{A}_k}
    \end{equation}
    for any $k$ pure states $\ket{\psi_1}, \ket{\psi_2}, \dots, \ket{\psi_k}$. 
    Let $L$ be a linear operator acting on subsystem $\mathsf{A}_1$.
    For every density operator $\rho$ on subsystem $\mathsf{A}_1$, we have
    \begin{equation} \label{eq:shift-L}
    \tr\rbra[\Big]{ \rbra*{L_{\mathsf{A}_1} \otimes I_{\mathsf{A}_2} \dots \otimes I_{\mathsf{A}_k}} \cdot \textup{Shift}_k \cdot \rbra*{\rho_{\mathsf{A}_1} \otimes \rho_{\mathsf{A}_2} \otimes \dots \otimes \rho_{\mathsf{A}_k}} } = \tr\rbra[\Big]{L \rho^k}.
    \end{equation}
\end{lemma}
\begin{proof}
    See \cref{{sec:proof-of-lemma-tr-L-rho-k}}.
\end{proof}

\subsection{Quantum query algorithms with block-encoding}

In this paper, we need the notion of quantum query algorithms with unitary query oracles as block-encodings. 
We first give a formal definition of block-encoding below. 

\begin{definition} [Block-encoding, {\cite[Definition 24]{GSLW19}}]
    An $\rbra{n+a}$-qubit unitary operator $U$ is said to be an $\rbra{\alpha, a, \varepsilon}$-block-encoding of an $n$-qubit operator $A$, if 
    \begin{equation}
        \Abs*{ \alpha \rbra*{\bra{0}^{\otimes a} \otimes I_n} U \rbra*{\ket{0}^{\otimes a} \otimes I_n} - A } \leq \varepsilon,
    \end{equation}
    where $I_n$ is the $n$-qubit identity operator and $\Abs{\cdot}$ is the operator norm. 
\end{definition}

The following is the formal definition of quantum algorithms considered in this paper. 

\begin{definition} [Quantum query algorithms with query access to block-encodings]
    Let $m, n, \ell, Q \geq 1$. 
    An $\rbra{m, n, \ell, Q}$-quantum query algorithm $C$ is an $\ell$-qubit quantum query algorithm $C$ with $m$-ancilla block-encoded access to an $n$-qubit operator with query complexity $Q$ described as a family of $\ell$-qubit quantum circuits of the form 
    \begin{equation}
    C\sbra{U} = G_Q U_Q \cdots G_1 U_1 G_0
    \end{equation}
    for any $\rbra{n+m}$-qubit unitary operator $U$, 
    where each $U_q$ ($1 \leq q \leq Q$) is either \mbox{(controlled-)$U$} or \mbox{(controlled-)$U^\dag$} acting on $\rbra{n+m}$ (out of $\ell$) qubits and $G_0, G_1, \dots, G_Q$ are $\ell$-qubit unitary operators independent of $U$.
    For simplicity, the output of $C\sbra{U}$ is the measurement outcome of the first qubit of $C\sbra{U} \ket{0}^{\otimes \ell}$ in the computational basis.
    In particular,
    \begin{equation}
    \Pr\sbra{C\sbra{U} \textup{ outputs } b} = \Abs*{ M_b \cdot C\sbra{U} \cdot \ket{0}^{\otimes \ell} }^2,
    \end{equation}
    where $M_b = \ketbra{b}{b} \otimes I_{\ell-1}$ for $b \in \cbra{0, 1}$ and $I_{k}$ is the $k$-qubit identity operator. 

    The $\rbra{m, n, \ell, Q}$-quantum query algorithm $C$ is said to be \emph{well-behaved}, if its output depends only on the matrix block-encoded in the query oracle. 
    That is, for any unitary operator $U_A$ that is a $\rbra{1, m, 0}$-block-encoding of an $n$-qubit operator $A$, $\Pr\sbra{C\sbra{U_A} \textup{ outputs } 0}$ depends only on $A$. 
\end{definition}

Throughout this paper, we only consider \textit{well-behaved} quantum query algorithms. 
For simplicity, a quantum query algorithm is said to use $Q$ queries to a unitary oracle $U$, if it uses $Q$ queries to $U$, $U^\dag$, \mbox{controlled-$U$}, and \mbox{controlled-$U^\dag$} in total. 

\subsection{Quantum singular value transformation}

Quantum singular value transformation \cite{GSLW19} is a useful tool for the design of quantum algorithms. 
Here, we need the specific version as follows. 

\begin{lemma}[Quantum singular value transformation, {\cite[Theorem 2]{GSLW19}}] \label{lemma:qsvt}
    Let $p\rbra{x} \in \mathbb{R}\sbra{x}$ be an even/odd polynomial of degree $d$ with $\abs{p\rbra{x}} \leq 1$ for $x \in \sbra{-1, 1}$, where $d \geq 1$. 
    For any $a \geq 1$ and $n \geq 1$, there is an $\rbra{a, n, n+a+1, O\rbra{d}}$-quantum query algorithm $C$ such that for any unitary operator $U_A$ that is a $\rbra{1, a, 0}$-block-encoding of an $n$-qubit Hermitian operator $A$, $C\sbra{U_A}$ is a $\rbra{1, a+1, 0}$-block-encoding of $p\rbra{A}$. 
    Moreover, the implementation of $C\sbra{U_A}$ uses $O\rbra{ad}$ two-qubit gates that are independent of $U_A$. 
\end{lemma}

To apply \cref{lemma:qsvt} for our purpose, we need the following polynomials. 

\begin{lemma} [Polynomial approximation of positive power functions, {\cite[Corollary 3.4.14]{Gil19}}] \label{lemma:poly-approx-pos}
    For $\varepsilon, \delta \in (0, \frac{1}{2}]$ and $c \in \rbra{0, 1}$, there is an even/odd polynomial $p \in \mathbb{R}\sbra{x}$ of degree $O\rbra{\frac{1}{\delta}\log\rbra{\frac{1}{\varepsilon}}}$ such that (i) $\abs{p\rbra{x} - \frac{1}{2}x^c} \leq \varepsilon$ for $x \in \sbra{\delta, 1}$ and (ii) $\abs{p\rbra{x}} \leq 1$ for $x \in \sbra{-1, 1}$. 
    Here, it is noted that the $O\rbra{\cdot}$ hides a constant that is independent of $c$. 
\end{lemma}

\subsection{Quantum samplizer}

Quantum samplizer \cite{WZ25b,WZ25} is a useful tool that converts a quantum query algorithm to a quantum channel that is implemented using samples of quantum states. 
Formally, we consider the following type of quantum channels. 

\begin{definition}[Quantum channels implemented by quantum states] \label{def:qchannel}
    An $\rbra{n, \ell, S}$-quantum channel $\mathcal{E}$ is an $\ell$-qubit quantum channel that can be implemented using $S$ independent samples of the $n$-qubit input quantum state $\rho$. 
    Specifically, the $\rbra{n, \ell, S}$-quantum channel $\mathcal{E}$ is implemented by an $\rbra{Sn+n_{\mathsf{anc}}+\ell}$-qubit unitary operator $W$ for some $n_{\mathsf{anc}} \geq 0$ such that for any $n$-qubit quantum state $\rho$ and any $\ell$-qubit quantum state $\sigma$, 
    \begin{equation}
        \mathcal{E}\sbra{\rho}\rbra{\sigma} = \tr_{\mathsf{env}}\rbra*{ W \rbra*{ \underbrace{\rho^{\otimes S} \otimes \ketbra{0}{0}^{\otimes n_{\mathsf{anc}}}}_{\mathsf{env}} \otimes \sigma } W^\dag }.
    \end{equation}
    This is also illustrated in \cref{fig:qchannel}.
\end{definition}

\begin{figure} [!htp]
    \centering
    \begin{quantikz} [row sep = {25pt, between origins}]
        \lstick[3]{$S$}~ & \lstick{$\rho$} \wireoverride{n} & \qwbundle{n} & \gate[5,disable auto height]{~~W~~} & \meterD{} \\
        & \lstick{$\vdots$} \wireoverride{n} & \qwbundle{n} & & \meterD{} \\
        & \lstick{$\rho$} \wireoverride{n} & \qwbundle{n} & & \meterD{} \\
        & \lstick{$\ket{\bar 0}$} \wireoverride{n} & \qwbundle{n_{\mathsf{anc}}} & & \meterD{} \\
        & \lstick{$\sigma$} \wireoverride{n} & \qwbundle{\ell} & & \rstick{$\mathcal{E}\sbra{\rho}\rbra{\sigma}$}
    \end{quantikz}
    \caption{Quantum circuit for an $\rbra{n, \ell, S}$-quantum channel $\mathcal{E}$.}
    \label{fig:qchannel}
\end{figure}
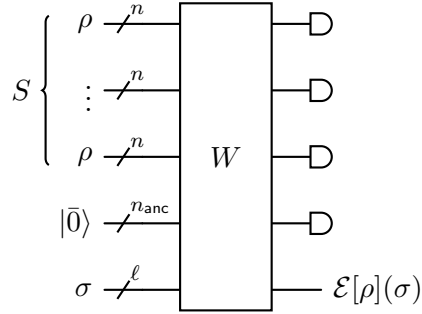

The definition of samplizer is given as follows. 

\begin{definition}[Samplizer, {\cite[Definition I.1]{WZ25}}]
    Let $\delta \in \rbra{0, 1}$ be a precision parameter. 
    An $\rbra{m,n,\ell,\delta}$-samplizer, $\mathsf{Samplize}_{\delta}\ave{*}$, converts an $\rbra{m,n,\ell,Q}$-quantum query algorithm for some $Q \geq 1$ to an $\rbra{n, \ell, S}$-quantum channel for some $S \geq 1$ with the following property. 
    For any $\delta > 0$, any $\rbra{m,n,\ell,Q}$-quantum query algorithm $C$, and any $n$-qubit mixed quantum state $\rho$, there exists a unitary operator $U_\rho$ that is a $\rbra{2, m, 0}$-block-encoding of $\rho$ such that 
    \begin{equation}
        \Abs*{ \mathsf{Samplize}_{\delta}\ave{C}\sbra{\rho} - \rbra{C\sbra{U_\rho}} \rbra{\cdot} \rbra{C\sbra{U_\rho}}^\dag }_{\diamond} \leq \delta,
    \end{equation}
    where $\mathsf{Samplize}_{\delta}\ave{C}$ is an $\rbra{n, \ell, S}$-quantum channel. 
\end{definition}

Below is an implementation of the samplizer given in \cite{WZ25}, which is based on the density matrix exponentiation \cite{LMR14,KLL+17,GKP+25} and is inspired by \cite{GP22,WZ24}. 

\begin{lemma}[An implementation of quantum samplizer, {\cite[Theorem III.1]{WZ25}}] \label{lemma:samplizer}
    For $m \geq 4$, there is an $\rbra{m, n, \ell, \delta}$-samplizer that converts every $\rbra{m, n, \ell, Q}$-quantum query algorithm $C$ to an $\rbra{n, \ell, S}$-quantum channel $\mathsf{Samplize}_{\delta}\ave{C}$ with $S = O\rbra{\frac{Q^2}{\delta}\log^2\rbra{\frac{Q}{\delta}}}$. 
    Moreover, the implementation of $\mathsf{Samplize}_{\delta}\ave{C}$ uses $O\rbra{nS}$ additional two-qubit quantum gates. 
\end{lemma}

\section{Upper Bounds} \label{sec:upper}

In this section, we present an estimator for $\mathrm{F}_{\alpha}\rbra{\rho}$ and thus also for $\mathrm{F}_{\alpha}\rbra{P}$, as stated below. 

\begin{theorem}[Estimator for $\mathrm{F}_\alpha\rbra{P}$ and $\mathrm{F}_{\alpha}\rbra{\rho}$] \label{thm:upper}
    For any real number $\alpha \geq 2$, there is an estimator $\hat{E}_{\mathsf{n}}$ such that
    \begin{equation}
        \sup_{P \in \mathcal{M}_S} \textup{MSE}\rbra{\hat{E}_\mathsf{n}, \mathrm{F}_\alpha,P} 
        \leq \sup_{\rho \in \mathcal{D}\rbra{\mathcal{H}}} \textup{MSE}\rbra{\hat{E}_\mathsf{n}, \mathrm{F}_\alpha,\rho}
        \leq \frac{32\alpha}{\mathsf{n}} + \Theta\rbra*{ \rbra*{\frac{\log^4\rbra{\mathsf{n}}}{\mathsf{n}}}^{\frac{2\floor{\alpha}-2}{3\floor{\alpha}-1}} }.
    \end{equation}
\end{theorem}

In \cref{sec:gswap,sec:subq,sec:sampl}, useful unitary matrices are constructed using known techniques from quantum computing. 
Readers not familiar with quantum computing can skip directly to \cref{sec:simple-estimator,sec:final-estimator} for the main construction of the estimator, assuming the correctness of \cref{sec:gswap,sec:subq,sec:sampl}, where no further knowledge of quantum computing is required. 

\subsection{Generalized SWAP test for block-encodings} \label{sec:gswap}

For our purpose, we present a generalized SWAP test that is useful when we have unitary operators as block-encodings. 

\begin{lemma} [Generalized SWAP test for block-encodings] \label{lemma:swap-block}
    Let $W$ be an $\rbra{n+a}$-qubit unitary operator that is a $\rbra{1, a, 0}$-block-encoding of an $n$-qubit linear operator $B$. 
    Let 
    \begin{equation}
    U \coloneqq H_{\mathsf{C}} \cdot \textup{Ctrl}_{\mathsf{C}}\textup{-}W_{\mathsf{A}_1} \cdot \textup{Ctrl}_{\mathsf{C}}\textup{-}\rbra{\textup{Shift}_{k}}_{\mathsf{A}_1\dots\mathsf{A}_{k}} \cdot H_{\mathsf{C}}
    \end{equation}
    be the unitary operator implemented by the quantum circuit in \cref{fig:swap-intro}, where the subsystem $\mathsf{C}$ consists of $1$ qubit and each subsystem $\mathsf{A}_i$ for $1 \leq i \leq k$ consists of $\rbra{n+a}$ qubits. 
    Then, applying the unitary operator $U$ to the quantum state 
    \begin{equation}
        \sigma \coloneqq \ketbra{0}{0}_{\mathsf{C}} \otimes \rbra{\rho \otimes \ketbra{0}{0}^{\otimes a}}_{\mathsf{A}_1} \otimes \cdots \otimes \rbra{\rho \otimes \ketbra{0}{0}^{\otimes a}}_{\mathsf{A}_k},
    \end{equation}
    followed by a measurement in the computational basis of the subsystem $\mathsf{C}$, we obtain a measurement outcome $x \in \cbra{0, 1}$ such that
    \begin{equation}
        \Pr\sbra{x = 0} = \tr\rbra*{\Pi_0 U \sigma U^\dag \Pi_0^\dag} = \frac{1 + \Real\rbra{\tr\rbra{B \rho^k}}}{2},
    \end{equation}
    where $\Pi_0 = \ketbra{0}{0}_{\mathsf{C}}$. 
\end{lemma}

\begin{proof}
    We describe the subsystem $\mathsf{A}_i$ as two subsystems $\mathsf{X}_i$ and $\mathsf{Y}_i$.
    For example, $\rbra{\rho \otimes \ketbra{\bar 0}{\bar 0}}_{\mathsf{A}_i} = \rho_{\mathsf{X}_i} \otimes \ketbra{\bar 0}{\bar 0}_{\mathsf{Y}_i}$, where we use $\ket{\bar 0} \coloneqq \ket{0}^{\otimes a}$. 
    Then, 
    $W_{\mathsf{A}_1}$ is the unitary operator $W$ acting on the subsystem $\mathsf{A}_1$ such that 
    \begin{equation} \label{eq:def-W}
    \rbra{I_{\mathsf{X}_1} \otimes \bra{\bar 0}_{\mathsf{Y}_1}} \cdot W_{\mathsf{A}_1} \cdot \rbra{I_{\mathsf{X}_1} \otimes \ket{\bar 0}_{\mathsf{Y}_1}} = B_{\mathsf{X}_1}.
    \end{equation}
    By \cref{lemma:hadamard}, the measurement outcome $x$ satisfies
    \begin{align}
        \Pr\sbra{x = 0} 
        & = \tr\rbra*{\Pi_0 U \sigma U^\dag \Pi_0^\dag} \\
        & = \frac{1 + \Real\rbra*{\tr\rbra[\big]{W_{\mathsf{A}_1} \cdot \rbra{\textup{Shift}_{k}}_{\mathsf{A}_1\dots\mathsf{A}_{k}} \cdot \rbra{\rbra{\rho \otimes \ketbra{\bar 0}{\bar 0}}_{\mathsf{A}_1} \otimes \cdots \otimes \rbra{\rho \otimes \ketbra{\bar 0}{\bar 0}}_{\mathsf{A}_k}}}}}{2}. \label{eq:swap-prob}
    \end{align}
    By \cref{lemma:tr-L-rho-k}, we have
    \begin{align}
        \eqref{eq:swap-prob} 
        & = \frac{1 + \Real\rbra*{\tr\rbra*{W_{\mathsf{A}_1} \rbra{\rho \otimes \ketbra{\bar 0}{\bar 0}}_{\mathsf{A}_1}^k}}}{2} \\
        & = \frac{1 + \Real\rbra*{\tr\rbra*{W_{\mathsf{A}_1} \rbra{\rho_{\mathsf{X}_1}^k \otimes \ketbra{\bar 0}{\bar 0}_{\mathsf{Y}_1}}}}}{2}, \label{eq:swap-prob2}
    \end{align}
    where note that $\ketbra{\bar 0}{\bar 0}^{k} = \ketbra{\bar 0}{\bar 0}$. 
    By \cref{eq:def-W}, we have
    \begin{align}
        \eqref{eq:swap-prob2}
        & = \frac{1 + \Real\rbra*{\tr\rbra*{ \rbra*{\rbra{I_{\mathsf{X}_1} \otimes \bra{\bar 0}_{\mathsf{Y}_1}} \cdot W_{\mathsf{A}_1} \cdot \rbra{I_{\mathsf{X}_1} \otimes \ket{\bar 0}_{\mathsf{Y}_1}}} \cdot \rho_{\mathsf{X}_1}^k}}}{2} \\
        & = \frac{1 + \Real\rbra*{\tr\rbra*{B_{\mathsf{X}_1} \cdot \rho_{\mathsf{X}_1}^k}}}{2} = \frac{1 + \Real\rbra*{\tr\rbra*{B \rho^k}}}{2}. 
    \end{align}
\end{proof}

\subsection{Subroutines with queries} \label{sec:subq}

Building on the generalized SWAP test for block-encodings in \cref{lemma:swap-block}, we construct the following quantum subroutine further based on quantum singular value transformation (\cref{lemma:qsvt}). 

\begin{lemma} \label{lemma:subroutine-qsvt}
    Let $p \in \mathbb{R}\sbra{x}$ be an even/odd polynomial of degree $d$ with $\abs{p\rbra{x}} \leq 1$ for $x \in \sbra{-1, 1}$. 
    For any $k \geq 1$, $a \geq 1$ and $n \geq 1$, there is an $\rbra{a, n, k\rbra{n+a+1}+1, O\rbra{d}}$-quantum query algorithm $U_{k,p}$ such that for any unitary operator $U_A$ that is a $\rbra{1, a, 0}$-block-encoding of an $n$-qubit Hermitian operator $A$, applying the unitary operator $U_{k,p}\sbra{U_A}$ to the quantum state 
    \begin{equation} \label{def:sigma}
        \sigma \coloneqq \ketbra{0}{0}_{\mathsf{C}} \otimes \rbra{\rho \otimes \ketbra{0}{0}^{\otimes \rbra{a+1}}}_{\mathsf{A}_1} \otimes \cdots \otimes \rbra{\rho \otimes \ketbra{0}{0}^{\otimes \rbra{a+1}}}_{\mathsf{A}_k} \simeq \rho^{\otimes k} \otimes \ketbra{0}{0}^{\otimes \rbra{k\rbra{a+1}+1}}
    \end{equation}
    with $\rho$ an $n$-qubit quantum state, followed by a measurement in the computational basis of the one-qubit subsytem $C$, we obtain a measurement outcome $x \in \cbra{0, 1}$ such that 
    \begin{equation}
        \Pr\sbra{x = 0} = \tr\rbra*{\Pi_0 \rbra{U_{k,p}\sbra{U_A}} \sigma \rbra{U_{k,p}\sbra{U_A}}^\dag \Pi_0^\dag} = \frac{1 + \Real\rbra{\tr\rbra{p\rbra{A} \rho^k}}}{2},
    \end{equation}
    where $\Pi_0 = \ketbra{0}{0}_{\mathsf{C}}$. 
\end{lemma}

\begin{proof}
    By \cref{lemma:qsvt}, there is an $\rbra{a, n, n+a+1, O\rbra{\deg\rbra{p}}}$-quantum query algorithm $C$ such that for any unitary operator $U_A$ that is a $\rbra{1, a, 0}$-block-encoding of an $n$-qubit Hermitian operator $A$, $C\sbra{U_A}$ is a $\rbra{1, a+1, 0}$-block-encoding of $p\rbra{A}$. 

    By taking $W \coloneqq C\sbra{U_A}$ in \cref{lemma:swap-block} (with $a \coloneqq a + 1$), let 
    \begin{equation}
        U\sbra{U_A} \coloneqq H_{\mathsf{C}} \cdot \textup{Ctrl}_{\mathsf{C}}\textup{-}\rbra{C\sbra{U_A}}_{\mathsf{A}_1} \cdot \textup{Ctrl}_{\mathsf{C}}\textup{-}\rbra{\textup{Shift}_{k}}_{\mathsf{A}_1\dots\mathsf{A}_{k}} \cdot H_{\mathsf{C}},
    \end{equation}
    where $\mathsf{C}$ is a one-qubit subsystem and each $\mathsf{A}_i$ is an $\rbra{n+a+1}$-qubit subsystem.
    Note that $U\sbra{U_A}$ is an $\rbra{a, n, k\rbra{n+a+1}+1, O\rbra{\deg\rbra{p}}}$-quantum query algorithm. 
    By \cref{lemma:swap-block}, applying $U\sbra{U_A}$ to the quantum state $\sigma$,
    followed by a measurement in the computational basis of the subsystem $\mathsf{C}$, we obtain a measurement outcome $x \in \cbra{0, 1}$ such that
    \begin{equation}
        \Pr\sbra{x = 0} = \tr\rbra*{\Pi_0 \rbra{U\sbra{U_A}} \sigma \rbra{U\sbra{U_A}}^\dag \Pi_0^\dag} = \frac{1 + \Real\rbra{\tr\rbra{p\rbra{A} \rho^k}}}{2},
    \end{equation}
    where $\Pi_0 = \ketbra{0}{0}_{\mathsf{C}}$. 
\end{proof}

\subsection{Samplization} \label{sec:sampl}

Building on \cref{lemma:subroutine-qsvt}, we present its samplized version where $A$ is set to $\propto \rho$, based on the quantum samplizer (\cref{lemma:samplizer}).

\begin{lemma} \label{lemma:samplized-test}
    Let $p \in \mathbb{R}\sbra{x}$ be an even/odd polynomial of degree $d$ with $\abs{p\rbra{x}} \leq 1$ for $x \in \sbra{-1, 1}$. 
    For $k \geq 1$, $n \geq 1$ and $\delta \in \rbra{0, 1}$, there is a $\rbra{Tn+\ell}$-qubit unitary matrix $V_{k,p}$ for some $\ell \geq 0$ and
    \begin{equation}
        T = k + O\rbra*{\frac{d^2}{\delta}\log^2\rbra*{\frac{d}{\delta}}},
    \end{equation}
    such that for any $n$-qubit quantum state $\rho$, measuring $V_{k,p} \rbra{\rho^{\otimes T} \otimes \ketbra{0}{0}^{\otimes \ell}} V_{k,p}^\dag$
    in the computational basis of a one-qubit subsystem $\mathsf{C}$ will produce a measurement outcome $x \in \cbra{0, 1}$ with 
    \begin{equation}
        \Pr\sbra{x = 0} = \tr\rbra*{ \Pi_0 V_{k,p} \rbra*{\rho^{\otimes T} \otimes \ketbra{0}{0}^{\otimes \ell}} V_{k,p}^\dag \Pi_0^\dag }
    \end{equation}
    where $\Pi_0 = \ketbra{0}{0}_{\mathsf{C}}$, satisfying
    \begin{equation}
        \abs*{ \Pr\sbra{x = 0} - \frac{1 + \Real\rbra{\tr\rbra{p\rbra{\rho/2} \rho^k}}}{2} } \leq \delta.
    \end{equation}
\end{lemma}

\begin{proof}
    Let $a  = 4$. 
    Let $U_{k,p}\sbra{U_A}$ be the $\rbra{a, n, k\rbra{n+a+1}+1, O\rbra{d}}$-quantum query algorithm given in \cref{lemma:subroutine-qsvt}. 
    Let $\delta \in \rbra{0, 1}$ to be determined.
    By \cref{lemma:samplizer}, there is an $\rbra{n, k\rbra{n+a+1}+1, S}$-quantum channel $\mathcal{E} \coloneqq \mathsf{Samplize}_{\delta}\ave{U_{k, p}}$ with
    \begin{equation}
        S = O\rbra*{\frac{d^2}{\delta}\log^2\rbra*{\frac{d}{\delta}}},
    \end{equation}
    such that there is a unitary operator $U_{\rho/2}$ that is a $\rbra{2, a, 0}$-block-encoding of $\rho$, 
    \begin{equation} \label{eq:diff-diamond-norm}
        \Abs*{\mathcal{E}\sbra{\rho} - \rbra{U_{k, p}\sbra{U_{\rho/2}}} \rbra{\cdot} \rbra{U_{k, p}\sbra{U_{\rho/2}}}^\dag}_\diamond \leq \delta.
    \end{equation}
    
    Let $\sigma$ be defined by \cref{def:sigma}. 
    Then, by \cref{eq:diff-diamond-norm}, it holds that 
    \begin{equation} \label{eq:diff-tr-norm}
        \Abs*{ \mathcal{E}\sbra{\rho}\rbra{\sigma} - \rbra{U_{k, p}\sbra{U_{\rho/2}}} \sigma \rbra{U_{k, p}\sbra{U_{\rho/2}}}^\dag }_1 \leq \delta. 
    \end{equation}
    Measuring $\mathcal{E}\sbra{\rho}\rbra{\sigma}$ in the computational basis of the one-qubit subsystem $\mathsf{C}$ will produce a measurement outcome $x \in \cbra{0, 1}$ such that 
    \begin{equation} \label{eq:prob-samplize}
        \Pr\sbra{x = 0} = \tr\rbra*{\Pi_0 \cdot \mathcal{E}\sbra{\rho}\rbra{\sigma} \cdot \Pi_0^\dag},
    \end{equation}
    where $\Pi_0 = \ketbra{0}{0}_\mathsf{C}$. 
    By \cref{fact:holder}, we have
    \begin{align}
        & \abs*{ \tr\rbra*{\Pi_0 \cdot \mathcal{E}\sbra{\rho}\rbra{\sigma} \cdot \Pi_0^\dag} - \tr\rbra*{\Pi_0 \rbra{U_{k, p}\sbra{U_{\rho/2}}} \sigma \rbra{U_{k, p}\sbra{U_{\rho/2}}}^\dag \Pi_0^\dag} } \\
        = {} & \abs*{ \tr\rbra*{ \Pi_0 \cdot \rbra*{\mathcal{E}\sbra{\rho}\rbra{\sigma} - \rbra{U_{k, p}\sbra{U_{\rho/2}}} \sigma \rbra{U_{k, p}\sbra{U_{\rho/2}}}^\dag} } } \\
        \leq {} & \Abs*{\Pi_0} \cdot \Abs*{\mathcal{E}\sbra{\rho}\rbra{\sigma} - \rbra{U_{k, p}\sbra{U_{\rho/2}}} \sigma \rbra{U_{k, p}\sbra{U_{\rho/2}}}^\dag}_1 \\
        \leq {} & 1 \cdot \delta = \delta, \label{eq:neq-prob-diff}
    \end{align}
    where \cref{eq:neq-prob-diff} is due to \cref{eq:diff-tr-norm} and the fact that $\Abs{\Pi_0} = 1$. 
    On the other hand, by \cref{lemma:subroutine-qsvt}, we have
    \begin{equation} \label{eq:prob-subroutine}
        \tr\rbra*{\Pi_0 \rbra{U_{k,p}\sbra{U_{\rho/2}}} \sigma \rbra{U_{k,p}\sbra{U_{\rho/2}}}^\dag \Pi_0^\dag} = \frac{1 + \Real\rbra{\tr\rbra{p\rbra{\rho/2} \rho^k}}}{2}.
    \end{equation}
    Combining \cref{eq:prob-samplize,eq:neq-prob-diff,eq:prob-subroutine}, we finally have
    \begin{equation} \label{eq:prob-diff}
        \abs*{ \Pr\sbra{x = 0} - \frac{1 + \Real\rbra{\tr\rbra{p\rbra{\rho/2} \rho^k}}}{2} } \leq \delta.
    \end{equation}
    
    Now we recall that $\mathcal{E}$ is an $\rbra{n, k\rbra{n+5}+1, S}$-quantum channel. 
    By \cref{def:qchannel}, there is an $\rbra{Sn+k\rbra{n+5}+1+n_{\mathsf{anc}}}$-qubit unitary matrix $W$ for some $n_{\mathsf{anc}} \geq 0$ such that for any $\rbra{k\rbra{n+5}+1}$-qubit quantum state $\sigma$, 
    \begin{equation}
        \mathcal{E}\sbra{\rho}\rbra{\sigma} \coloneqq \tr_{\mathsf{env}}\rbra*{ W \rbra*{ \underbrace{\rho^{\otimes S} \otimes \ketbra{0}{0}^{\otimes n_{\mathsf{anc}}}}_{\mathsf{env}} \otimes \sigma } W^\dag }.
    \end{equation}
    Note that $\rho^{\otimes S} \otimes \ketbra{0}{0}^{\otimes n_{\mathsf{anc}}} \otimes \sigma \simeq \rho^{\otimes \rbra{S+k}} \otimes \ketbra{0}{0}^{\otimes \rbra{5k+1+n_{\mathsf{anc}}}}$.
    Then, there is a permutation matrix $U_{\pi}$ such that 
    \begin{equation}
        \rho^{\otimes S} \otimes \ketbra{0}{0}^{\otimes n_{\mathsf{anc}}} \otimes \sigma = U_\pi \rbra*{\rho^{\otimes \rbra{S+k}} \otimes \ketbra{0}{0}^{\otimes \rbra{5k+1+n_{\mathsf{anc}}}}} U_\pi^\dag.
    \end{equation}
    Therefore, \cref{eq:prob-samplize} can be written as
    \begin{align}
        \Pr\sbra{x = 0}
        & = \tr\rbra*{\Pi_0 \cdot \tr_{\mathsf{env}}\rbra*{ W \rbra*{ \underbrace{\rho^{\otimes S} \otimes \ketbra{0}{0}^{\otimes n_{\mathsf{anc}}}}_{\mathsf{env}} \otimes \sigma } W^\dag } \cdot \Pi_0^\dag} \\
        & = \tr\rbra*{\Pi_0 \cdot W \rbra*{ \rho^{\otimes S} \otimes \ketbra{0}{0}^{\otimes n_{\mathsf{anc}}} \otimes \sigma } W^\dag \cdot \Pi_0^\dag} \\
        & = \tr\rbra*{\Pi_0 \cdot W U_\pi \rbra*{\rho^{\otimes \rbra{S+k}} \otimes \ketbra{0}{0}^{\otimes \rbra{5k+1+n_{\mathsf{anc}}}}} U_\pi^\dag W^\dag \cdot \Pi_0^\dag}.
    \end{align}
    To complete the proof, together with \cref{eq:prob-diff}, it suffices to take $V_{k,p} \coloneqq W U_\pi$, which gives
    \begin{equation}
        \Pr\sbra{x = 0} = \tr\rbra*{\Pi_0 V_{k,p} \rbra*{\rho^{\otimes T} \otimes \ketbra{0}{0}^{\otimes \ell}} V_{k,p}^\dag \Pi_0^\dag}
    \end{equation}
    where $T = S+k$ and $\ell = 5k+1+n_{\mathsf{anc}}$. 
\end{proof}

\subsection{A baby estimator} \label{sec:simple-estimator}

In this section, we construct an estimator for $\mathrm{F}_\alpha\rbra{\rho}$ (and thus also for $\mathrm{F}_\alpha\rbra{P}$) with its expectation close to $\mathrm{F}_\alpha\rbra{\rho}$, based on the results in \cref{lemma:samplized-test}. 
This ``baby'' estimator can be described by an Hermitian matrix $\mathcal{O}_{\textup{baby}}$, which will be explicitly constructed as follows. 

\paragraph{Step 1: Approximation polynomial of $x^{\cbra{\alpha}}$.}

Let $\varepsilon', \delta' \in (0, \frac{1}{2}]$ be some parameters to be determined. 
By \cref{lemma:poly-approx-pos} with $\varepsilon \coloneq \varepsilon'$, $\delta \coloneq \delta'$, and $c \coloneq \cbra{\alpha}$, there is an even/odd polynomial $p_{\cbra{\alpha},\delta',\varepsilon'}$ of degree $O\rbra{\frac{1}{\delta'}\log\rbra{\frac{1}{\varepsilon'}}}$, where $O\rbra{\cdot}$ hides a constant factor that is independent of $\cbra{\alpha}$, such that (i) $\abs{p_{\cbra{\alpha},\delta',\varepsilon'}\rbra{x} - \frac{1}{2}x^{\cbra{\alpha}}} \leq \varepsilon'$ for $x \in \sbra{\delta', 1}$ and (ii) $\abs{p_{\cbra{\alpha},\delta',\varepsilon'}\rbra{x}} \leq 1$ for $x \in \sbra{-1, 1}$.\footnote{If $\alpha$ is an integer, i.e., $c = 0$, then one can simply set $p_{0, \delta', \varepsilon'}\rbra{x} \equiv \frac{1}{2}$.} 

\paragraph{Step 2: Relate $\tr\rbra{\rho^\alpha}$ to unitary matrices.}

Let $\delta \in \rbra{0, 1}$ to be determined. 
By \cref{lemma:samplized-test} with $n \coloneqq \ceil{\log_2\rbra{S}}$, $k \coloneqq \floor{\alpha}$, $\delta \coloneqq \delta$, and $p \coloneqq p_{\cbra{\alpha},\delta',\varepsilon'}$, there is a $\rbra{T\ceil{\log_2\rbra{S}}+\ell}$-qubit unitary matrix $V_{\floor{\alpha}, p_{\cbra{\alpha},\delta',\varepsilon'}}$ for some $\ell \geq 0$ such that (note that $\tr\rbra{p_{\cbra{\alpha},\delta',\varepsilon'}\rbra{\rho/2} \rho^{\floor{\alpha}}} \in \mathbb{R}$)
\begin{equation} \label{eq:unitary-diff}
    \abs*{ \tr\rbra*{ \Pi_0 V_{\floor{\alpha}, p_{\cbra{\alpha},\delta',\varepsilon'}} \rbra*{\underbrace{\rho^{\otimes T}}_{\mathsf{A}} \otimes \underbrace{\ketbra{0}{0}^{\otimes \ell}}_{\mathsf{B}}} V_{\floor{\alpha}, p_{\cbra{\alpha},\delta',\varepsilon'}}^\dag \Pi_0^\dag } - \frac{1 + \tr\rbra{p_{\cbra{\alpha},\delta',\varepsilon'}\rbra{\rho/2} \rho^{\floor{\alpha}}}}{2} } \leq \delta,
\end{equation}
where 
\begin{align}
    T 
    & = \floor{\alpha} + O\rbra*{\frac{\rbra{\deg\rbra{p_{\cbra{\alpha},\delta',\varepsilon'}}}^2}{\delta}\log^2\rbra*{\frac{\deg\rbra{p_{\cbra{\alpha},\delta',\varepsilon'}}}{\delta}}} \\
    & = \floor{\alpha} + O\rbra*{\frac{1}{\delta'^2 \delta} \log^2\rbra*{\frac{1}{\varepsilon'}} \log^2\rbra*{\frac{\log\rbra{1/\varepsilon'}}{\delta'\delta}}}, \label{def:T}
\end{align}
$\Pi_0 = \ketbra{0}{0} \otimes I_{Tn+\ell-1}$ and $I_{m}$ denotes the $m$-qubit identity operator. 

\paragraph{Step 3: A subnormalized estimator.}

Let
\begin{equation} \label{def:tildeObaby}
    \widetilde{\mathcal{O}}_{\textup{baby}} \coloneqq 2 \underbrace{\bra{0}^{\otimes \ell}}_{\mathsf{B}} V_{\floor{\alpha}, p_{\cbra{\alpha},\delta',\varepsilon'}}^\dag \Pi_0 V_{\floor{\alpha}, p_{\cbra{\alpha},\delta',\varepsilon'}} \underbrace{\ket{0}^{\otimes \ell}}_{\mathsf{B}} - I.
\end{equation}
Then, it can be shown the following property. 
\begin{proposition} \label{prop:baby-diff}
Let $\delta', \varepsilon', \delta, T$ and $\widetilde{\mathcal{O}}_{\textup{baby}}$ defined above. Then, for $\alpha > 1$, 
\begin{equation}
    \abs*{ \tr\rbra*{\widetilde{\mathcal{O}}_{\textup{baby}} \rho^{\otimes T}} - \frac{1}{2^{1+\cbra{\alpha}}} \tr\rbra*{\rho^\alpha} } \leq 2\delta + \varepsilon' + \frac{3}{2} \rbra{2\delta'}^{\floor{\alpha}-1}. 
\end{equation}
\end{proposition}

\paragraph{Step 4: The baby estimator.} 
Let 
\begin{equation} \label{def:Obaby}
    \mathcal{O}_{\textup{baby}} \coloneqq 2^{1+\cbra{\alpha}} \widetilde{\mathcal{O}}_{\textup{baby}}.
\end{equation}
Then, 
\begin{align}
    \abs*{ \tr\rbra*{\mathcal{O}_{\textup{baby}}\rho^{\otimes T}} - \tr\rbra{\rho^\alpha} }
    & = 2^{1+\cbra{\alpha}} \abs*{ \tr\rbra*{\widetilde{\mathcal{O}}_{\textup{baby}} \rho^{\otimes T}} - \frac{1}{2^{1+\cbra{\alpha}}} \tr\rbra*{\rho^\alpha} } \\
    & \leq 4 \rbra*{2\delta + \varepsilon' + \frac{3}{2} \rbra{2\delta'}^{\floor{\alpha}-1}} \\
    & = 8\delta + 4\varepsilon' + 6 \rbra{2\delta'}^{\floor{\alpha}-1}. \label{eq:error-Obaby}
\end{align}

To rigorously define our baby estimator $\hat{E}_{\textup{baby},T}$, it can be seen that $\mathcal{O}_{\textup{baby}}$ defined by \cref{def:Obaby} has the form as of \cref{def:O-intro}, i.e.,
\begin{equation} \label{def:O-baby}
    \mathcal{O}_{\textup{baby}} = 
    \sum_{b \in \cbra{0, 1}} 2^{1+\cbra{\alpha}} \rbra{-1}^{b} 
    \underbrace{\bra{0}^{\otimes \ell}}_{\mathsf{B}} V_{\floor{\alpha}, p_{\cbra{\alpha},\delta',\varepsilon'}}^\dag \Pi_b V_{\floor{\alpha}, p_{\cbra{\alpha},\delta',\varepsilon'}} \underbrace{\ket{0}^{\otimes \ell}}_{\mathsf{B}},
\end{equation}
where $\Pi_1 = I - \Pi_0$. 
For convenience, we rewrite \cref{def:O-baby} using simpler notations as follows:
\begin{equation} \label{eq:m-range}
    \mathcal{O}_{\textup{baby}} = \sum_{m} m \bra{\bar 0}_{\mathsf{B}} U^\dag \Pi_m U \ket{\bar 0}_{\mathsf{B}}.
\end{equation}
where $m$ ranges over $\cbra{\pm 2^{1+\cbra{\alpha}}}$, $U$ is a unitary matrix, and $\cbra{\Pi_m}$ forms a quantum measurement with each $\Pi_m$ a projection matrix. 
The baby estimator $\hat{E}_{\textup{baby},T}$ on input $T$ samples of $\rho$ is defined according to the formula in \cref{eq:def-hatE-q} and it satisfies
\begin{equation} \label{eq:def-baby-estimator}
    \Pr\sbra*{ \hat{E}_{\textup{baby}, T}\rbra{\rho} = m } = \tr\rbra*{ \Pi_m U \rbra[\big]{ \rho^{\otimes T} \otimes \ketbra{\bar 0}{\bar 0}_{\mathsf{anc}} } U^\dag }. 
\end{equation}
It can be seen that $\hat{E}_{\textup{baby}, T}$ is actually an estimator for $\mathrm{F}_{\alpha}\rbra{P}$ and $\mathrm{F}_{\alpha}\rbra{\rho}$ with sample complexity $\mathsf{n} = T$, with the expectation:
\begin{align}
    \E\sbra*{ \hat{E}_{\textup{baby}, T}\rbra{\rho} }
    & = \sum_{m} m \Pr\sbra*{ \hat{E}_{\textup{baby}, T}\rbra{\rho} = m } \\
    & = \sum_{m} m \tr\rbra*{\Pi_m U \rbra[\big]{ \rho^{\otimes T} \otimes \ketbra{\bar 0}{\bar 0}_{\mathsf{anc}} } U^\dag } \\
    & = \tr\rbra*{\mathcal{O}_{\textup{baby}}\rho^{\otimes T}}. \label{eq:E-Obaby}
\end{align}
For completeness, we also explicitly define the estimator for the classical functional $\mathrm{F}_{\alpha}\rbra{P}$ and explain intuition. 
Let $z_1, z_2, \dots, z_{\mathsf{n}}$ be $\mathsf{n}$ independent and identical samples from the distribution $P$. 
Let $\mathbf{z} \coloneqq \rbra{z_1, z_2, \dots, z_{\mathsf{n}}}$ be the shorthand for all the $\mathsf{n}$ samples and let $\ket{\mathbf{z}} \coloneqq \ket{z_1} \otimes \ket{z_2} \otimes \dots \otimes \ket{z_{\mathsf{n}}}$. 
We now construct the baby estimator $\hat{E}_{\textup{baby},T}$ for $\mathrm{F}_\alpha\rbra{P}$ of the form 
\begin{equation}
    \hat{E}_{\textup{baby}, T}\rbra{P} = f_{\mathcal{O}_{\textup{baby}}}\rbra{\mathbf{z}} \coloneqq \tr\rbra*{\mathcal{O}_{\textup{baby}}\ketbra{\mathbf{z}}{\mathbf{z}}}.
\end{equation}
Note that according to the derivation of \cref{eq:fOz=Orhon}, we have 
\begin{equation}
    \E\sbra*{\hat{E}_{\textup{baby}, T}\rbra{P}} = \E_{\mathbf{z} \sim P^{T}} \sbra{ f_{\mathcal{O}_{\textup{baby}}}\rbra{\mathbf{z}} } = \tr\rbra*{\mathcal{O}_{\textup{baby}} \rho_P^{\otimes T}} = \E\sbra*{ \hat{E}_{\textup{baby}, T}\rbra{\rho_P} },
\end{equation}
where
\begin{equation} \label{eq:def-rhoP}
    \rho_P = \sum_{i=1}^S p_i \ketbra{i}{i}.
\end{equation}

To conclude this section, we provide a proof of \cref{prop:baby-diff} below. 

\begin{proof} [Proof of \cref{prop:baby-diff}]
    Note that
    \begin{align}
        \tr\rbra*{\widetilde{\mathcal{O}}_{\textup{baby}} \rho^{\otimes T}}
        & = 2 \tr\rbra*{ \rbra*{\underbrace{\bra{0}^{\otimes \ell}}_{\mathsf{B}} V_{\floor{\alpha}, p_{\cbra{\alpha},\delta',\varepsilon'}}^\dag \Pi_0 V_{\floor{\alpha}, p_{\cbra{\alpha},\delta',\varepsilon'}} \underbrace{\ket{0}^{\otimes \ell}}_{\mathsf{B}}} \cdot \rho^{\otimes T} } - 1 \\
        & = 2 \tr\rbra*{ \Pi_0 V_{\floor{\alpha}, p_{\cbra{\alpha},\delta',\varepsilon'}} \rbra*{\underbrace{\rho^{\otimes T}}_{\mathsf{A}} \otimes \underbrace{\ketbra{0}{0}^{\otimes \ell}}_{\mathsf{B}}} V_{\floor{\alpha}, p_{\cbra{\alpha},\delta',\varepsilon'}}^\dag \Pi_0^\dag } - 1. 
    \end{align}
    By \cref{eq:unitary-diff}, we have 
    \begin{equation} \label{eq:tildeObaby-diff}
        \abs*{ \tr\rbra*{\widetilde{\mathcal{O}}_{\textup{baby}} \rho^{\otimes T}} - \tr\rbra*{p_{\cbra{\alpha},\delta',\varepsilon'}\rbra*{\frac{\rho}{2}} \rho^{\floor{\alpha}}} } \leq 2\delta. 
    \end{equation}

    On the other hand, if the $\ceil{\log_2\rbra{S}}$-qubit density matrix $\rho$ has eigenvalues $\lambda_1, \lambda_2, \dots, \lambda_{2^{\ceil{\log_2\rbra{S}}}}$, then
    \begin{align}
    \abs*{\tr\rbra*{p_{\cbra{\alpha},\delta',\varepsilon'}\rbra*{\frac{\rho}{2}} \rho^{\floor{\alpha}}} - \frac{1}{2^{1+\cbra{\alpha}}} \tr\rbra*{\rho^\alpha}} 
    & = \abs*{\tr\rbra*{p_{\cbra{\alpha},\delta',\varepsilon'}\rbra*{\frac{\rho}{2}}\cdot\rho^{\floor{\alpha}}} - \tr\rbra*{\frac{1}{2}\rbra*{\frac{\rho}{2}}^{\cbra{\alpha}}\cdot\rho^{\floor{\alpha}}}} \\
    & \leq \sum_{j} \abs*{ \lambda_j^{\floor{\alpha}} \cdot p_{\cbra{\alpha},\delta',\varepsilon'}\rbra*{\frac{\lambda_j}{2}} - \lambda_j^{\floor{\alpha}} \cdot \frac{1}{2}\rbra*{\frac{\lambda_j}{2}}^{\cbra{\alpha}} } \\
    & = \sum_{j} \lambda_j \abs*{ \lambda_j^{\floor{\alpha}-1} \cdot p_{\cbra{\alpha},\delta',\varepsilon'}\rbra*{\frac{\lambda_j}{2}} - \lambda_j^{\floor{\alpha}-1} \cdot \frac{1}{2}\rbra*{\frac{\lambda_j}{2}}^{\cbra{\alpha}} }. \label{eq:error-by-eigen}
    \end{align}
    To bound \cref{eq:error-by-eigen}, we split the summation into two cases:
    \begin{equation}
        \sum_{j} = \sum_{j \colon \lambda_j \in \sbra{2\delta', 1}} + \sum_{j \colon \lambda_j \in [0, 2\delta')}.
    \end{equation}

    \paragraph{Case 1: $\lambda_j \in \sbra{2\delta', 1}$.}
    In this case, 
    \begin{equation} \label{eq:case-geq-2delta-diff}
        \abs*{p_{\cbra{\alpha},\delta',\varepsilon'}\rbra*{\frac{\lambda_j}{2}} - \frac{1}{2} \rbra*{\frac{\lambda_j}{2}}^{\cbra{\alpha}}} \leq \varepsilon'.
    \end{equation}
    Therefore, 
    \begin{align}
        \abs*{ \lambda_j^{\floor{\alpha}-1} \cdot p_{\cbra{\alpha},\delta',\varepsilon'}\rbra*{\frac{\lambda_j}{2}} - \lambda_j^{\floor{\alpha}-1} \cdot \frac{1}{2}\rbra*{\frac{\lambda_j}{2}}^{\cbra{\alpha}} } 
        & = \lambda_j^{\floor{\alpha}-1} \abs*{p_{\cbra{\alpha},\delta',\varepsilon'}\rbra*{\frac{\lambda_j}{2}} - \frac{1}{2}\rbra*{\frac{\lambda_j}{2}}^{\cbra{\alpha}}} \\
        & \leq \lambda_j^{\floor{\alpha}-1} \varepsilon' \label{eq:case-geq-2delta} \\
        & \leq \varepsilon', \label{eq:large-lambda}
    \end{align}
    where \cref{eq:case-geq-2delta} uses \cref{eq:case-geq-2delta-diff}. 

    \paragraph{Case 2: $\lambda_j \in [0, 2\delta')$.}
    In this case, 
    \begin{equation}
        \abs*{ p_{\cbra{\alpha},\delta',\varepsilon'}\rbra*{\frac{\lambda_j}{2}} } \leq 1.
    \end{equation}
    Therefore
    \begin{align}
    \abs*{ \lambda_j^{\floor{\alpha}-1} \cdot p_{\cbra{\alpha},\delta',\varepsilon'}\rbra*{\frac{\lambda_j}{2}} - \lambda_j^{\floor{\alpha}-1} \cdot \frac{1}{2}\rbra*{\frac{\lambda_j}{2}}^{\cbra{\alpha}} } 
    & \leq \abs*{ \lambda_j^{\floor{\alpha}-1} \cdot p_{\cbra{\alpha},\delta',\varepsilon'}\rbra*{\frac{\lambda_j}{2}} } + \abs*{ \lambda_j^{\floor{\alpha}-1} \cdot \frac{1}{2}\rbra*{\frac{\lambda_j}{2}}^{\cbra{\alpha}} } \\
    & \leq \lambda_j^{\floor{\alpha}-1} + \frac{1}{2^{\cbra{\alpha}+1}} \lambda_j^{\alpha-1} \\
    & \leq \rbra{2\delta'}^{\floor{\alpha}-1} + \frac{1}{2} \rbra{2\delta'}^{\alpha-1} \\
    & \leq \frac{3}{2} \rbra{2\delta'}^{\floor{\alpha}-1}. \label{eq:small-lambda}
    \end{align}

    Combining \cref{eq:large-lambda,eq:small-lambda}, 
    \begin{align}
        \eqref{eq:error-by-eigen}
        & \leq \sum_{j \colon \lambda_j \in \sbra{2\delta', 1}} \varepsilon' \lambda_j + \sum_{j \colon \lambda_j \in [0, 2\delta')} \frac{3}{2} \rbra{2\delta'}^{\floor{\alpha}-1} \lambda_j \\
        & \leq \sum_{j} \rbra*{\varepsilon' + \frac{3}{2} \rbra{2\delta'}^{\floor{\alpha}-1}} \lambda_j \\
        & = \varepsilon' + \frac{3}{2} \rbra{2\delta'}^{\floor{\alpha}-1}, \label{eq:total-error}
    \end{align}
    where \cref{eq:total-error} uses the fact that $\sum_j \lambda_j = 1$. 

    Combining \cref{eq:tildeObaby-diff,eq:total-error}, by the triangle inequality, we have
    \begin{equation}
        \abs*{ \tr\rbra*{\widetilde{\mathcal{O}}_{\textup{baby}} \rho^{\otimes T}} - \frac{1}{2^{1+\cbra{\alpha}}} \tr\rbra*{\rho^\alpha} } \leq 2\delta + \varepsilon' + \frac{3}{2} \rbra{2\delta'}^{\floor{\alpha}-1}.
    \end{equation}
\end{proof}

\subsection{Final estimator} \label{sec:final-estimator}

Now we are going to construct the final estimator, based on the baby estimator $\hat{E}_{\textup{baby},T}$ defined in \cref{sec:simple-estimator}. 
For $\mathsf{n} \geq T$, we introduce $\mathsf{m} = \floor{\frac{\mathsf{n}}{T}}$ subsystems $\mathsf{Z}_1, \mathsf{Z}_2, \dots, \mathsf{Z}_{\mathsf{m}}$, each with the same dimension of $\mathcal{O}_{\textup{baby}}$. 
The final estimator $\hat{E}_{\mathsf{n}}$ (for $\mathsf{n} \geq T$) is defined as the average value of the output of the baby estimator on each individual subsystem, i.e., 
\begin{equation} \label{eq:def-final-estimator}
    \hat{E}_{\mathsf{n}}\rbra{\rho} \coloneqq \frac{1}{\mathsf{m}} \sum_{i=1}^{\mathsf{m}} \hat{E}_{\textup{baby},T}^{\rbra{i}}\rbra{\rho},
\end{equation}
where $\hat{E}_{\textup{baby},T}^{\rbra{i}}\rbra{\rho}$ means the output (a random variable) of the baby estimator using the $T$ samples of $\rho$ in the subsystem $\mathsf{Z}_{i}$. 
In particular, if $\mathsf{n} < T$, we define $\hat{E}_{\mathsf{n}}$ to be the trivial estimator that always outputs $0$. 
For readability, we also explicitly define the final estimator for $\mathrm{F}_\alpha\rbra{P}$ as follows:
\begin{equation} \label{def:fOz}
    \hat{E}_{\mathsf{n}}\rbra{P} \coloneqq f_{\mathcal{O}}\rbra{\mathbf{z}} \coloneqq \tr\rbra{\mathcal{O} \ketbra{\mathbf{z}}{\mathbf{z}}},
\end{equation}
where
\begin{equation} \label{def:O}
    \mathcal{O} \coloneqq \frac{1}{\mathsf{m}} \sum_{i=1}^{\mathsf{m}} \rbra{\mathcal{O}_{\textup{baby}}}_{\mathsf{Z}_i}
\end{equation}
is an operator acting on $\rho^{\otimes \mathsf{n}}$. 
The expectation and variance of the estimator are given as follows. 
\begin{proposition} \label{prop:E-Var-O}
    Let $\hat{E}_{\mathsf{n}}$ be the estimator defined by \cref{eq:def-final-estimator}.
    Then,
    \begin{align}
        \E \sbra*{\hat{E}_{\mathsf{n}}\rbra{\rho}} 
        & = \tr\rbra*{\mathcal{O}_{\textup{baby}} \rho^{\otimes T}}, \\
        \Var \sbra*{\hat{E}_{\mathsf{n}}\rbra{\rho} } 
        & \leq \frac{16}{\mathsf{m}}.
    \end{align}
    In addition, the classical estimator $\hat{E}_{\mathsf{n}}$ defined by \cref{def:fOz} also holds that 
    \begin{align}
        \E \sbra*{\hat{E}_{\mathsf{n}}\rbra{P}} 
        & = \tr\rbra*{\mathcal{O}_{\textup{baby}} \rho_P^{\otimes T}}, \\
        \Var \sbra*{\hat{E}_{\mathsf{n}}\rbra{P}} 
        & \leq \frac{16}{\mathsf{m}},
    \end{align}
    where $\rho_P$ is defined by \cref{eq:def-rhoP}. 
\end{proposition}

By \cref{prop:E-Var-O}, we can show the MSE of the estimator $\hat{E}_{\mathsf{n}}$ below. 

\begin{theorem}
    Let $\hat{E}_{\mathsf{n}}$ be the estimator defined by \cref{eq:def-final-estimator,def:fOz}.
    For $\alpha \geq 2$, the MSE is bounded by
    \begin{align}
        \textup{MSE}\rbra{\hat{E}_\mathsf{n}, F, \rho} 
        & \leq \frac{32\alpha}{\mathsf{n}} + \Theta\rbra*{ \rbra*{\frac{\log^4\rbra{\mathsf{n}}}{\mathsf{n}}}^{\frac{2\floor{\alpha}-2}{3\floor{\alpha}-1}} }, \\
        \textup{MSE}\rbra{\hat{E}_\mathsf{n}, F, P} 
        & \leq \frac{32\alpha}{\mathsf{n}} + \Theta\rbra*{ \rbra*{\frac{\log^4\rbra{\mathsf{n}}}{\mathsf{n}}}^{\frac{2\floor{\alpha}-2}{3\floor{\alpha}-1}} }. \label{eq:mse-upper-FP}
    \end{align}
\end{theorem}
\begin{proof}
    Applying the results in \cref{prop:E-Var-O} into \cref{eq:mse-q-intro}, we have 
    \begin{align}
        \textup{MSE}\rbra{\hat{E}_\mathsf{n}, F, \rho} 
        & = \frac{16}{\mathsf{m}} + \rbra*{ \tr\rbra*{\mathcal{O}_{\textup{baby}} \rho^{\otimes T}} - \tr\rbra{\rho^{\alpha}} }^2 \\
        & \leq \frac{16}{\mathsf{m}} + \rbra*{8\delta + 4\varepsilon' + 6 \rbra{2\delta'}^{\floor{\alpha}-1}}^2 \label{eq:use-eq:error-Obaby} \\
        & \leq \frac{32T}{\mathsf{n}} + \rbra*{8\delta + 4\varepsilon' + 6 \rbra{2\delta'}^{\floor{\alpha}-1}}^2. \label{eq:use-mTn}
    \end{align}
    where \cref{eq:use-eq:error-Obaby} uses \cref{eq:error-Obaby}, and \cref{eq:use-mTn} uses the fact that $\frac{1}{\mathsf{m}} \leq \frac{2T}{\mathsf{n}}$. 

    By \cref{def:T}, we have
    \begin{equation} \label{eq:mse-bound-by-eps}
    \begin{aligned}
        \eqref{eq:use-mTn} \leq {}
        \frac{32}{\mathsf{n}} \rbra*{\floor{\alpha} + O\rbra*{\frac{1}{\delta'^2 \delta} \log^2\rbra*{\frac{1}{\varepsilon'}} \log^2\rbra*{\frac{\log\rbra{1/\varepsilon'}}{\delta'\delta}}}} + \rbra*{8\delta + 4\varepsilon' + 6 \rbra{2\delta'}^{\floor{\alpha}-1}}^2.
    \end{aligned}
    \end{equation}
    Taking
    \begin{equation}
        \delta' = \rbra*{\frac{\log^4\rbra{\mathsf{n}}}{\mathsf{n}}}^{\frac{1}{3\floor{\alpha}-1}}, \quad \delta = \varepsilon' = \rbra*{\frac{\log^4\rbra{\mathsf{n}}}{\mathsf{n}}}^{\frac{\floor{\alpha}-1}{3\floor{\alpha}-1}}, 
    \end{equation}
    \cref{eq:mse-bound-by-eps} gives 
    \begin{equation}
        \textup{MSE}\rbra{\hat{E}_\mathsf{n}, F, \rho} \leq \frac{32\alpha}{\mathsf{n}} + \Theta\rbra*{ \rbra*{\frac{\log^4\rbra{\mathsf{n}}}{\mathsf{n}}}^{\frac{2\floor{\alpha}-2}{3\floor{\alpha}-1}} }.
    \end{equation}
    To complete the proof, \cref{eq:mse-upper-FP} can be obtained similarly. 
\end{proof}

To conclude this section, we provide a proof of \cref{prop:E-Var-O} below. 
\begin{proof}[Proof of \cref{prop:E-Var-O}]
    By \cref{eq:E-Obaby}, the expectation of $\hat{E}_{\mathsf{n}}\rbra{\rho}$ is given as
    \begin{align}
        \E \sbra*{\hat{E}_{\mathsf{n}}\rbra{\rho}} 
        & = \frac{1}{\mathsf{m}} \sum_{i=1}^{\mathsf{m}}\E\sbra*{\hat{E}_{\textup{baby}, T}^{\rbra{i}}\rbra{\rho}} \\
        & = \E\sbra*{\hat{E}_{\textup{baby}, T}\rbra{\rho}} = \tr\rbra*{\mathcal{O}_{\textup{baby}}\rho^{\otimes T}}.
    \end{align}
    
    For the variance, we first note that each $\hat{E}_{\textup{baby}, T}^{\rbra{i}}\rbra{\rho}$ is independent and identical. 
    Therefore,
    \begin{align}
        \Var\sbra*{\hat{E}_{\mathsf{n}}\rbra{\rho}} 
        & = \frac{1}{\mathsf{m}^2} \sum_{i=1}^{\mathsf{m}} \Var\sbra*{\hat{E}_{\textup{baby}, T}^{\rbra{i}}\rbra{\rho}} \\
        & = \frac{1}{\mathsf{m}} \Var\sbra*{\hat{E}_{\textup{baby}, T}\rbra{\rho}}. \label{eq:Var-eq}
    \end{align}
    The variance of $\hat{E}_{\textup{baby}, T}\rbra{\rho}$ can be bounded by
    \begin{align}
        \Var\sbra*{\hat{E}_{\textup{baby}, T}\rbra{\rho}}
        & \leq \E\sbra*{\rbra*{\hat{E}_{\textup{baby}, T}\rbra{\rho}}^2} \\
        & = \sum_{m} m^2 \Pr\sbra*{ \hat{E}_{\textup{baby}, T}\rbra{\rho} = m } \\
        & \leq 2^{2\rbra{1 + \cbra{\alpha}}} \label{eq:var-leq} \\
        & \leq 16, \label{eq:var-leq-16}
    \end{align}
    where \cref{eq:var-leq} is because $m$ ranges over $\cbra{\pm 2^{1 + \cbra{\alpha}}}$ in \cref{eq:m-range}.
    Therefore, combining \cref{eq:Var-eq,eq:var-leq-16}, we have
    \begin{equation}
        \Var\sbra*{\hat{E}_{\mathsf{n}}\rbra{\rho}} \leq \frac{16}{\mathsf{m}}.
    \end{equation}
\end{proof}

\section{Lower Bounds}

In this section, we show a lower bound on the minimax MSE risk for estimating $\mathrm{F}_\alpha\rbra{P}$, as stated below. 

\begin{theorem} [Minimax MSE lower bound for estimating $\mathrm{F}_\alpha\rbra{P}$]  \label{thm:lower}
    For any real number $\alpha \geq 4$ and any integer $\mathsf{n} \geq 1$, 
    \begin{equation}
    \inf_{\hat{E}_\mathsf{n}} \sup_{P \in \mathcal{M}_S} \textup{MSE}\rbra{\hat{E}_\mathsf{n}, \mathrm{F}_\alpha,P} \geq \frac{1}{4e^3} \cdot \min\cbra*{\frac{\alpha-1}{4\mathsf{n}}, 1}.
    \end{equation}
\end{theorem}

This is achieved by a lower bound on the sample complexity of estimating $\mathrm{F}_\alpha\rbra{P}$ given in \cref{sec:sample-lower}. 
Then, the proof of \cref{thm:lower} is given in \cref{sec:mse-lower}. 

\subsection{Sample complexity lower bounds} \label{sec:sample-lower}

We first prove a lower bound on the sample complexity of estimating $\mathrm{F}_\alpha\rbra{P}$. 

\begin{theorem}[Sample lower bound for estimating $\mathrm{F}_\alpha\rbra{P}$] \label{thm:sample-lb}
    For $\alpha \geq 4$, $\varepsilon \in \rbra{0, \frac{1}{e}}$, and $\delta \in (0, \frac{1}{4}]$, any algorithm that estimates $\mathrm{F}_\alpha\rbra{P}$ for an unknown discrete distribution $P$ to within additive error $\varepsilon$ with success probability $\geq 1 - \delta$ requires more than 
    \begin{equation}
        \frac{\alpha-1}{4e^2\varepsilon^2} \ln\rbra*{\frac{1}{4\delta}}
    \end{equation}
    samples drawn from the unknown distribution. 
\end{theorem}

To this end, we need the following sample lower bound for hypothesis testing. 

\begin{theorem} [Sample lower bound for hypothesis testing, {\cite[Theorem 4.7]{BY02}}] \label{thm:hellinger}
    Suppose that $P$ and $Q$ are two discrete distributions such that $d_{\mathrm{H}}^2\rbra{P, Q} \leq \frac{1}{2}$. 
    Then, for $\delta \in (0, \frac{1}{4}]$, determining whether an unknown discrete distribution is $P$ or $Q$ (promised that it is in either case) with success probability $\geq 1 - \delta$ requires more than
    \begin{equation}
        \frac{1}{4d_{\mathrm{H}}^2\rbra{P, Q}} \ln\rbra*{\frac{1}{4\delta}}
    \end{equation}
    samples drawn from the unknown distribution. 
\end{theorem}

Now we give a proof of \cref{thm:sample-lb} below.

\begin{proof}[Proof of \cref{thm:sample-lb}]
    We use the hard instance in \cite[Theorem 4.3]{CWYZ25}. 
    Consider the two discrete distributions $P^{\pm}$ with alphabet size $S = 2$ such that
    \begin{equation}
    p^{\pm}_1 = 1 - \frac{1}{\alpha} \pm \gamma, \qquad p^{\pm}_2 = \frac{1}{\alpha} \mp \gamma,
    \end{equation}
    where $\gamma \in \rbra{0, \frac{1}{\alpha}}$ is to be determined. 
    On the one hand, the Hellinger distance between $P^{\pm}$ can be bounded by (using the same argument in the proof of \cite[Theorem 5.3]{Wan25})
    \begin{equation}
    d_{\mathrm{H}}\rbra{P^+, P^-} \leq \frac{\alpha\gamma}{\sqrt{\alpha-1}}.
    \end{equation}
    On the other hand, the difference between $\mathrm{F}_\alpha\rbra{P^\pm}$ is
    \begin{equation}
    \mathrm{F}_\alpha\rbra{P^+} - \mathrm{F}_\alpha\rbra{P^-} = \rbra*{1-\frac{1}{\alpha}+\gamma}^\alpha - \rbra*{1-\frac{1}{\alpha}-\gamma}^\alpha + \rbra*{\frac{1}{\alpha}-\gamma}^\alpha - \rbra*{\frac{1}{\alpha}+\gamma}^\alpha.
    \end{equation}
    It can be shown (later in \cref{lemma:F-alpha-diff-lb}) that for any real number $\alpha \geq 4$, we have
    \begin{equation}
    \mathrm{F}_\alpha\rbra{P^+} - \mathrm{F}_\alpha\rbra{P^-} \geq \frac{2}{e}\alpha\gamma.
    \end{equation}
    Taking $\gamma = e\varepsilon/\alpha$, any estimator for $\mathrm{F}_\alpha\rbra{P}$ to additive error $\varepsilon$ with success probability $\geq 1-\delta$ can be used to determine whether an unknown distribution is $P^+$ or $P^-$ with the same success probability. 
    By \cref{thm:hellinger}, this requires more than 
    \begin{equation}
    \frac{1}{4d_{\mathrm{H}}^2\rbra{P^+, P^-}} \ln\rbra*{\frac{1}{4\delta}} \geq \frac{\alpha-1}{4e^2\varepsilon^2} \ln\rbra*{\frac{1}{4\delta}}
    \end{equation}
    samples drawn from the unknown distribution $P$. 
\end{proof}

To conclude this section, we will give a proof of the following lemma. 

\begin{lemma} \label{lemma:F-alpha-diff-lb}
    For $\alpha \geq 4$ and $\gamma \in \rbra{0, \frac{1}{\alpha}}$, 
    \begin{equation}
    \rbra*{1-\frac{1}{\alpha}+\gamma}^\alpha - \rbra*{1-\frac{1}{\alpha}-\gamma}^\alpha + \rbra*{\frac{1}{\alpha}-\gamma}^\alpha - \rbra*{\frac{1}{\alpha}+\gamma}^\alpha \geq \frac{2}{e}\alpha\gamma.
    \end{equation}
\end{lemma}

To this end, we list \cref{lemma:neq-pqx,lemma:neq-lnx,lemma:lb-alpha-e} in the following together with their proofs. 

\begin{lemma} \label{lemma:neq-pqx}
    For $\beta \geq 3$ and $p \geq q \geq x \geq 0$, we have
    \begin{equation}
    \rbra{p+x}^\beta + \rbra{q-x}^\beta - \rbra{p-x}^\beta - \rbra{q+x}^\beta \geq 2\beta\rbra{p^{\beta-1}-q^{\beta-1}}x.
    \end{equation}
\end{lemma}
\begin{proof}
    Let
    \begin{equation}
    f\rbra{x} \coloneqq \rbra{p+x}^\beta + \rbra{q-x}^\beta - \rbra{p-x}^\beta - \rbra{q+x}^\beta - 2\beta\rbra{p^{\beta-1}-q^{\beta-1}}x.
    \end{equation}
    Then, we only have to show that $f\rbra{x} \geq 0$ for $x \in \sbra{0, q}$. 
    To this end, note that
    \begin{align}
    f'\rbra{x} 
    & = \beta \rbra*{ \rbra{p+x}^{\beta-1} - \rbra{q-x}^{\beta-1} + \rbra{p-x}^{\beta-1} - \rbra{q+x}^{\beta-1} - 2\rbra{p^{\beta-1}-q^{\beta-1}} } \\
    & = \beta \rbra*{ g\rbra{x} + g\rbra{-x} - 2\rbra{p^{\beta-1}-q^{\beta-1}} },
    \end{align}
    where
    \begin{equation}
    g\rbra{x} \coloneqq \rbra{p+x}^{\beta-1} - \rbra{q+x}^{\beta-1}. 
    \end{equation}
    Note that
\begin{align}
    g'\rbra{x} & = \rbra{\beta-1} \rbra*{ \rbra{p+x}^{\beta-2} - \rbra{q+x}^{\beta-2} }, \\
    g''\rbra{x} & = \rbra{\beta-1}\rbra{\beta-2}\rbra*{\rbra{p+x}^{\beta-3}-\rbra{q+x}^{\beta-3}}.
\end{align}
When $\beta \geq 3$, we have $\rbra{p+x}^{\beta-3} \geq \rbra{q+x}^{\beta-3}$ for all $p \geq q \geq 0$ and $\abs{x} \leq q$. 
Therefore, $g''\rbra{x} \geq 0$ when $\beta \geq 3$, $p \geq q \geq 0$, and $\abs{x} \leq q$. 
This means that $g\rbra{x}$ is convex for $\abs{x} \leq q$.
Therefore,
\begin{equation}
    \frac{g\rbra{x} + g\rbra{-x}}{2} \geq g\rbra{0} = p^{\beta-1} - q^{\beta-1},
\end{equation}
which gives $f'\rbra{x} \geq 0$ for $x \in \sbra{0, q}$. 
Therefore, for any $x \in \sbra{0, q}$, we have $f\rbra{x} \geq f\rbra{0} = 0$, which yields the proof. 
\end{proof}

\begin{lemma} \label{lemma:neq-lnx}
    For $x > 1$, 
    \begin{equation}
    1 + x \ln\rbra*{1 - \frac{1}{x}} \leq -\frac{1}{2x}.
    \end{equation}
\end{lemma}
\begin{proof}
    Let $t = \frac{1}{x} \in \rbra{0, 1}$. 
    Then, we only have to show that for $t \in \rbra{0, 1}$, 
    \begin{equation}
    g\rbra{t} = t^2 + 2t + 2\ln\rbra{1-t} \leq 0. 
    \end{equation}
    Note that for $t \in [0, 1)$, 
    \begin{equation}
    g'\rbra{t} = -\frac{2t^2}{1-t} \leq 0.
    \end{equation}
    Therefore, $g\rbra{t}$ is non-increasing for $t \in [0, 1)$, which gives $g\rbra{t} \leq g\rbra{0} = 0$ for $t \in [0, 1)$.
\end{proof}

\begin{lemma} \label{lemma:lb-alpha-e}
    For $\alpha \geq 4$, 
    \begin{equation}
    \rbra*{1-\frac{1}{\alpha}}^{\alpha-1} - \rbra*{\frac{1}{\alpha}}^{3} \geq \frac{1}{e}. 
    \end{equation}
\end{lemma}
\begin{proof}
    Let
    \begin{equation}
    f\rbra{x} = \rbra*{1-\frac{1}{x}}^{x-1} - \rbra*{\frac{1}{x}}^{3}.
    \end{equation}
    We only have to show that for any $x \geq 4$, $f\rbra{x} \geq 1/e$. 
    To this end, note that
    \begin{equation}
    f'\rbra{x} = \frac{1}{x^4}\rbra*{3 + \rbra*{1-\frac{1}{x}}^{x-1}x^3\rbra*{1 + x \ln\rbra*{1-\frac{1}{x}}}}.
    \end{equation}
    Note that by \cref{lemma:neq-lnx}, for $x > 1$, 
    \begin{equation}
    1 + x \ln\rbra*{1 - \frac{1}{x}} \leq -\frac{1}{2x}.
    \end{equation}
    Then, for $x > 1$, 
    \begin{align}
        f'\rbra{x} 
        & \leq \frac{1}{x^4}\rbra*{3 + \rbra*{1-\frac{1}{x}}^{x-1} x^3 \cdot \rbra*{-\frac{1}{2x}}} \\
        & = \frac{1}{x^4}\rbra*{3 - \frac{1}{2} \rbra*{1-\frac{1}{x}}^{x-1} x^2}.
    \end{align}
    Note that $\rbra{1-\frac{1}{x}}^{x-1} x^2$ is increasing for $x \geq 4$ and thus for $x \geq 4$, 
    \begin{equation}
        \rbra*{1-\frac{1}{x}}^{x-1} x^2 \geq \left. \rbra*{1-\frac{1}{x}}^{x-1} x^2 \right|_{x=4} = \frac{27}{4}.
    \end{equation}
    Therefore, for $x \geq 4$, 
    \begin{align}
        f'\rbra{x} \leq \frac{1}{x^4}\rbra*{3 - \frac{1}{2} \cdot \frac{27}{4}}
        = -\frac{3}{8x^4} < 0,
    \end{align}
    meaning that $f\rbra{x}$ is decreasing for $x \geq 4$. 
    Therefore, for any $x \geq 4$, we have
    \begin{equation}
    f\rbra{x} \geq \lim_{x \to +\infty} f\rbra{x} = \frac{1}{e}. 
    \end{equation}
\end{proof}

Now we are ready to prove \cref{lemma:F-alpha-diff-lb}. 

\begin{proof}[Proof of \cref{lemma:F-alpha-diff-lb}]
    By \cref{lemma:neq-pqx} with $p \coloneq 1 - \frac{1}{\alpha}$, $q \coloneq \frac{1}{\alpha}$, $\beta \coloneq \alpha$, and $x \coloneq \gamma$, we have
    \begin{equation}
    \rbra*{1-\frac{1}{\alpha}+\gamma}^\alpha - \rbra*{1-\frac{1}{\alpha}-\gamma}^\alpha + \rbra*{\frac{1}{\alpha}-\gamma}^\alpha - \rbra*{\frac{1}{\alpha}+\gamma}^\alpha \geq 2\alpha\rbra*{\rbra*{1-\frac{1}{\alpha}}^{\alpha-1} - \rbra*{\frac{1}{\alpha}}^{\alpha-1}}\gamma.
    \end{equation}
    To complete the proof, by \cref{lemma:lb-alpha-e}, we have for $\alpha \geq 4$, 
    \begin{equation}
    \rbra*{1-\frac{1}{\alpha}}^{\alpha-1} - \rbra*{\frac{1}{\alpha}}^{\alpha-1} \geq \rbra*{1-\frac{1}{\alpha}}^{\alpha-1} - \rbra*{\frac{1}{\alpha}}^{3} \geq \frac{1}{e}.
    \end{equation}
\end{proof}

\subsection{MSE lower bounds} \label{sec:mse-lower}

Now we are going to prove \cref{thm:lower}. 

\begin{proof}[Proof of \cref{thm:lower}]
    Let 
    \begin{align}
        L_\alpha\rbra{\mathsf{n}} 
        & \coloneqq \inf_{\hat{E}_\mathsf{n}} \sup_{P \in \mathcal{M}_S} \textup{MSE}\rbra{\hat{E}_\mathsf{n}, \mathrm{F}_\alpha,P} \\
        & = \inf_{\hat{E}_\mathsf{n}} \sup_{P \in \mathcal{M}_S} \E\sbra*{\rbra*{\hat{E}_\mathsf{n}\rbra{P} - \mathrm{F}_\alpha\rbra{P}}^2}.
    \end{align}
    Let $\delta \in (0, \frac{1}{4}]$ and $\varepsilon \in \rbra{0, \frac{1}{e}}$ to be determined. 
    Suppose that 
    \begin{equation}
    \mathsf{n} \leq \frac{\alpha-1}{4e^2\varepsilon^2} \ln\rbra*{\frac{1}{4\delta}}.
    \end{equation}
    By \cref{thm:sample-lb}, every estimator $\hat E_{\mathsf{n}}$ fails with probability greater than $\delta$ on a probability distribution $P^*$, i.e., 
    \begin{equation}
    \Pr\sbra*{ \abs*{\hat E_{\mathsf{n}}\rbra{P^*} - \mathrm{F}_\alpha\rbra{P^*}} \geq \varepsilon } > \delta.
    \end{equation}
    Markov's inequality gives 
    \begin{equation}
    \Pr\sbra*{ \abs*{\hat E_{\mathsf{n}}\rbra{P^*} - \mathrm{F}_\alpha\rbra{P^*}} \geq \varepsilon } \leq \frac{1}{\varepsilon^2} \mathbb{E}\sbra*{\rbra*{\hat E_{\mathsf{n}}\rbra{P^*} - \mathrm{F}_\alpha\rbra{P^*}}^2},
    \end{equation}
    and hence 
    \[
    \sup_{P \in \mathcal{M}_S} \E\sbra*{\rbra*{\hat{E}_\mathsf{n}\rbra{P} - \mathrm{F}_\alpha\rbra{P}}^2} \geq \mathbb{E}\sbra*{\rbra*{\hat E_{\mathsf{n}}\rbra{P^*} - \mathrm{F}_\alpha\rbra{P^*}}^2} > \delta \varepsilon^2.
    \]
    Taking the infimum over all estimators $\hat{E}_\mathsf{n}$ gives 
    \begin{equation}
    L_\alpha\rbra{\mathsf{n}} \geq \delta\varepsilon^2.
    \end{equation}

    Now we consider the following two cases.

    \paragraph{Case 1: $\mathsf{n} > \frac{\alpha-1}{4}$.}
    Taking 
    \begin{equation}
    \delta = \frac{1}{4e} \in \left(0, \frac{1}{4}\right], \qquad \varepsilon = \frac{1}{e}\sqrt{\frac{\alpha-1}{4\mathsf{n}}} \in \rbra*{0, \frac{1}{e}},
    \end{equation}
    we have
    \begin{equation}
    L_\alpha\rbra{\mathsf{n}} \geq \delta \varepsilon^2 = \frac{\alpha-1}{16e^3\mathsf{n}}.
    \end{equation}

    \paragraph{Case 2: $1 \leq \mathsf{n} \leq \frac{\alpha-1}{4}$.} 
    Taking $\delta = \frac{1}{4e} \in (0, \frac{1}{4}]$, we have 
    \begin{equation}
    L_\alpha\rbra{\mathsf{n}} \geq \delta \varepsilon^2 = \frac{\varepsilon^2}{4e}.
    \end{equation}
    As $\varepsilon \in \rbra{0, \frac{1}{e}}$ can be arbitrarily close to $\frac{1}{e}$, we have 
    \begin{equation}
        L_\alpha\rbra{\mathsf{n}} \geq \frac{1}{16e^3}.
    \end{equation}

    Combining the two cases, we have
    \begin{equation}
        L_\alpha\rbra{\mathsf{n}} \geq \frac{1}{4e^3} \cdot \min\cbra*{\frac{\alpha-1}{4\mathsf{n}}, 1},
    \end{equation}
    which completes the proof. 
\end{proof}

\section*{Acknowledgment}

The author thanks Kean Chen for helpful discussions, and anonymous reviewers for pointing out the related work \cite{DdW11}.

\addcontentsline{toc}{section}{References}

\bibliographystyle{alphaurl}
\bibliography{main}

@PREAMBLE{"\DeclareRobustCommand{\dutchPrefix}[2]{#2}"}

@PREAMBLE{"\providecommand{\dutchPrefix}[2]{#2}"}

@PREAMBLE{"\renewcommand{\dutchPrefix}[2]{#2}"}

@article{AJL09,
    author = {Aharonov, Dorit and Jones, Vaughan and Landau, Zeph},
    title = {A polynomial quantum algorithm for approximating the {Jones} polynomial},
    journal = {Algorithmica},
    volume = {55},
    number = {3},
    pages = {395--421},
    doi = {10.1007/s00453-008-9168-0},
    year = {2009}
}

@article{EAO+02,
    author = {Ekert, Artur K. and Alves, Carolina Moura and Oi, Daniel K. L. and Horodecki, Micha{\l} and Horodecki, Pawe{\l} and Kwek, L. C.},
    title = {Direct estimations of linear and nonlinear functionals of a quantum state},
    journal = {Physical Review Letters},
    volume = {88},
    number = {21},
    pages = {217901},
    doi = {10.1103/PhysRevLett.88.217901},
    year = {2002}
}

@article{HMO+21,
    author = {Huggins, William J. and McArdle, Sam and O'Brien, Thomas E. and Lee, Joonho and Rubin, Nicholas C. and Boixo, Sergio and Whaley, K. Birgitta and Babbush, Ryan and McClean, Jarrod R.},
    title = {Virtual distillation for quantum error mitigation},
    journal = {Physical Review X},
    volume = {11},
    number = {4},
    pages = {041036},
    doi = {10.1103/PhysRevX.11.041036},
    year = {2021}
}

@article{JVHW17,
    author = {Jiao, Jiantao and Venkat, Kartik and Han, Yanjun and Weissman, Tsachy},
    title = {Maximum likelihood estimation of functionals of discrete distributions},
    journal = {IEEE Transactions on Information Theory},
    volume = {63},
    number = {10},
    pages = {6774--6798},
    doi = {10.1109/TIT.2017.2733537},
    year = {2017}
}

@phdthesis{Gil19,
    type = {{PhD} Thesis},
    title = {Quantum Singular Value Transformation \& Its Algorithmic Applications},
    author = {Gily\'{e}n, Andr\'{a}s},
    school = {University of Amsterdam},
    url = {https://pure.uva.nl/ws/files/35292358/Thesis.pdf},
    year = {2019},
}

@inproceedings{GSLW19,
    author = {Gily\'{e}n, Andr\'{a}s and Su, Yuan and Low, Guang Hao and Wiebe, Nathan},
    title = {Quantum singular value transformation and beyond: exponential improvements for quantum matrix arithmetics},
    booktitle = {Proceedings of the 51st Annual ACM SIGACT Symposium on Theory of Computing},
    pages = {193-204},
    doi = {10.1145/3313276.3316366},
    year = {2019}
}

@article{WZ25,
    author = {Wang, Qisheng and Zhang, Zhicheng},
    title = {Time-efficient quantum entropy estimator via samplizer},
    journal = {IEEE Transactions on Information Theory},
    volume = {71},
    number = {12},
    pages = {9569--9599},
    doi = {10.1109/TIT.2025.3576137},
    year = {2025},
}

@article{LW25,
    author = {Liu, Yupan and Wang, Qisheng},
    title = {On estimating the trace of quantum state powers},
    journal = {IEEE Transactions on Information Theory},
    volume = {},
    number = {},
    pages = {},
    doi = {10.1137/1.9781611978322.28},
    year = {2026},
}

@inproceedings{CW25,
    author = {Chen, Kean and Wang, Qisheng},
    title = {Improved Sample Upper and Lower Bounds for Trace Estimation of Quantum State Powers},
    booktitle = {Proceedings of the 38th Annual Conference on Learning Theory},
    pages = {1008--1028},
    doi = {},
    url = {https://proceedings.mlr.press/v291/chen25d.html},
    year = {2025}
}

@phdthesis{BY02,
    type = {{PhD} thesis},
    title = {The Complexity of Massive Data Set Computations},
    author = {Bar-Yossef, Ziv},
    school = {University of California, Berkeley},
    url = {https://www.proquest.com/docview/304791145},
    year = {2002},
}

@article{CWYZ25,
    author = {Chen, Kean and Wang, Qisheng and Yu, Zhan and Zhang, Zhicheng},
    title = {Simultaneous estimation of nonlinear functionals of a quantum state},
    journal = {IEEE Transactions on Information Theory},
    volume = {},
    number = {},
    pages = {},
    doi = {10.1109/TIT.2026.3699531},
    year = {2026}
}

@article{Sha48a,
    author = {Shannon, C. E.},
    title = {A mathematical theory of communication},
    journal = {The Bell System Technical Journal},
    volume = {27},
    number = {3},
    pages = {379--423},
    doi = {10.1002/j.1538-7305.1948.tb01338.x},
    year = {1948}
}

@article{Sha48b,
    author = {Shannon, C. E.},
    title = {A mathematical theory of communication},
    journal = {The Bell System Technical Journal},
    volume = {27},
    number = {4},
    pages = {623--656},
    doi = {10.1002/j.1538-7305.1948.tb00917.x},
    year = {1948}
}

@article{Pan03,
    author = {Paninski, Liam},
    title = {Estimation of entropy and mutual information},
    journal = {Neural Computation},
    volume = {15},
    number = {6},
    pages = {1191--1253},
    doi = {10.1162/089976603321780272},
    year = {2003}
}

@article{Pan04,
    author = {Paninski, Liam},
    title = {Estimating entropy on $m$ bins given fewer than $m$ samples},
    journal = {IEEE Transactions on Information Theory},
    volume = {50},
    number = {9},
    pages = {2200--2203},
    doi = {10.1109/TIT.2004.833360},
    year = {2004}
}

@inproceedings{VV11a,
    author = {Valiant, Gregory and Valiant, Paul},
    title = {Estimating the unseen: an n/log(n)-sample estimator for entropy and support size, shown optimal via new {CLTs}},
    booktitle = {Proceedings of the 43rd Annual ACM Symposium on Theory of Computing},
    pages = {685--694},
    doi = {10.1145/1993636.1993727},
    year = {2011}
}

@inproceedings{VV11b,
    author = {Valiant, Gregory and Valiant, Paul},
    title = {The power of linear estimators},
    booktitle = {Proceedings of the 52nd IEEE Annual Symposium on Foundations of Computer Science},
    pages = {403-412},
    doi = {10.1109/FOCS.2011.81},
    year = {2011}
}

@article{VV17,
    author = {Valiant, Gregory and Valiant, Paul},
    title = {Estimating the unseen: improved estimators for entropy and other properties},
    journal = {Journal of the ACM},
    volume = {64},
    number = {6},
    pages = {37:1--37:41},
    doi = {10.1145/3125643},
    year = {2017}
}

@article{JVHW15,
    author = {Jiao, Jiantao and Venkat, Kartik and Han, Yanjun and Weissman, Tsachy},
    title = {Minimax estimation of functionals of discrete distributions},
    journal = {IEEE Transactions on Information Theory},
    volume = {61},
    number = {5},
    pages = {2835--2885},
    doi = {10.1109/TIT.2015.2412945},
    year = {2015}
}

@article{WY16,
    author = {Wu, Yihong and Yang, Pengkun},
    title = {Minimax rates of entropy estimation on large alphabets via best polynomial approximation},
    journal = {IEEE Transactions on Information Theory},
    volume = {62},
    number = {6},
    pages = {3702--3720},
    doi = {10.1109/TIT.2016.2548468},
    year = {2016}
}

@inproceedings{OS17,
    author = {Obremski, Maciej and Skorski, Maciej},
    title = {Renyi entropy estimation revisited},
    booktitle = {Approximation, Randomization, and Combinatorial Optimization. Algorithms and Techniques (APPROX/RANDOM 2017)},
    pages = {20:1-20:15},
    doi = {10.4230/LIPIcs.APPROX-RANDOM.2017.20},
    year = {2017}
}

@article{AOST17,
    author = {Acharya, Jayadev and Orlitsky, Alon and Suresh, Ananda Theertha and Tyagi, Himanshu},
    title = {Estimating {Renyi} entropy of discrete distributions},
    journal = {IEEE Transactions on Information Theory},
    volume = {63},
    number = {1},
    pages = {38--56},
    doi = {10.1109/TIT.2016.2620435},
    year = {2017}
}

@inproceedings{Ren61,
    author = {R{\'{e}}nyi, Alfr{\'{e}}d},
    title = {On measures of entropy and information},
    booktitle = {Proceedings of the Fourth Berkeley Symposium on Mathematics, Statistics and Probability},
    pages = {547--562},
    doi = {},
    url = {https://static.renyi.hu/renyi_cikkek/1961_on_measures_of_entropy_and_information.pdf},
    year = {1961}
}

@article{Tsa88,
    author = {Tsallis, Constantino},
    title = {Possible generalization of {Boltzmann-Gibbs} statistics},
    journal = {Journal of Statistical Physics},
    volume = {52},
    number = {},
    pages = {479--487},
    doi = {10.1007/BF01016429},
    year = {1988}
}

@article{Hil73,
    author = {Hill, M. O.},
    title = {Diversity and evenness: a unifying notation and its consequences},
    journal = {Ecology},
    volume = {54},
    number = {2},
    pages = {427--432},
    doi = {10.2307/1934352},
    year = {1973}
}

@article{AMS99,
    author = {Alon, Noga and Matias, Yossi and Szegedy, Mario},
    title = {The space complexity of approximating the frequency moments},
    journal = {Journal of Computer and System Sciences},
    volume = {58},
    number = {1},
    pages = {137--147},
    doi = {10.1006/jcss.1997.1545},
    year = {1999}
}

@book{Gin12,
    author = {Gini, Corrado},
    publisher = {Tipografia di Paolo Cuppin},
    title = {Variabilit{\`a} e Mutabilit{\`a}: contributo allo studio delle distribuzioni e delle relazioni statistiche},
    year = {1912},
    url = {https://www.byterfly.eu/islandora/object/librib:680892}
}

@book{BFOS84,
  title = {Classification and Regression Trees}, 
  author={Breiman, Leo and Friedman, Jerome and Olshen, R. A. and Stone, Charles J.},
  year = {1984},
  publisher = {Chapman and Hall/CRC},
  doi={10.1201/9781315139470}
}

@article{AK01,
    author = {Antos, Andr{\'{a}}s and Kontoyiannis, Ioannis},
    title = {Convergence properties of functional estimates for discrete distributions},
    journal = {Random Structures \& Algorithms},
    volume = {19},
    number = {3--4},
    pages = {163--193},
    doi = {10.1002/rsa.10019},
    year = {2001}
}

@inproceedings{BYKS01,
    author = {Bar-Yossef, Ziv and Kumar, Ravi and Sivakumar, D.},
    title = {Sampling algorithms: lower bounds and applications},
    booktitle = {Proceedings of the 33rd Annual ACM Symposium on Theory of Computing},
    pages = {266--275},
    doi = {10.1145/380752.380810},
    year = {2001}
}

@inproceedings{IW05,
    author = {Indyk, Piotr and Woodruff, David},
    title = {Optimal approximations of the frequency moments of data streams},
    booktitle = {Proceedings of the 37th Annual ACM Symposium on Theory of Computing},
    pages = {202--208},
    doi = {10.1145/1060590.1060621},
    year = {2005}
}

@inproceedings{BGKS06,
    author = {Bhuvanagiri, Lakshminath and Ganguly, Sumit and Kesh, Deepanjan and Saha, Chandan},
    title = {Simpler algorithm for estimating frequency moments of data streams},
    booktitle = {Proceedings of the 17th Annual ACM-SIAM Symposium on Discrete Algorithm},
    pages = {708--713},
    doi = {},
    url = {https://dl.acm.org/doi/abs/10.5555/1109557.1109634},
    year = {2006}
}

@inproceedings{GC07,
    author = {Ganguly, Sumit and Cormode, Graham},
    title = {On estimating frequency moments of data streams},
    booktitle = {Approximation, Randomization, and Combinatorial Optimization. Algorithms and Techniques (APPROX 2007, RANDOM 2007)},
    pages = {479--493},
    doi = {10.1007/978-3-540-74208-1_35},
    year = {2007}
}

@article{CJ15,
    author = {Chao, Anne and Jost, Lou},
    title = {Estimating diversity and entropy profiles via discovery rates of new species},
    journal = {Methodsin Ecology and Evolution},
    volume = {6},
    number = {8},
    pages = {873--882},
    doi = {10.1111/2041-210X.12349},
    year = {2015}
}

@article{Fur06,
    author = {Furuichi, Shigeru},
    title = {Information theoretical properties of {Tsallis} entropies},
    journal = {Journal of Mathematical Physics},
    volume = {47},
    number = {2},
    pages = {023302},
    doi = {10.1063/1.2165744},
    year = {2006}
}

@misc{Wan25,
    author = {Wang, Qisheng},
    title = {Information-theoretic lower bounds for approximating monomials via optimal quantum {Tsallis} entropy estimation},
    howpublished = {ArXiv preprint},
    eprint = {2509.03496},
    year = {2025}
}

@article{OW21,
    author = {O'Donnell, Ryan and Wright, John},
    title = {Quantum spectrum testing},
    journal = {Communications in Mathematical Physics},
    volume = {387},
    number = {1},
    pages = {1--75},
    doi = {10.1007/s00220-021-04180-1},
    year = {2021}
}

@article{AISW20,
    author = {Acharya, Jayadev and Issa, Ibrahim and Shende, Nirmal V. and Wagner, Aaron B.},
    title = {Estimating Quantum Entropy},
    journal = {IEEE Journal on Selected Areas in Information Theory},
    volume = {1},
    number = {2},
    pages = {454--468},
    doi = {10.1109/JSAIT.2020.3015235},
    year = {2020}
}

@article{BCWdW01,
    author = {Buhrman, Harry and Cleve, Richard and Watrous, John and de Wolf, Ronald},
    title = {Quantum Fingerprinting},
    journal = {Physical Review Letters},
    volume = {87},
    number = {16},
    pages = {167902},
    doi = {10.1103/PhysRevLett.87.167902},
    year = {2001}
}

@article{WZ25b,
    author = {Wang, Qisheng and Zhang, Zhicheng},
    title = {Quantum lower bounds by sample-to-query lifting},
    journal = {SIAM Journal on Computing},
    volume = {54},
    number = {5},
    pages = {1294--1334},
    doi = {10.1137/24M1638616},
    year = {2025},
}

@inproceedings{AKN98,
    author = {Aharonov, Dorit and Kitaev, Alexei and Nisan, Noam},
    title = {Quantum circuits with mixed states},
    booktitle = {Proceedings of the 13th Annual ACM Symposium on Theory of Computing},
    pages = {20--30},
    doi = {10.1145/276698.276708},
    year = {1998}
}

@misc{GP22,
    title  = {Improved quantum algorithms for fidelity estimation},
    author = {Gily{\'e}n, Andr{\'a}s and Poremba, Alexander},
    note   = {ArXiv preprint},
    eprint = {2203.15993},
    year   = {2022}
}

@inproceedings{CHW07,
    author = {Childs, Andrew M. and Harrow, Aram W. and Wocjan, Pawel},
    title = {Weak {Fourier-Schur} sampling, the hidden subgroup problem, and the quantum collision problem},
    booktitle = {Proceedings of the 24th Annual Symposium on Theoretical Aspects of Computer Science},
    pages = {598--609},
    doi = {10.1007/978-3-540-70918-3_51},
    year = {2007}
}

@inproceedings{OW17,
    author = {O'Donnell, Ryan and Wright, John},
    title = {Efficient quantum tomography {II}},
    booktitle = {Proceedings of the 49th Annual ACM Symposium on Theory of Computing},
    pages = {962--974},
    doi = {10.1145/3055399.3055454},
    year = {2017}
}

@article{BHH11,
    author = {Bravyi, Sergey and Harrow, Aram W. and Hassidim, Avinatan},
    title = {Quantum algorithms for testing properties of distributions},
    journal = {IEEE Transactions on Information Theory},
    volume = {57},
    number = {6},
    pages = {3971--3981},
    doi = {10.1109/TIT.2011.2134250},
    year = {2011}
}

@inproceedings{GL20,
    author = {Gily{\'e}n, Andr{\'a}s and Li, Tongyang},
    title = {Distributional property testing in a quantum world},
    booktitle = {Proceedings of the 11th Innovations in Theoretical Computer Science Conference},
    pages = {25:1--25:19},
    doi = {10.4230/LIPIcs.ITCS.2020.25},
    year = {2020},
}

@inproceedings{CFMdW10,
    author = {Chakraborty, Sourav and Fischer, Eldar and Matsliah, Arie and {\dutchPrefix{Wolf}{d}}e Wolf, Ronald},
    title = {New results on quantum property testing},
    booktitle  = {Proceedings of the 30th IARCS Annual Conference on Foundations of Software Technology and Theoretical Computer Science},
    pages = {145--156},
    year = {2010},
    doi = {10.4230/LIPICS.FSTTCS.2010.145},
}

@article{Mon15,
    author = {Montanaro, Ashley},
    journal = {Proceedings of the Royal Society A},
    title = {Quantum speedup of {Monte} {Carlo} methods},
    volume = {471},
    number = {2181},
    pages = {20150301},
    doi = {10.1098/rspa.2015.0301},
    year = {2015}
}

@article{LWL24,
  title = {Succinct quantum testers for closeness and $k$-wise uniformity of probability distributions},
  author = {Luo, Jingquan and Wang, Qisheng and Li, Lvzhou},
  journal = {IEEE Transactions on Information Theory},
  volume = {70},
  number = {7},
  pages = {5092--5103},
  year = {2024},
  doi = {10.1109/TIT.2024.3393756},
}

@article{LW19,
    author = {Li, Tongyang and Wu, Xiaodi},
    title = {Quantum query complexity of entropy estimation},
    journal = {IEEE Transactions on Information Theory},
    volume = {65},
    number = {5},
    pages = {2899--2921},
    doi = {10.1109/TIT.2018.2883306},
    year = {2019}
}

@article{WZL24,
    author = {Wang, Xinzhao and Zhang, Shengyu and Li, Tongyang},
    title = {A quantum algorithm framework for discrete probability distributions with applications to {R\'{e}nyi} entropy estimation},
    journal = {IEEE Transactions on Information Theory},
    volume = {70},
    number = {5},
    pages = {3399--3426},
    doi = {10.1109/TIT.2024.3382037},
    year = {2024}
}

@misc{SJ25,
    title  = {Near optimal quantum algorithm for estimating {Shannon} entropy},
    author = {Shin, Myeongjin and Jeong, Kabgyun},
    note   = {ArXiv preprint},
    eprint = {2509.07452},
    year   = {2025}
}

@article{BDKR05,
    author = {Batu, Tugkan and Dasgupta, Sanjoy and Kumar, Ravi and Rubinfeld, Ronitt},
    title = {The complexity of approximating the entropy},
    journal = {SIAM Journal on Computing},
    volume = {35},
    number = {1},
    pages = {132--150},
    doi = {10.1137/S0097539702403645},
    year = {2005}
}

@inproceedings{CKO25,
    author = {Canonne, Cl{\'e}ment L. and Kothari, Robin and O'Donnell, Ryan},
    title = {Uniformity testing when you have the source code},
    booktitle = {Proceedings of the 20th Conference on the Theory of Quantum Computation, Communication and Cryptography},
    pages = {7:1-7:20},
    year = {2025},
    doi = {10.4230/LIPIcs.TQC.2025.7},
}

@article{BFR+13,
    author = {Batu, Tu{\u{g}}kan and Fortnow, Lance and Rubinfeld, Ronitt and Smith, Warren D. and White, Patrick},
    title = {Testing closeness of discrete distributions},
    journal = {Journal of the ACM},
    volume = {60},
    number = {1},
    pages = {4:1--4:25},
    doi = {10.1145/2432622.2432626},
    year = {2013}
}

@inproceedings{CDVV14,
    author = {Chan, Siu-On and Diakonikolas, Ilias and Valiant, Paul and Valiant, Gregory},
    title = {Optimal algorithms for testing closeness of discrete distributions},
    booktitle = {Proceedings of the 2014 Annual ACM-SIAM Symposium on Discrete Algorithms},
    pages = {1193--1203},
    year = {2014},
    doi = {10.1137/1.9781611973402.88},
}

@article{JHW18,
    author = {Jiao, Jiantao and Han, Yanjun and Weissman, Tsachy},
    title = {Minimax estimation of the {$L_1$} distance},
    journal = {IEEE Transactions on Information Theory},
    volume = {64},
    number = {10},
    pages = {6672--6706},
    doi = {10.1109/TIT.2018.2846245},
    year = {2018}
}

@inproceedings{DK16,
    author = {Diakonikolas, Ilias and Kane, Daniel M.},
    title = {A new approach for testing properties of discrete distributions},
    booktitle = {Proceedings of the 57th IEEE Annual Symposium on Foundations of Computer Science},
    pages = {685--694},
    year = {2016},
    doi = {10.1109/FOCS.2016.78},
}

@inproceedings{BFF+01,
    author = {Batu, Tu{\u{g}}kan and Fischer, Eldar and Fortnow, Lance and Kumar, Ravi and Rubinfeld, Ronitt and White, Patrick},
    title = {Testing random variables for independence and identity},
    booktitle = {Proceedings of the 42nd IEEE Symposium on Foundations of Computer Science},
    pages = {442--451},
    year = {2001},
    doi = {10.1109/SFCS.2001.959920},
}

@article{Pan08,
    author = {Paninski, Liam},
    title = {A coincidence-based test for uniformity given very sparsely sampled discrete data},
    journal = {IEEE Transactions on Information Theory},
    volume = {54},
    number = {10},
    pages = {4750--4755},
    doi = {10.1109/TIT.2008.928987},
    year = {2008}
}

@article{Hel09,
    author = {Hellinger, E.},
    title = {Neue begr{\"u}ndung der Theorie quadratischer formen von unendlichvielen ver{\"a}nderlichen},
    journal = {Journal f{\"u}r die reine und angewandte Mathematik},
    volume = {1909},
    number = {136},
    pages = {210--271},
    doi = {10.1515/crll.1909.136.210},
    year = {1909}
}

@article{LMR14,
	author = {Lloyd, Seth and Mohseni, Masoud and Rebentrost, Patrick},
	journal = {Nature Physics},
	title = {Quantum principal component analysis},
	volume = {10},
	number = {9},
	pages = {631-633},
	doi = {10.1038/nphys3029},
	year = {2014}
}

@article{KLL+17,
    author = {Kimmel, Shelby and Lin, Cedric Yen-Yu and Low, Guang Hao and Ozols, Maris and Yoder, Theodore J.},
    title = {Hamiltonian simulation with optimal sample complexity},
    journal = {npj Quantum Information},
    volume = {3},
    number = {1},
    pages = {1--7},
    doi = {10.1038/s41534-017-0013-7},
    year = {2017}
}

@article{GKP+25,
	author = {Go, Byeongseon and Kwon, Hyukjoon and Park, Siheon and Patel, Dhrumil and Wilde, Mark M.},
	title = {Sample-based {Hamiltonian} and {Lindbladian} simulation: Non-asymptotic analysis of sample complexity},
	journal = {Quantum Science and Technology},
    doi = {10.1088/2058-9565/ae075b},
	year = {2025},
}

@article{WZ24,
  title = {Fast quantum algorithms for trace distance estimation},
  author = {Wang, Qisheng and Zhang, Zhicheng},
  journal = {IEEE Transactions on Information Theory},
  volume = {70},
  number = {4},
  pages = {2720--2733},
  year = {2024},
  doi = {10.1109/TIT.2023.3321121},
}

@misc{Bau11,
    author = {Baumgartner, Bernhard},
    title = {An inequality for the trace of matrix products, using absolute values},
    howpublished = {ArXiv e-prints},
    eprint = {1106.6189},
    year = {2011}
}

@book{NC10,
    author = {Nielsen, Michael A. and Chuang, Isaac L.},
    title = {Quantum Computation and Quantum Information},
    publisher = {Cambridge University Press},
    doi = {10.1017/CBO9780511976667},
    year = {2010}
}

@misc{BMW16,
    author = {Bavarian, Mohammad and Mehraban, Saeed and Wright, John},
    title = {Learning entropy},
    howpublished = {A manuscript on von Neumann entropy estimation, private communication},
    year = {2016}
}

@article{SH21,
    author = {Subramanian, Sathyawageeswar and Hsieh, Min-Hsiu},
    title = {Quantum algorithm for estimating $\alpha$-Renyi entropies of quantum states},
    journal = {Physical Review A},
    volume = {104},
    number = {2},
    pages = {022428},
    doi = {10.1103/PhysRevA.104.022428},
    year = {2021}
}

@misc{GHS21,
    author = {Gur, Tom and Hsieh, Min-Hsiu and Subramanian, Sathyawageeswar},
    title = {Sublinear quantum algorithms for estimating von {Neumann} entropy},
    howpublished = {ArXiv e-prints},
    eprint = {2111.11139},
    year = {2021}
}

@misc{CLW20,
    author = {Chowdhury, Anirban N. and Low, Guang Hao and Wiebe, Nathan},
    title = {A variational quantum algorithm for preparing quantum {Gibbs} states},
    howpublished = {ArXiv preprints},
    eprint = {2002.00055},
    year = {2020}
}

@article{HJW20,
  title = {Minimax estimation of divergences between discrete distributions},
  author = {Han, Yanjun and Jiao, Jiantao and Weissman, Tsachy},
  journal = {IEEE Journal on Selected Areas in Information Theory},
  volume = {1},
  number = {3},
  pages = {814--823},
  year = {2020},
  doi = {10.1109/JSAIT.2020.3041036},
}

@article{BZLV18,
  title = {Estimation of {KL} divergence: Optimal minimax rate},
  author = {Bu, Yuheng and Zou, Shaofeng and Liang, Yingbin and Veeravalli, Venugopal V.},
  journal = {IEEE Transactions on Information Theory},
  volume = {64},
  number = {4},
  pages = {2648--2674},
  year = {2018},
  doi = {10.1109/TIT.2018.2805844},
}

@article{WZW23,
  title = {Quantum algorithms for estimating quantum entropies},
  author = {Wang, Youle and Zhao, Benchi and Wang, Xin},
  journal = {Physical Review Applied},
  volume = {19},
  number = {4},
  pages = {044041},
  year = {2023},
  doi = {10.1103/PhysRevApplied.19.044041},
}

@article{WZC+23,
    author = {Wang, Qisheng and Zhang, Zhicheng and Chen, Kean and Guan, Ji and Fang, Wang and Liu, Junyi and Ying, Mingsheng},
    title = {Quantum Algorithm for Fidelity Estimation},
    journal = {IEEE Transactions on Information Theory},
    volume = {69},
    number = {1},
    pages = {273--282},
    doi = {10.1109/TIT.2022.3203985},
    year = {2023}
}

@article{WZYW23,
    author = {Wang, Youle and Zhang, Lei and Yu, Zhan and Wang, Xin},
    title = {Quantum phase processing and its applications in estimating phase and entropies},
    journal = {Physical Review A},
    volume = {108},
    number = {6},
    pages = {062413},
    doi = {10.1103/PhysRevA.108.062413},
    year = {2023}
}

@article{WGL+24,
    author = {Wang, Qisheng and Guan, Ji and Liu, Junyi and Zhang, Zhicheng and Ying, Mingsheng},
    title = {New Quantum Algorithms for Computing Quantum Entropies and Distances},
    journal = {IEEE Transactions on Information Theory},
    volume = {70},
    number = {8},
    pages = {5653--5680},
    doi = {10.1109/TIT.2024.3399014},
    year = {2024}
}

@misc{UNWT25,
    title={Quantum algorithms for {Uhlmann} transformation},
    author={Utsumi, Takeru and Nakata, Yoshifumi and Wang, Qisheng and Takagi, Ryuji},
    howpublished={ArXiv preprints},
    eprint={2509.03619},
    year={2025}
}

@article{Hay25,
    title = {Measuring quantum relative entropy with finite-size effect},
    author = {Hayashi, Masahito},
    journal = {Quantum},
    volume = {9},
    number = {},
    pages = {1725},
    year = {2025},
    doi = {10.22331/q-2025-05-05-1725}
}

@inproceedings{LW25b,
    author = {Liu, Yupan and Wang, Qisheng},
    title = {On estimating the quantum $\ell_\alpha$ distance},
    booktitle = {Proceedings of the 33rd Annual European Symposium on Algorithms},
    pages = {106:1-106:19},
    doi = {10.4230/LIPIcs.ESA.2025.106},
    year = {2025}
}

@article{Wan24,
    author = {Wang, Qisheng},
    title = {Optimal trace distance and fidelity estimations for pure quantum states},
    journal = {IEEE Transactions on Information Theory},
    volume = {70},
    number = {12},
    pages = {8791--8805},
    doi = {10.1109/TIT.2024.3447915},
    year = {2024}
}

@inproceedings{WZ24b,
    author = {Wang, Qisheng and Zhang, Zhicheng},
    title = {Sample-optimal quantum estimators for pure-state trace distance and fidelity via samplizer},
    booktitle = {Proceedings of the 53rd International Colloquium on Automata, Languages, and Programming},
    pages = {154:1-154:21},
    doi = {10.4230/LIPIcs.ICALP.2026.154},
    year = {2026}
}

@misc{BGW25,
    author = {Bao, Jinge and Gao, Minbo and Wang, Qisheng},
    title = {On estimating the quantum {Tsallis} relative entropy},
    eprint = {2510.00752},
    howpublished = {ArXiv preprints},
    year = {2025}
}

@inproceedings{BOW19,
    author = {B{\u{a}}descu, Costin and O'Donnell, Ryan and Wright, John},
    title = {Quantum state certification},
    booktitle = {Proceedings of the 51st Annual ACM SIGACT Symposium on Theory of Computing},
    pages = {503--514},
    doi = {10.1145/3313276.3316344},
    year = {2019}
}

@inproceedings{OW26,
    author = {O’Donnell, Ryan and Wadhwa, Chirag},
    title = {Instance-optimal quantum state certitfication with entangled measurements},
    booktitle = {Proceedings of the 58th ACM Symposium on Theory of Computing},
    pages = {398--409},
    doi = {10.1145/3798129.3800759},
    year = {2026}
}

@article{vN27,
author = {von Neumann, J.},
journal = {Nachrichten von der Gesellschaft der Wissenschaften zu Göttingen, Mathematisch-Physikalische Klasse},
pages = {273-291},
title = {Thermodynamik quantenmechanischer Gesamtheiten},
url = {http://eudml.org/doc/59231},
volume = {1927},
year = {1927},
}

@article{GL95,
    author = {Gill, Richard D. and Levit, Boris Y.},
    title = {Applications of the van {Trees} inequality: A {Bayesian} {Cram{\'{e}}r-Rao} bound},
    journal = {Bernoulli},
    volume = {1},
    number = {1--2},
    pages = {59--79},
    doi = {10.2307/3318681},
    year = {1995}
}

@incollection{MdW16,
    author = {Montanaro, Ashley and de Wolf, Ronald},
    title = {A survey of quantum property testing},
    year = {2016},
    publisher = {University of Chicago},
    booktitle = {Theory of Computing Library},
    series = {Graduate Surveys},
    number = {7},
    doi = {10.4086/toc.gs.2016.007},
    pages = {1--81},
}

@incollection{DdW11,
    author = {Drucker, Andrew and de Wolf, Ronald},
    title = {Quantum proofs for classical theorems},
    year = {2011},
    publisher = {University of Chicago},
    booktitle = {Theory of Computing Library},
    series = {Graduate Surveys},
    number = {2},
    doi = {10.4086/toc.gs.2011.002},
    pages = {1--54},
}

@article{TD16,
    author = {Toranzo, I. V. and Dehesa, J. S.},
    title = {{R{\'{e}}nyi}, {Shannon} and {Tsallis} entropies of {Rydberg} hydrogenic systems},
    journal = {Europhysics Letters},
    volume = {113},
    number = {4},
    pages = {48003},
    doi = {10.1209/0295-5075/113/48003},
    year = {2016}
}

@article{TPCD16,
    author = {Toranzo, I. V. and Puertas-Centeno, D. and Dehesa, J. S.},
    title = {Entropic properties of {$D$}-dimensional {Rydberg} systems},
    journal = {Physica A},
    volume = {462},
    number = {},
    pages = {1197--1206},
    doi = {10.1016/j.physa.2016.06.144},
    year = {2016}
}

@article{QKW24,
  title={Multivariate trace estimation in constant quantum depth},
  author={Quek, Yihui and Kaur, Eneet and Wilde, Mark M.},
  journal={Quantum},
  volume={8},
  pages={1220},
  year={2024},
  doi = {10.22331/Q-2024-01-10-1220}
}

@article{WD20,
  title={Calculating {R{\'{e}}nyi} entropies with neural autoregressive quantum states},
  author={Wang, Zhaoyou and Davis, Emily J.},
  journal={Physical Review A},
  volume={102},
  number={6},
  pages={062413},
  year={2020},
  doi = {10.1103/PhysRevA.102.062413}
}

\appendix

\section{Proof of \texorpdfstring{\cref{lemma:tr-L-rho-k}}{Lemma \ref{lemma:tr-L-rho-k}}} \label{sec:proof-of-lemma-tr-L-rho-k}

For readability, we use $\begin{matrix}
    A \\ B
\end{matrix}$ to denote $A \otimes B$. 
Then, \cref{eq:shift-L} can be restated as follows.
\begin{equation} \label{eq:restate-shift-L}
\tr\rbra*{ \begin{matrix}
    L_{\mathsf{A}_1} \\ I_{\mathsf{A}_2} \\ \vdots \\ I_{\mathsf{A}_k}
\end{matrix} \cdot \textup{Shift}_k \cdot \begin{matrix}
    \rho_{\mathsf{A}_1} \\ \rho_{\mathsf{A}_2} \\ \vdots \\ \rho_{\mathsf{A}_k}
\end{matrix} } = \tr\rbra*{L \rho^k}.
\end{equation}

The proof of \cref{lemma:tr-L-rho-k} is given as follows.

\begin{proof}[Proof of \cref{lemma:tr-L-rho-k}]
By the definition of trace, we expand the left hand side of \cref{eq:restate-shift-L} as follows.
\begin{equation} \label{eq:expand-tr}
    \tr\rbra*{ \begin{matrix}
    L_{\mathsf{A}_1} \\ I_{\mathsf{A}_2} \\ \vdots \\ I_{\mathsf{A}_k}
\end{matrix} \cdot \textup{Shift}_k \cdot \begin{matrix}
    \rho_{\mathsf{A}_1} \\ \rho_{\mathsf{A}_2} \\ \vdots \\ \rho_{\mathsf{A}_k}
\end{matrix} }
    =
    \sum_{i_1, i_2, \dots, i_k} \begin{matrix}
    \bra{i_1}_{\mathsf{A}_1} \\ \bra{i_2}_{\mathsf{A}_2} \\ \vdots \\ \bra{i_k}_{\mathsf{A}_k}
\end{matrix} \cdot \begin{matrix}
    L_{\mathsf{A}_1} \\ I_{\mathsf{A}_2} \\ \vdots \\ I_{\mathsf{A}_k}
\end{matrix} \cdot \textup{Shift}_k \cdot \begin{matrix}
    \rho_{\mathsf{A}_1} \\ \rho_{\mathsf{A}_2} \\ \vdots \\ \rho_{\mathsf{A}_k}
\end{matrix} \cdot \begin{matrix}
    \ket{i_1}_{\mathsf{A}_1} \\ \ket{i_2}_{\mathsf{A}_2} \\ \vdots \\ \ket{i_k}_{\mathsf{A}_k}
\end{matrix},
\end{equation}
where $i_1, i_2, \dots, i_k$ range over the computational basis of $\mathsf{A}_1, \mathsf{A}_2, \dots, \mathsf{A}_k$, respectively. 
Note that $\textup{Shift}_k$ can be described in the following form:
\begin{equation} \label{eq:def-shift-k}
    \textup{Shift}_k = \sum_{j_1, j_2, \dots, j_k} \ketbra{j_2, \dots, j_k, j_1}{j_1, j_2, \dots, j_k} = \sum_{j_1, j_2, \dots, j_k} \begin{matrix}
        \ketbra{j_2}{j_1}_{\mathsf{A}_1} \\ \ketbra{j_3}{j_2}_{\mathsf{A}_2} \\ \vdots \\ \ketbra{j_1}{j_k}_{\mathsf{A}_k}
    \end{matrix},
\end{equation}
where $j_1, j_2, \dots, j_k$ range over the computational basis of $\mathsf{A}_1, \mathsf{A}_2, \dots, \mathsf{A}_k$, respectively. 
Taking \cref{eq:def-shift-k} into \cref{eq:expand-tr}, we have
\begin{align} 
    \eqref{eq:expand-tr}
    = \sum_{\substack{i_1, i_2, \dots, i_k \\ j_1, j_2, \dots, j_k}} \begin{matrix}
    \bra{i_1}_{\mathsf{A}_1} \\ \bra{i_2}_{\mathsf{A}_2} \\ \vdots \\ \bra{i_k}_{\mathsf{A}_k}
\end{matrix} \cdot \begin{matrix}
    L_{\mathsf{A}_1} \\ I_{\mathsf{A}_2} \\ \vdots \\ I_{\mathsf{A}_k}
\end{matrix} \cdot \begin{matrix}
        \ketbra{j_2}{j_1}_{\mathsf{A}_1} \\ \ketbra{j_3}{j_2}_{\mathsf{A}_2} \\ \vdots \\ \ketbra{j_1}{j_k}_{\mathsf{A}_k}
    \end{matrix} \cdot \begin{matrix}
    \rho_{\mathsf{A}_1} \\ \rho_{\mathsf{A}_2} \\ \vdots \\ \rho_{\mathsf{A}_k}
\end{matrix} \cdot \begin{matrix}
    \ket{i_1}_{\mathsf{A}_1} \\ \ket{i_2}_{\mathsf{A}_2} \\ \vdots \\ \ket{i_k}_{\mathsf{A}_k}
\end{matrix} 
    = \sum_{\substack{i_1, i_2, \dots, i_k \\ j_1, j_2, \dots, j_k}} \begin{matrix}
    \bra{i_1} L \ket{j_2}_{\mathsf{A}_1} \\ \braket{i_2}{j_3}_{\mathsf{A}_2} \\ \vdots \\ \braket{i_k}{j_1}_{\mathsf{A}_k}
\end{matrix} \cdot \begin{matrix}
    \bra{j_1}\rho\ket{i_1}_{\mathsf{A}_1} \\ \bra{j_2}\rho\ket{i_2}_{\mathsf{A}_2} \\ \vdots \\ \bra{j_k}\rho\ket{i_k}_{\mathsf{A}_k}
\end{matrix}. \label{eq:full-expand}
\end{align}
In \cref{eq:full-expand}, we only have to consider the non-zero terms in the summation, in which case $i_2 = j_3$, $i_3 = j_4$, \dots, $i_k = j_1$. 
Then, 
\begin{align}
    \eqref{eq:full-expand}
    & = \sum_{\substack{i_1, i_2, \dots, i_k \\ j_2}} \bra{i_1} L \ket{j_2} \cdot \bra{i_k}\rho\ket{i_1} \cdot \rbra[\Big]{ \bra{j_2}\rho\ket{i_2} \cdot \bra{i_2}\rho\ket{i_3} \cdot \dots \cdot \bra{i_{k-1}}\rho\ket{i_k} } \\
    & = \sum_{\substack{i_1, i_2, \dots, i_k \\ j_2}} \bra{i_1} L \ket{j_2} \cdot \rbra[\Big]{ \bra{j_2}\rho\ket{i_2} \cdot \bra{i_2}\rho\ket{i_3} \cdot \dots \cdot \bra{i_{k-1}}\rho\ket{i_k} } \cdot \bra{i_k}\rho \ket{i_1} \\
    & = \sum_{\substack{i_1, i_2, \dots, i_k \\ j_2}} \bra{i_1} L \cdot \ketbra{j_2}{j_2} \cdot \rho \cdot \ketbra{i_2}{i_2} \cdot \rho \cdot \ketbra{i_3}{i_3} \cdot \dots \cdot \ketbra{i_{k-1}}{i_{k-1}} \cdot \rho \cdot \ketbra{i_k}{i_k} \cdot \rho \cdot \ket{i_1} \\
    & = \sum_{i_1} \bra{i_1} L \cdot \sum_{j_2}\ketbra{j_2}{j_2} \cdot \rho \cdot \underbrace{\sum_{i_2}\ketbra{i_2}{i_2} \cdot \rho \cdot \sum_{i_3}\ketbra{i_3}{i_3} \cdot \dots \cdot \sum_{i_{k-1}} \ketbra{i_{k-1}}{i_{k-1}} \cdot \rho \cdot \sum_{i_k}\ketbra{i_k}{i_k} \cdot \rho}_{k-1} \cdot \ket{i_1} \\
    & = \sum_{i_1} \bra{i_1} L \cdot \underbrace{I \cdot \rho \cdot I \cdot \rho \cdot I \cdots \cdot I \cdot \rho \cdot I \cdot \rho}_{k} {} \cdot \ket{i_1} \\
    & = \sum_{i_1} \bra{i_1} L \rho^{k} \ket{i_1} \\
    & = \tr\rbra*{L\rho^k}. 
\end{align}

\end{proof}

\end{document}